\documentclass[nocopyrightspace]{sig-alternate}
\usepackage{amsmath}
\usepackage{amssymb}
\usepackage{mathpartir}
\usepackage{stmaryrd}
\usepackage[usenames]{color}
\usepackage{enumitem}
\usepackage{xspace}
\usepackage{times}
\usepackage{paralist}
\usepackage{comment}
\usepackage{url,hyperref}
\includecomment{tronly}
\excludecomment{subonly}

\usepackage{tikz}
\usetikzlibrary{matrix,arrows,decorations.pathmorphing}

\newenvironment{round}
  {\begin{tikzpicture}[baseline={([yshift={-\ht\strutbox}]current bounding box.west)},
                    outer sep=0pt,inner sep=0pt]}
  {\draw [rounded corners=.5em] (table.north west) rectangle (table.south east);
   \end{tikzpicture}}

\newtheorem{theorem}{Theorem}
\newtheorem{proposition}[theorem]{Proposition}
\newtheorem{corollary}[theorem]{Corollary}
\newtheorem{lemma}[theorem]{Lemma}

\newcommand{\proofheader}{}

\newcommand{\sepskip}{\smallskip}
\newcommand{\james}[1]{}
\newcommand{\sam}[1]{}

\newcommand{\xs}{\mathit{xs}}
\newcommand{\ys}{\mathit{ys}}

\newcommand{\delim}{\text{\textvisiblespace}}
\newcommand{\consl}[1]{{#1}\_}

\newcommand{\isPoor}{\mathit{isPoor}}
\newcommand{\isRich}{\mathit{isRich}}
\newcommand{\outliers}{\mathit{outliers}}
\newcommand{\clients}{\mathit{clients}}

\newcommand{\getTasks}{\mathit{getTasks}}
\newcommand{\tasksOfEmp}{\mathit{tasksOfEmp}}
\newcommand{\employeesOfDept}{\mathit{employeesOfDept}}
\newcommand{\contactsOfDept}{\mathit{contactsOfDept}}
\newcommand{\employeesByTask}{\mathit{employeesByTask}}

\newcommand{\set}[1]{\{{#1}\}}

\newcommand{\xtoy}[2]{_{{#1}=1}^{#2}}
\newcommand{\iton}{\xtoy{i}{n}}
\newcommand{\jtom}{\xtoy{j}{m}}

\newcommand{\sem}[1]{\llbracket{#1}\rrbracket}

\newcommand{\tsem}[1]{\sem{#1}}

\newcommand{\nsem}[1]{\mathcal{N}\sem{#1}}

\newcommand{\ssemname}{\mathcal{S}}

\newcommand{\ssem}[1]{\ssemname\sem{#1}}
\newcommand{\ssemcomp}[1]{\ssemname\sem{#1}}
\newcommand{\sflatsem}[1]{\ssemname^\flat\sem{#1}}

\newcommand{\isem}[1]{\mathcal{I}\sem{#1}}

\newcommand{\iflatsem}[1]{\mathcal{I}^\flat\sem{#1}}

\newcommand{\inatsem}[1]{\mathcal{I}^\natural\sem{#1}}

\newcommand{\asem}[1]{\mathcal{A}\sem{#1}}
\newcommand{\asemcomp}[1]{\mathcal{A}\sem{#1}}

\newcommand{\bsem}[1]{\ssemname\sem{#1}}

\newcommand{\psem}[1]{\mathcal{H}\sem{#1}}

\newcommand{\pasem}[1]{\psem{#1}}

\newcommand{\lsem}[1]{\mathcal{L}\sem{#1}}
\newcommand{\lsemcomp}[1]{\mathcal{L}\sem{#1}}

\newcommand{\shouterl}{\mathopen{\llceil}}
\newcommand{\shouterr}{\mathclose{\rrceil}}
\newcommand{\shouter}[1]{\shouterl{#1}\shouterr}
\newcommand{\shinner}[1]{\mathopen{\llfloor}{#1}\mathclose{\rrfloor}}
\newcommand{\shlistl}{\shouterl}
\newcommand{\shlistr}{\shouterr^\star}
\newcommand{\shlist}[1]{\shlistl{#1}\shlistr}

\newcommand{\desc}[1]{\mathopen{\sslash}{#1}\mathclose{\bbslash}}
\newcommand{\desclist}[1]{\mathopen{\sslash}{#1}\mathclose{\bbslash}}

\newcommand{\indexes}[2]{\mathit{indexes}_{#1}(#2)}

\newcommand{\ora}[1]{\overrightarrow{#1}}

\newcommand{\recordl}{\langle}
\newcommand{\recordr}{\rangle}
\newcommand{\record}[1]{\recordl{#1}\recordr}

\newcommand{\tuplel}{\recordl}
\newcommand{\tupler}{\recordr}
\newcommand{\tuple}[1]{\record{#1}}

\newcommand{\unit}{\tuple{}}

\newcommand{\resvl}{[}
\newcommand{\resvr}{]}
\newcommand{\resv}[1]{\resvl{#1}\resvr}

\newcommand{\bagvl}{}
\newcommand{\bagvr}{}
\newcommand{\bagv}[1]{\bagvl{#1}\bagvr}

\newcommand{\bagtl}{\mathrm{Bag}\,}
\newcommand{\bagtr}{}
\newcommand{\bagt}[1]{\bagtl{#1}\bagtr}
\newcommand{\singleton}[1]{\ret{#1}}
\newcommand{\emptybag}{\emptyset}

\newcommand{\isEmpty}[1]{\Empty\,{#1}}

\newcommand{\Int}{\mathit{Int}}
\newcommand{\Bool}{\mathit{Bool}}
\newcommand{\String}{\mathit{String}}

\newcommand{\Index}{\mathit{Index}}

\newcommand{\lam}[1]{\lambda{#1}.\,}

\newcommand{\True}{\mathsf{true}}
\newcommand{\False}{\mathsf{false}}

\newcommand{\union}{\mathbin{\uplus}}
\newcommand{\bigunion}{\biguplus}

\newcommand{\For}{\mathsf{for}}

\newcommand{\Ret}{\mathsf{return}}

\newcommand{\Empty}{\mathsf{empty}}

\newcommand{\retl}{\Ret\,(}
\newcommand{\retr}{)}
\newcommand{\ret}[1]{\Ret\,{#1}}
\newcommand{\retlann}[1]{\Ret^{\cann{#1}}\,(}
\newcommand{\retrann}[1]{\retr}

\newcommand{\retrannnp}[1]{}
\newcommand{\retann}[2]{\Ret^{\cann{#2}}\,{#1}}

\newcommand{\If}{\mathsf{if}}
\newcommand{\Then}{\mathsf{then}}
\newcommand{\Else}{\mathsf{else}}

\newcommand{\Table}{\mathsf{table}}

\newcommand{\Where}{\mathsf{where}}

\newcommand{\Unionall}{\mathsf{union~all}}
\newcommand{\Select}{\mathsf{select}}
\newcommand{\From}{\mathsf{from}}
\newcommand{\As}{\mathsf{as}}

\newcommand{\Let}{\mathsf{let}}
\newcommand{\In}{\mathsf{in}}

\newcommand{\With}{\mathsf{with}}
\newcommand{\Over}{\mathsf{over}}

\newcommand{\Ind}[2]{{#1}\,\mathord{\diamond}\,{#2}}
\newcommand{\DUnit}{1}
\newcommand{\Unit}{\Ind{\top}{\DUnit}}

\newcommand{\iup}{\mathsf{out}}  
\newcommand{\idown}{\mathsf{in}} 

\newcommand{\cindex}{\mathsf{index}}         

\newcommand{\ann}[1]{{\color{red}@{#1}}}

\newcommand{\cann}[1]{#1} 

\newcommand{\cce}{\mathord{::=}}

\newcommand{\rewriteto}{\leadsto}
\newcommand{\nrewriteto}{\not\leadsto}

\newcommand{\nil}{\mathord{[\,]}}
\newcommand{\cons}{\mathbin{::}}
\newcommand{\append}{\mathbin{+\!\!+}}

\newcommand{\ba}{\begin{array}}
\newcommand{\ea}{\end{array}}

\newcommand{\bl}{\ba[t]{@{}l@{}}}
\newcommand{\el}{\ea}

\newenvironment{syntax}{\[\ba{l@{\quad}r@{~}c@{~}l}}{\ea\]}
\newenvironment{equations}{\[\small\ba{r@{~}c@{~}l}}{\ea\]}
\newenvironment{iequations}{\small\ba{r@{~}c@{~}l}}{\ea\normalsize}
\newenvironment{derivation}{~\\ \begin{minipage}{\fill}\[\ba{cl}}{\ea\]\end{minipage}\\~\\}

\newcommand{\byy}[1]{\quad({#1})}

\newcommand{\pmap}{\mathit{pmap}}

\newcommand{\filter}{\mathit{filter}}
\newcommand{\zip}{\mathit{zip}}
\newcommand{\unzip}{\mathit{unzip}}
\newcommand{\Concat}{\mathit{concat}}

\newcommand{\enum}{\mathit{enum}}
\newcommand{\key}{\mathit{index}}

\newcommand{\cl}{\ell}

\newcommand{\iv}{w}

\newcommand{\lins}{\mathbf{L}}

\newcommand{\rownum}{\mathsf{row\_number}}
\newcommand{\orderby}{\mathsf{order}\,\mathsf{by}}
\newcommand{\rownumindex}[1]{\rownum()\,\Over\,(\orderby\,{#1})}

\newcommand{\maxr}{\mathit{max}}   
\newcommand{\size}{\mathit{size}}  

\newcommand{\langname}{\ensuremath{\lambda_{\mathit{NRC}}}\xspace}

\newenvironment{Cases}{
\begin{description}[leftmargin=0cm]
}{\end{description}}

\newcommand{\dom}{\mathit{dom}}
\newcommand{\paths}{\mathit{paths}}

\newcommand{\pbag}{\mathord{\downarrow}}

\newcommand{\buildname}{\mathit{stitch}}
\newcommand{\buildtop}[1]{\buildname({#1})}
\newcommand{\build}[2]{\buildname_{#1}({#2})}

\newcommand{\norm}{\mathit{norm}}

\newcommand{\lab}[1]{\textrm{\color{blue} #1}}

\newcommand{\Task}{\mathit{Task}}
\newcommand{\Contact}{\mathit{Contact}}
\newcommand{\Employee}{\mathit{Employee}}
\newcommand{\Department}{\mathit{Department}}
\newcommand{\Organisation}{\mathit{Organisation}}

\newcommand{\Result}{\mathit{Result}}

\newcommand{\lId}{\lab{id}}
\newcommand{\lName}{\lab{name}}
\newcommand{\lSalary}{\lab{salary}}
\newcommand{\lTasks}{\lab{tasks}}

\newcommand{\lEmployees}{\lab{employees}}
\newcommand{\lContacts}{\lab{contacts}}

\newcommand{\lClient}{\lab{client}}

\newcommand{\lDepartment}{\lab{department}}

\newcommand{\lPeople}{\lab{people}}

\newcommand{\tTasks}{\mathit{tasks}}
\newcommand{\tEmployees}{\mathit{employees}}
\newcommand{\tDepartments}{\mathit{departments}}
\newcommand{\tOrganisation}{\mathit{organisation}}

\newcommand{\tContacts}{\mathit{contacts}}

\newcommand{\emp}{\lab{employee}}
\newcommand{\tsk}{\lab{task}}
\newcommand{\dpt}{\lab{dept}}

\newcommand{\cStr}[1]{``\mathrm{#1}"}

\newcommand{\sterm}{\Lambda_S}
\newcommand{\stype}{\mathcal{T}_S}

\newcommand{\Qcomp}{Q_{\mathit{comp}}}
\newcommand{\Qstitch}{Q_{\mathit{stitch}}}
\newcommand{\Qorg}{Q_{\mathit{org}}}
\newcommand{\Qoutliers}{Q}

\newcommand{\nf}{\mathit{nf}}

\newcommand{\figorganisation}{
\begin{figure}[tb]

\[\small
\begin{array}{l}
\sem{departments}=\smallskip\\
\begin{round}
\node (table) [inner sep=0pt] {
\begin{tabular}{r|c}
(\lId) & \lName \\
\hline
1 & Product \\
2 & Quality \\
3 & Research \\
4 & Sales \\
\end{tabular}
};
\end{round}
\end{array}
\quad
\begin{array}{l}
\sem{employees}=\smallskip\\
\begin{round}
\node (table) [inner sep=0pt] {
\begin{tabular}{r|c|c|c}
(\lId) & \dpt & \lName & \lSalary \\ 
\hline
 1 & Product  & Alex & 20000 \\   
 2 & Product  & Bert & 900 \\     
 3 & Research & Cora & 50000 \\   
 4 & Research & Drew & 60000 \\   
 5 & Sales    & Erik & 2000000 \\ 
 6 & Sales    & Fred & 700 \\     
 7 & Sales    & Gina & 100000 \\  
\end{tabular}
};
\end{round}
\end{array}
\]

\vspace{-.5cm}

\[\small
\begin{array}{l}
\sem{\tTasks} =\smallskip\\
\begin{round}
\node (table) [inner sep=0pt] {
\begin{tabular}{r|c|c}
(\lId) & \emp & \tsk \\
\hline
 1 & Alex & build \\
 2 & Bert & build \\
 3 & Cora & abstract \\
 4 & Cora & build \\
 5 & Cora & call \\
 6 & Cora & dissemble \\
 7 & Cora & enthuse \\
 8 & Drew & abstract \\
 9 & Drew & enthuse \\
10 & Erik & call \\
11 & Erik & enthuse \\
12 & Fred & call \\
13 & Gina & call \\
14 & Gina & dissemble \\
\end{tabular}
};
\end{round} 
\end{array}
\!\!\!\!
\ba{l}
\sem{\tContacts} =\smallskip\\
\begin{round}
\node (table) [inner sep=0pt] {
\begin{tabular}{r|c|c|c}
(\lId) & \dpt & \lName & \lClient \\
\hline
 1 & Product  & Pam & $\False$ \\
 2 & Product  & Pat & $\True$ \\
 3 & Research & Rob & $\False$ \\
 4 & Research & Roy & $\False$ \\
 5 & Sales    & Sam & $\False$ \\
 6 & Sales    & Sid & $\False$ \\
 7 & Sales    & Sue & $\True$ \\ 
\end{tabular}
};
\end{round}
\ea
\]

\caption{Sample data}\label{fig:sampledata}
\end{figure}
}

\title{Query shredding: Efficient relational evaluation of queries over nested
  multisets (extended version)}
 \numberofauthors{3}

\author{
\alignauthor
James Cheney \\
 \affaddr{University of Edinburgh} \\
 \email{jcheney@inf.ed.ac.uk}
\and
\alignauthor
Sam Lindley \\
 \affaddr{University of Edinburgh} \\
 \email{Sam.Lindley@ed.ac.uk}
\and
\alignauthor
Philip Wadler \\
 \affaddr{University of Edinburgh}
 \email{wadler@inf.ed.ac.uk}
}

\toappear{This paper is an extended version of a paper to appear at SIGMOD 2014.}

\begin{document}

\maketitle

\begin{abstract}
  Nested relational query languages have been explored extensively,
  and underlie industrial language-integrated query systems such as
  Microsoft's LINQ.  However, relational databases do not natively
  support nested collections in query results.  This can lead to major
  performance problems: if programmers write queries that yield nested
  results, then such systems typically either fail or generate a large
  number of queries.  We present a new approach to query shredding,
  which converts a query returning nested data to a fixed number of
  SQL queries.  Our approach, in
  contrast to prior work, handles multiset semantics, and generates an
  idiomatic SQL:1999 query directly from a normal form for nested
  queries.  We provide a detailed description of our translation and
  present experiments showing that it offers comparable or better
  performance than a recent alternative approach on a range of
  examples.
\end{abstract}

\section{Introduction}

Databases are one of the most important applications of declarative
programming techniques.  However, relational databases only support
queries against flat tables, while programming languages typically
provide complex data structures that allow arbitrary combinations of
types including nesting of collections (e.g. sets of sets).
Motivated by this so-called \emph{impedance mismatch}, and inspired by
insights into language design based on monadic
comprehensions~\cite{Wadler92}, database researchers introduced nested
relational query
languages~\cite{schek86is,buneman94sigmod,buneman95tcs} as a
generalisation of flat relational queries to allow nesting collection
types inside records or other types.  Several recent language designs,
such as XQuery~\cite{xquery3} and PigLatin~\cite{piglatin}, have
further extended these ideas, and they have been particularly
influential on language-integrated querying systems such as
Kleisli~\cite{wong00jfp}, Microsoft's LINQ in C\# and
F\#~\cite{meijer06sigmod,fsharp3,cheney13icfp},
Links~\cite{cooper06fmco,lindley12tldi}, and Ferry~\cite{grust09sigmod}.

This paper considers the problem of translating nested queries over
nested data to flat queries over a flat representation of nested data,
or \emph{query shredding} for short. Our motivation is to support a
free combination of the features of nested relational query languages
with those of high-level programming languages, particularly systems
such as Links, Ferry, and LINQ. All three of these systems support
queries over nested data structures (e.g.  records containing nested
sets, multisets/bags, or lists) in principle; however, only Ferry supports them
in practice. Links and LINQ currently either reject such queries at
run-time or execute them inefficiently in-memory by loading
unnecessarily large amounts of data or issuing large numbers of
queries (sometimes called \emph{query storms} or
\emph{avalanches}~\cite{grust10pvldb} or the \emph{$N+1$ query
  problem}). To construct nested data structures while avoiding this
performance penalty, programmers must currently write flat queries
(e.g. loading in a superset of the needed source data) and convert the
results to nested data structures.  Manually reconstructing nested query
results is tricky and hard to maintain; it may also mask optimisation opportunities.

In the Ferry system, Grust et al.~\cite{grust09sigmod,grust10pvldb}
have implemented shredding for nested list queries by adapting an
XQuery-to-SQL translation called
\emph{loop-lifting}~\cite{grust04vldb}.  Loop-lifting produces queries
that make heavy use of advanced On-Line Analytic Processing (OLAP)
features of SQL:1999, such as \verb|ROW_NUMBER| and \verb|DENSE_RANK|,
and to optimise these queries Ferry relies on a SQL:1999 query
optimiser called Pathfinder~\cite{pathfinder}.

\begin{sloppypar}
Van den Bussche~\cite{vandenbussche01tcs} proved expressiveness
results showing that it is possible in principle to evaluate nested
queries over \emph{sets} via multiple flat queries. 
To strengthen the result, Van den Bussche's simulation eschews value invention
mechanisms such as SQL:1999's \verb|ROW_NUMBER|.  The downside,
however, is that the flat queries can produce results that are
quadratically larger than needed to represent sets and may not preserve bag
semantics.
\begin{tronly}
  In particular, for bag semantics the size of a union of two nested
  sets can be quadratic in the size of the inputs --- we give a
  concrete example in Appendix~\ref{app:vdb-incorrect}.
\end{tronly}
\end{sloppypar}

Query shredding is related to the well-studied \emph{query unnesting}
problem~\cite{kim82tods,fegaras00tods}.  However, most prior work on
unnesting only considers SQL queries that contain subqueries in WHERE
clauses, not queries returning nested results; the main exception is
Fegaras and Maier's work on query unnesting in a complex object
calculus~\cite{fegaras00tods}.

\begin{figure*}[tb!]
\centering
\begin{tabular}{ccc}
(a) & (b) & (c)\\
  \includegraphics[scale=0.17]{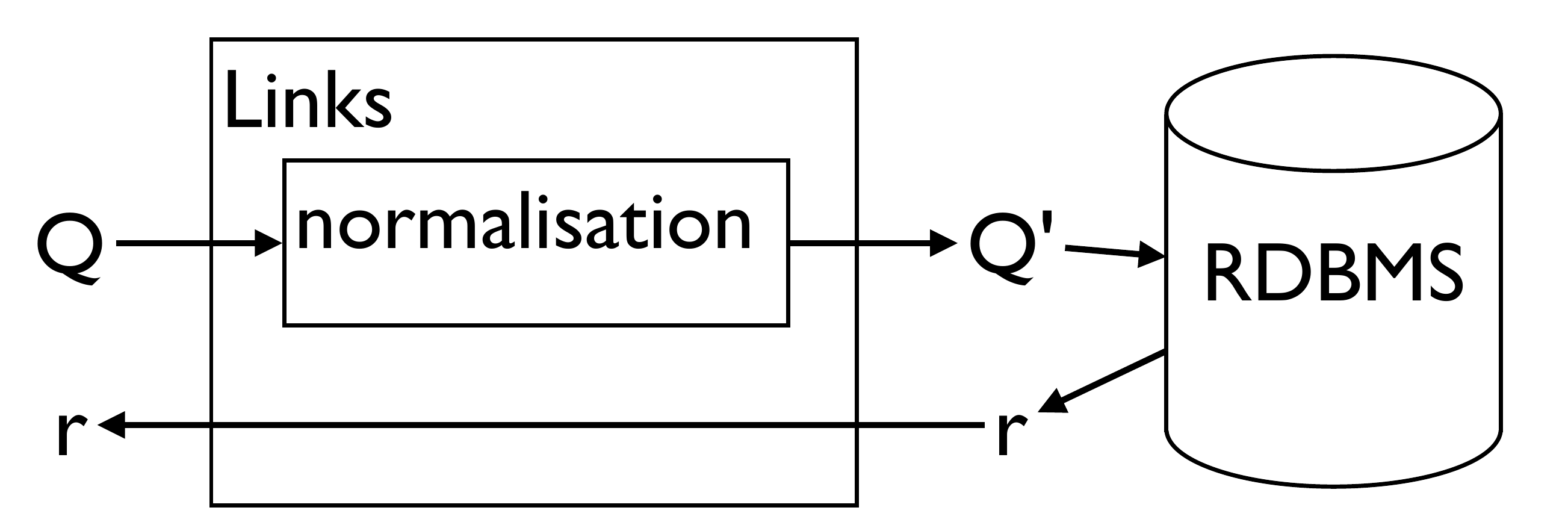}
&~~
  \includegraphics[scale=0.17]{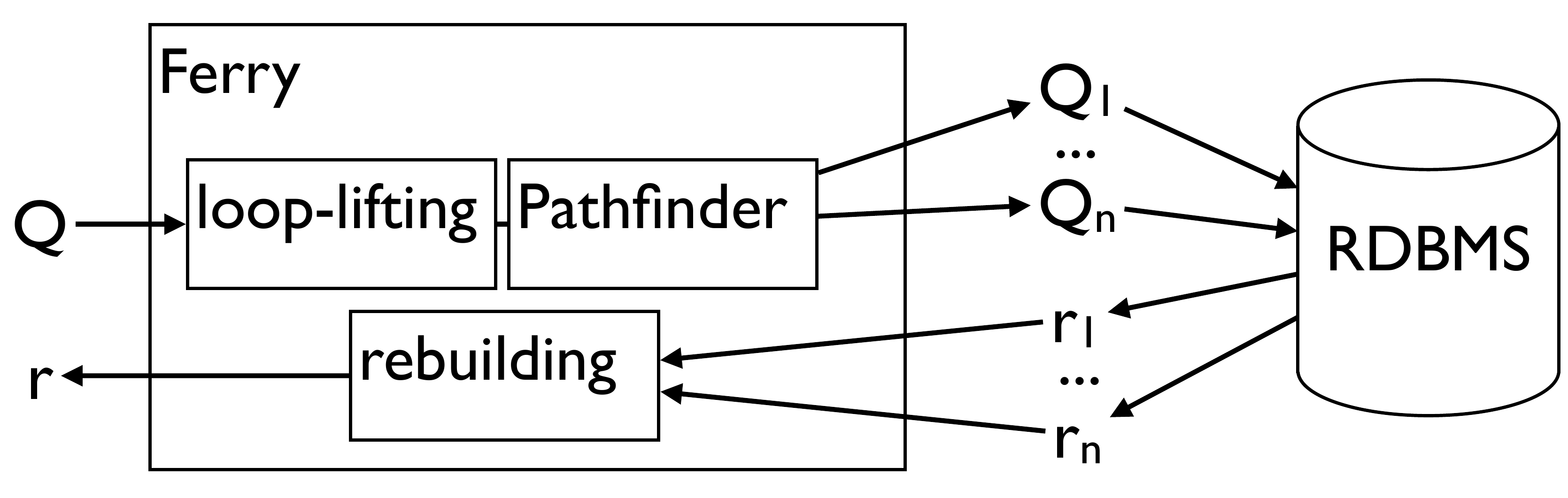}
&~~
  \includegraphics[scale=0.17]{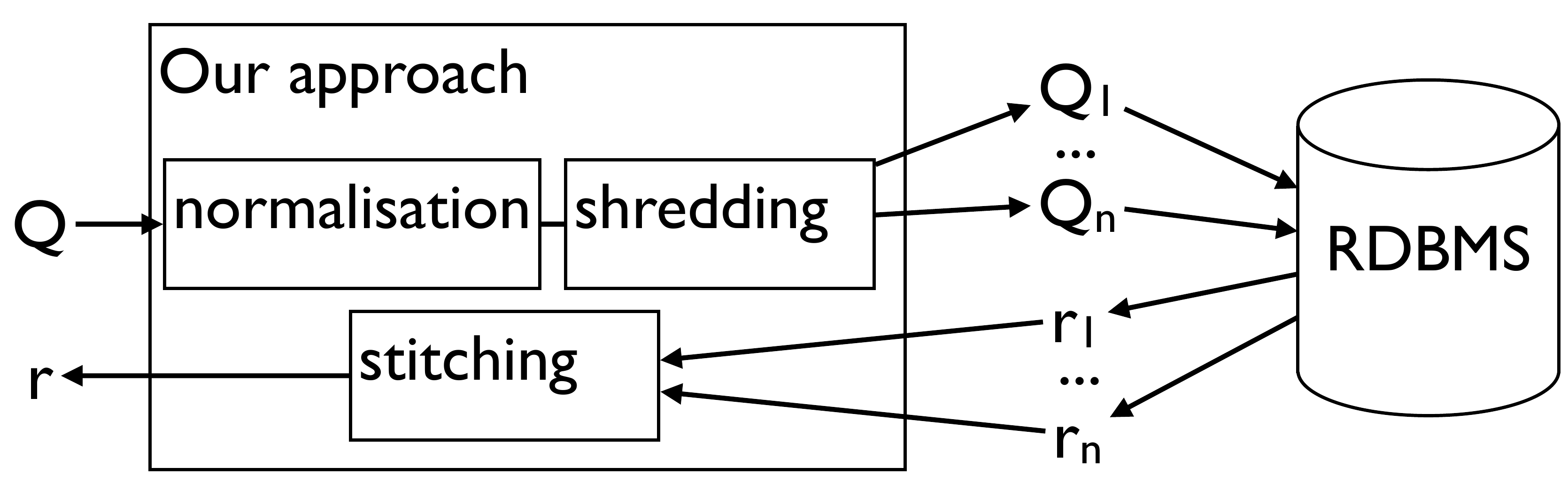}
\end{tabular}
 \caption{(a) Default Links behaviour (flat queries) ~~ (b) Ferry
   (nested list queries) ~~ (c) Our approach (nested bag queries)}
  \label{fig:diagrams}
\end{figure*}

In this paper, we introduce a new approach to query shredding 
for nested multiset queries (a case not handled by prior work).
Our work is
formulated in terms of the Links language, but should be applicable to
other language-integrated query systems, such as Ferry and LINQ, or to
other complex-object query languages~\cite{fegaras00tods}.  Figure~\ref{fig:diagrams}
illustrates the behaviour of Links, Ferry and our approach.

We decompose the translation from nested to flat queries into a series
of simpler translations.
We leverage prior work on \emph{normalisation} that translates a
higher-order query to a \emph{normal form} in which higher-order
features have been
eliminated~\cite{wong96jcss,cooper09dbpl,fegaras00tods}.
Our algorithm operates on normal forms. We review normalisation in Section~\ref{sec:background}, and give a running
example of our approach in Section~\ref{sec:example}.

Sections~\ref{sec:shredding}--\ref{sec:to-sql} present the remaining
phases of our approach, which are new.  The \emph{shredding} phase
translates a single, nested query to a number of flat queries. These
queries are organised in a \emph{shredded package}, which is
essentially a type expression whose collection type constructors are
annotated with queries. The different queries are linked by
\emph{indexes}, that is, keys and foreign keys.  The shredding translation is
presented in Section~\ref{sec:shredding}.  Shredding leverages the
normalisation phase in that we can define translations on types and
terms independently.  Section~\ref{sec:run-shredded} shows how to run
shredded queries and stitch the results back together to form nested
values.

The \emph{let-insertion} phase implements a flat indexing
scheme 
using a let-binding construct 
and a row-numbering operation.
Let-insertion (described in Section~\ref{sec:indexing-schemes}) is
conceptually straightforward, but provides a vital link to proper SQL
by providing an implementation of abstract indexes. The final stage is
to translate to SQL (Section~\ref{sec:to-sql}) by flattening records,
translating let-binding to SQL's \verb|WITH|, and translating
row-numbering to SQL's \verb|ROW_NUMBER|.
\begin{tronly}
  This phase also flattens nested records; the (standard) details are
  included in Appendix~\ref{app:record-flattening} for completeness.
\end{tronly}

We have implemented and experimentally evaluated our approach
(Section~\ref{sec:implementation}) in comparison with Ulrich's
implementation of loop-lifting in Links~\cite{ulrich11ferry-links}.
Our approach typically performs as well or better than loop-lifting,
and can be significantly faster.  

Our contribution over prior work can be summarised as follows.
Fegaras and Maier~\cite{fegaras00tods} show how to unnest complex
object queries including lists, sets, and bags, but target a
nonstandard nested relational algebra,
whereas we target standard SQL.  Van den Bussche's
simulation~\cite{vandenbussche01tcs} translates nested set queries to
several relational queries but was used to prove a theoretical result
and was not intended as a practical implementation technique, nor has
it been implemented and evaluated.
Ferry~\cite{grust09sigmod,grust10pvldb} translates nested list queries
to several SQL:1999 queries and then tries to simplify the resulting
queries using Pathfinder.  Sometimes, however, this produces queries
with cross-products inside OLAP operations such as \verb|ROW_NUMBER|,
which Pathfinder cannot simplify.  In contrast, we delay introducing
OLAP operations until the last stage, and our experiments show how
this leads to much better performance on some queries.  
Finally, we handle nested multisets, not sets or lists.

\begin{subonly}
  Additional details, proofs of correctness, and comparison with
  related work are available in a companion technical
  report~\cite{shredding-tr}.
\end{subonly}
\begin{tronly}
  This report is the extended version of a SIGMOD 2014 paper, and
  contains additional details, proofs of correctness, and comparison
  with related work in appendices.
\end{tronly}

\section{Background}
\label{sec:background}

We use metavariables $x,y,\ldots,f,g$ for \emph{variables}, and
$c,d,\ldots$ for \emph{constants} and primitive operations.  We also
use letters $t,t',\ldots$ for \emph{table names}, $\cl,\cl', \cl_i,
\ldots$ for \emph{record labels} and $a,b,\ldots$ for \emph{tags}.

We write $M[x:=N]$ for capture-avoiding substitution of $N$ for $x$ in
$M$. We write $\vec{x}$ for a vector $x_1,\ldots,x_n$. Moreover, we
extend vector notation pointwise to other constructs, writing, for
example, $\record{\ora{\cl = M}}$ for $\recordl{\cl_1=M_1}, \ldots,$
${\cl_n=M_n}\recordr$.

\newcommand{\init}{\mathit{init}}
\newcommand{\last}{\mathit{last}}

We write: square brackets $[-]$ for the meta level list constructor;
$w \cons \vec{v}$ for adding the element $w$ onto the front of the
list $\vec{v}$; $\vec{v} \append \vec{w}$ for the result of appending
the list $\vec{w}$ onto the end of the list $\vec{v}$; and $\Concat$
for the function that concatenates a list of lists.  We also make use of the following functions:
\[\bl
\init~[x_i]\iton = [x_i]_{i=1}^{n-1} \smallskip \qquad \last~[x_i]\iton = x_n \\
\enum([v_1, \dots, v_m]) = [\tuple{1, v_1}, \dots \tuple{m, v_m}]
\el
\]%

In the meta language we make extensive use of comprehensions,
primarily list comprehensions. For
instance, $[v \mid x \gets \xs, y \gets \ys, p]$, returns a copy of
$v$ for each pair $\tuple{x,y}$ of elements of $\xs$ and $\ys$ such
that the predicate $p$ holds.
We write $[v_i]\iton$ as shorthand for $[v_i \mid 1 \leq i \leq n]$ and
similarly, e.g., $\record{\cl_i=M_i}\iton$ for $\record{\cl_1=M_1,
  \dots, \cl_n=M_n}$.

\begin{figure}[tb!]
  \centering\small
  \begin{equations}
    \nsem{x}_\rho &=& \rho(x)\\
    \nsem{c(X_1, \dots, X_n)}_\rho &=& \sem{c}(\nsem{X_1}_\rho, \dots, \nsem{X_n}_\rho) \\
    \nsem{\lambda x. M}_\rho &=& \lambda v. \nsem{M}_{\rho[x\mapsto v]}\\
    \nsem{M~N}_\rho &=& \nsem{M}_\rho(\nsem{N}_\rho)\\
    \nsem{\record{\cl_i=M_i}\iton}_\rho &=& \record{\cl_i=\nsem{M_i}_\rho}\iton \\
    \nsem{M.\cl}_\rho &=& \nsem{M}_\rho.\cl \\
    \nsem{\If\,L\,\Then\,M\,\Else\,N}_\rho &=& \left\{
      \begin{array}{ll}
        \nsem{M}_\rho, & \text{if }\nsem{L}_\rho = \True\\
        \nsem{N}_\rho, & \text{if }\nsem{L}_\rho = \False
      \end{array}
    \right.\\
    \nsem{\singleton{M}}_\rho &=& [\nsem{M}_\rho] \\
    \nsem{\emptybag}_\rho &=& \nil \\
    \nsem{M \union N}_\rho &=& \nsem{M}_\rho \append \nsem{N}_\rho \\
    \nsem{\For\,(x \gets M)\, N}_\rho &=&
      \Concat  [\nsem{N}_{\rho[x\mapsto v]} \mid v \gets \nsem{M}_\rho] \\
    \nsem{\isEmpty{M}}_\rho &=&  \left\{
      \begin{array}{ll}
        \True,  & \text{if }\nsem{M}_\rho = \nil \\
        \False, & \text{if }\nsem{M}_\rho \neq \nil \\
      \end{array}
    \right.\\
    \nsem{\Table~t}_\rho &=& \sem{t} \\[1ex]
  \end{equations}
  
  \caption{Semantics of \langname}
\label{fig:bag-semantics}
\end{figure}

\subsection{Nested relational calculus}

We take the higher-order, nested relational calculus (evaluated over
bags) as our starting point.  We call this \langname; this is
essentially a core language for the query components of Links, Ferry,
and LINQ.
The types of $\langname$ include base types (integers, strings,
booleans), record types $\record{\ora{\cl:A}}$, bag types
$\bagt{A}$, and function types $A \to B$.
\begin{syntax}
\textrm{Types}      & A, B & \cce& O \mid \record{\ora{\cl:A}} \mid
\bagt{A} \mid  A \to B \\ 
\textrm{Base types} & O    & \cce& \Int \mid \Bool \mid \String 
\end{syntax}%
We say that a type is \emph{nested} if it contains no function types
and \emph{flat} if it contains only base and record types.

The terms of \langname include $\lambda$-abstractions, applications,
and the standard terms of nested relational calculus.
\begin{syntax}
\textrm{Terms}
& M,N    &  \cce & x \mid c(\vec{M}) \mid \Table~t \mid \If\,M\, \Then\,N\, \Else\, N' \\
         && \mid & \lam{x}M \mid M ~ N \mid \record{\ora{\cl=M}} \mid M.\cl   \mid  \isEmpty{M} \\
        && \mid & \singleton{M} \mid \emptybag \mid M \union N  \mid \For\,(x \gets M)\,N 
\end{syntax}
\begin{sloppypar}
We assume that the constants and primitive functions include boolean
values with negation and conjunction, and integer values with standard
arithmetic operations and equality tests.
We assume special labels $\#_1,\#_2,\ldots$ and encode tuple types $\tuple{A_1,\dots,A_n}$ as record types
$\record{\#_1:A_1, \dots, \#_n:A_n}$, and similarly tuple terms
$\tuple{M_1,\dots,M_n}$ as record terms $\record{\#_1=M_1, \dots,
  \#_n=M_n}$.
%
We assume fixed signatures $\Sigma(t)$ and $\Sigma(c)$ for tables and
constants. The tables are constrained to have flat relation type
($\bagt{\record{\cl_1:O_1,\dots,\cl_n:O_n}}$), and the constants
must be of base type or first order $n$-ary functions ($\tuple{O_1, \dots, O_n} \to O$).
\end{sloppypar}

Most language constructs are standard. The $\emptyset$ expression
builds an empty bag, $\singleton{M}$ constructs a singleton, and $M
\union N$ builds the bag union of two collections.  The $\For\,(x
\gets M)\;N$ comprehension construct iterates over a bag obtained by
evaluating $M$, binds $x$ to each element, evaluates $N$ to another
bag for each such binding, and takes the union of the results.  The
expression $\isEmpty{M}$ evaluates to $\True$ if $M$ evaluates to an
empty bag, $\False$ otherwise.

\langname employs a standard type system similar to that presented in
other work~\cite{wong00jfp,lindley12tldi,cheney13icfp}.  We will also
employ several typed intermediate languages and translations
mapping $\langname$ to SQL.  All of these (straightforward) type
systems are omitted due to space limits; they will be available in the
full version of this paper.

\paragraph*{Semantics}  
We give a denotational semantics in terms of lists. Though we
wish to preserve bag semantics, we interpret object-level bags as
meta-level lists. For meta-level values $v$ and $v'$,
we consider $v$ and $v'$ equivalent as multisets if they are equal up to
permutation of list elements.  We use lists mainly so that we can talk about
order-sensitive operations such as $\rownum$.

We interpret base types as integers, booleans and strings, function
types as functions, record types
as records, and bag types as lists.
For each table $t \in \dom(\Sigma)$, we assume a fixed interpretation
$\tsem{t}$ of $t$ as a list of records of type $\Sigma(t)$.
In SQL, tables do not have a list semantics by default, but we can
impose one by choosing a canonical ordering for the rows of the
table. We order by all
of the columns arranged in lexicographic order (assuming linear orderings
on field names and base types).

We assume fixed interpretations $\sem{c}$ for the constants and
primitive operations. The semantics of nested relational calculus are
shown in Figure~\ref{fig:bag-semantics}.  
\begin{tronly}
The (standard) typing rules  are given in Appendix~\ref{app:typing-rules}.
\end{tronly}
We let $\rho$ range over environments mapping variables to values,
writing $\varepsilon$ for the empty environment and $\rho[x \mapsto
v]$ for the extension of $\rho$ with $x$ bound to $v$.

\subsection{Query normalisation}
\label{subsec:normalisation}

In Links, query normalisation is an important part of the execution
model~\cite{cooper09dbpl,lindley12tldi}.  Links currently supports
only queries mapping flat tables to flat results, or \emph{flat--flat}
queries.  When a subexpression
denoting a query is evaluated, the subexpression is first
normalised and then converted to SQL, which is sent to the database
for evaluation; the tuples received in response are then converted
into a Links value and normal execution resumes (see Figure\ref{fig:diagrams}(a)).

For \emph{flat--nested} queries that read from flat tables and produce
a nested result value, our normalisation procedure is similar to the
one currently used in Links~\cite{lindley12tldi}, but we hoist all conditionals into the
nearest enclosing comprehension as $\Where$ clauses. This is a minor
change; the modified algorithm is given in the full version of this paper.
The resulting normal
forms are:
\begin{syntax}
  \text{Query terms} & L &\cce& \bigunion\, \vec{C}
  \sepskip\\
  \text{Comprehensions}   & C    &\cce& \For\,(\vec{G}\,\Where\,X)\,\ret{M} \\
  \text{Generators}       & G    &\cce& x \gets t \\
  \text{Normalised terms} & M, N &\cce& X \mid R \mid L \\
  \text{Record terms}     & R    &\cce& \record{\ora{\cl=M}} \\
  \text{Base terms}       & X    &\cce& x.\cl \mid c(\vec{X}) \mid \isEmpty{L} \\
\end{syntax}%
Any closed
flat--nested query can be converted to an equivalent term in the
above normal form.
\begin{theorem}
Given a closed flat--nested query $\vdash M :\bagt{ A}$, there exists a
normalisation function $\norm_{\bagt{A}}$,  mapping each $M$ to an equivalent
normal form $\norm_{\bagt{A}}(M)$.
\end{theorem}

The normalisation algorithm and correctness proof are similar to those
in previous papers~\cite{cooper09dbpl,lindley12tldi,cheney13icfp}. 
\begin{tronly}
  The details are given in Appendix~\ref{app:sn}.
\end{tronly}
The normal forms above can also be viewed as an SQL-like language
allowing relation-valued attributes (similar to the complex-object
calculus of Fegaras and Maier~\cite{fegaras00tods}).  Thus, our
results can also be used to support nested query results in an
SQL-like query language.  In this paper, however, we focus on the
functional core language based on comprehensions, as illustrated in the next section.


\section{Running example}
\label{sec:example}

To motivate and illustrate our work, we present an extended example
showing how our shredding translation could be used to provide useful
functionality to programmers working in LINQ using F\#, Ferry or Links. We first
describe the code the programmer would actually write and the results
the system produces. Throughout the rest of the paper, we return to
this example to illustrate how the shredding translation works.

\figorganisation

Consider a flat database schema $\Sigma$ for an organisation:
\[\small
\bl
\tTasks(\emp:\String, \tsk:\String)\\
\tEmployees(\dpt:\String,\lName:\String, 
\lSalary:\Int)\\
\tContacts( \dpt:\String,  \lName:\String, \lClient:\Bool)\\
\tDepartments(\lName:\String) 
\ea
\]%
Each department has a name, a collection of employees,
and a collection of external contacts. Each employee has a name,
salary and a collection of tasks. Some contacts are clients.  
Figure~\ref{fig:sampledata} shows a small instance of this schema.
For convenience, we also assume every table has an integer-valued key $id$.

In \langname, queries of the form $\For \ldots \Where \ldots
\ret{\ldots}$ are natively supported.  These are \emph{comprehensions}
as found in XQuery or functional programming languages, and they
generalise idiomatic SQL \verb|SELECT| \verb|FROM| \verb|WHERE|
queries~\cite{buneman94sigmod}.  Unlike SQL, we can use (nonrecursive)
functions to define queries with parameters, or parts of queries, and
freely combine them.  For example, the following functions define
useful query patterns over the above schema:
\[\small\bl
\tasksOfEmp~e = \bl 
\For\, (t \gets \tTasks)\\ 
\Where\, (t.\emp = e.\lName)\\
\ret{t.tsk} 
\el
\smallskip\\
\contactsOfDept~d = \bl
\For\, (c \gets \tContacts)\\
\Where\, (d.\dpt = c.\dpt)\\
\ret{\record{\lName=c.\lName,\lClient=c.\lClient}}
\el
\smallskip\\
\employeesByTask~t = \\
\qquad
\bl
\For\,(e\gets \tEmployees,d \gets \tDepartments)\\
\Where\, (e.\lName = t.\emp \wedge e.\dpt = d.\dpt)\\
  \ret{\record{b=e.\emp, c=d.\dpt}}
\el
\el
\]

Nested queries allow free mixing of collection (bag) types with record
or base types. For example, the following query
\[\small\bl
 \employeesOfDept~d = \bl
\For\, (e \gets \tEmployees)\\
\Where\, (d.\dpt = e.\dpt)\\
\ret{\recordl\bl 
\lName=e.\lName,\lSalary=e.\lSalary,\\
\lTasks=\tasksOfEmp~e\recordr
\el
}
\el
\el
\]
returns a nested result: a collection of employees in a department,
each with an associated collection of tasks.  That is, its return type
is $\bagtl \recordl\lName{:}\String, \lSalary{:} \Int,\lTasks{:}\bagt{ \String}\recordr\bagtr$

Consider the following nested schema for organisations:
\begin{equations}
\Task         &=& \String \\
\Employee     &=& \record{\lName:\String, \lSalary:\Int ,\lTasks:\bagt{\Task} } \\
\Contact      &=& \record{\lName:\String, \lClient:\Bool} \\
\Department   &=& \recordl\bl \lName:\String,\lEmployees:\bagt{\Employee},\\
\lContacts:\bagt{\Contact}\recordr
\el \\
\Organisation &=& \bagt{\Department} \\
\end{equations}%
Using some of the above functions we can write a query
$\Qorg$ that maps data in the flat schema $\Sigma$ to the
nested type $\Organisation$, as follows:
\[\small
\Qorg = \bl
  \For\,(d \gets \tDepartments) \\
  \quad
  \retl\recordl
    \bl
      \lName      = d.\lName, \\
      \lEmployees = \employeesOfDept~d,\\
      \lContacts = \contactsOfDept~d\recordr
      \el
\el
\]

We can also define and use higher-order functions to build queries,
such as the following:
\[\small
\bl
\filter~p~\xs = \For\,(x \gets \xs)\,\Where\,(p(x))\,\ret{x}\\
any~\xs~p = \neg(\Empty(\For\, (x \gets xs)\, \Where\, (p(x))\, \ret{\record{}}))\\
all~\xs~p = \neg(any~\xs~(\lambda x. \neg(p(x))))\\
contains~\xs~u = any~\xs~(\lambda x. x = u)
\el
\]

To illustrate the main technical challenges of shredding,
we consider a query with two levels of nesting and a union
operation. 

Suppose we wish to find for each department a collection of
people of interest, both employees and contacts, along with a list of
the tasks they perform. Specifically, we are interested in those
employees that earn less than a thousand euros and those who earn more
than a million euros, call them \emph{outliers}, along with those
contacts who are clients.  The following code defines poor, rich,
outliers, and clients:
\[\small
\bl
\isPoor~x = x.\lSalary < 1000 \\
\isRich~x = x.\lSalary > 1000000 \\
\outliers~\xs = \filter~(\lambda x.\isRich~x \vee \isPoor~x)~\xs \\
\clients~\xs = \filter~(\lambda x.\,x.\lClient)~xs \\
\el
\]
We also introduce a convenient higher-order function that uses its $f$
parameter to initialise the tasks of its elements:
\[\small\bl
\getTasks~\xs~f = \For\,(x \gets
\xs)\,\ret{\record{\lName=x.\lName,\lTasks=f~x}}
\el
\]

Using the above operations, the query $\Qoutliers$ returns each department, the
outliers and clients associated with that department, and their
tasks. We assign the special task $\cStr{buy}$ to clients.
\[
\small\ba{l}
\Qoutliers(organisation) = \\
\qquad\bl
\For\, (x \gets organisation) \\
\retl \recordl
     \bl
       \lDepartment = x.\lName, \\
       \lPeople = \\
       \quad\getTasks(\outliers(x.\lEmployees))~(\lambda y.\,y.\lTasks)   \\
       \union~
           \getTasks(\clients(x.\lContacts))~(\lambda y.\,\ret{\cStr{buy}})\recordr \retr\\
     \el
\el 
\ea
\]%
The result type of $\Qoutliers$ is:
\[\small
\bl
\Result = \bagtl\recordl
   \bl
     \lDepartment:\String, \\
     \lPeople:\bagtl\recordl
       \bl
         \lName:\String, \lTasks:\bagt{\String}\recordr\bagtr\recordr\bagtr \\
       \el
   \el
\el
\]
We can compose $\Qoutliers$ with $\Qorg$ to form a
query $\Qoutliers( \Qorg)$ from the flat data stored in $\Sigma$ to
the nested $\Result$.  The normal
form of this composed query, which we call $\Qcomp$, is as follows:
\begin{equations}
\Qcomp &=& \bl
  \For\,(x \gets \tDepartments) \\
  \retl\\
  ~~\recordl
     \bl
       \lDepartment = x.\lName, \\
       \lPeople = \\
        \phantom{\union}~(\bl
            \For\,(y \gets \tEmployees)\,
            \Where\,(x.\lName = y.\dpt ~\wedge \\
            \hfill(y.\lSalary < 1000 \vee y.\lSalary > 1000000)) \\
            \retl\recordl
              \bl
                \lName  = y.\lName, \\
                \lTasks =
                   \bl
                     \For\,(z \gets \tTasks) \\
                     \Where\,(z.\emp = y.\lName) \\
                     \ret{z.\tsk}
                       \recordr\retr) \\
                   \el \\
              \el \\
           \el \\
        \union
        ~(\bl
            \For\,(y \gets \tContacts)\\
            \Where\,(x.\lName = y.\dpt \wedge y.\lClient) \\
            \retl\recordl
              \bl
                \lName  = y.\lName, \\
                \lTasks = \ret{\cStr{buy}}
                   \recordr\retr)\recordr\retr \\
              \el \\
           \el \\
     \el \\
\el
\end{equations}

The result of running $\Qcomp$
on the data in Figure~\ref{fig:sampledata} is:
\[\small
\bl
  \resvl
    \recordl\bl
      \lDepartment = \cStr{Product}, \\
      \lPeople  = 
        \resvl
        \bl
          \record{\lName=\cStr{Bert},
                  \lTasks=\resv{\cStr{build}}}, \\
          \record{\lName=\cStr{Pat}, 
                  \lTasks=\resv{\cStr{buy}}}\resvr\recordr\resvr \\
        \el \\
    \el \\
    \recordl\bl
      \lDepartment   = \cStr{Research}, 
      \lPeople = \emptybag\recordr, \\
    \el \\
    \recordl\bl
      \lDepartment   = \cStr{Quality}, 
      \lPeople = \emptybag\recordr, \\
    \el \\
    \recordl\bl
      \lDepartment   = \cStr{Sales}, \\
      \lPeople = 
        \resvl\bl
          \record{\lName=\cStr{Erik}, \lTasks=\resv{\cStr{call}, \cStr{enthuse}}}, \\
          \record{\lName=\cStr{Fred}, \lTasks=\resv{\cStr{call}}}, \\
          \record{\lName=\cStr{Sue}, 
                  \lTasks=\resv{\cStr{buy}}}\resvr\recordr\resvr 
        \el 
    \el 
\el
\]

Now, however, we are faced with a problem: SQL databases do not
directly support nested multisets (or sets).  Our shredding translation,
like Van den Bussche's simulation for sets~\cite{vandenbussche01tcs}
and Grust et al.'s for lists~\cite{grust10pvldb},
can translate a normalised
query such as $\Qcomp : \Result$ that maps flat input
$\Sigma$ to nested output $\Result$ to a fixed number of flat
queries
$q_1 : \Result_1, \ldots,q_n:
\Result_n$ %
whose results can be combined via a \emph{stitching} operation
$\Qstitch : \Result_1 \times
\cdots \times \Result_n \to \Result$.
Thus, we can simulate the query $\Qcomp$ by running $q_1,\ldots,q_n$
remotely on the database and stitching the results together using  $\Qstitch$.
The number of intermediate queries $n$ is the \emph{nesting degree} of
$\Result$, that is, the number of collection type constructors in the
result type.  For
example, the nesting degree of $\bagtl\record{A:\bagtl \Int, B:\bagtl \String}$
is 3.  The nesting degree of the type $\Result$ is also 
3, which means $Q$ can be shredded into three flat queries.  

The basic idea is straightforward. Whenever a nested bag appears in
the output of a query, we generate an index that uniquely identifies
the current context. Then a separate query produces the contents of
the nested bag, where each element is paired up with its parent
index. Each inner level of nesting requires a further query.

We will illustrate by showing the results of the three queries and how
they can be combined to reconstitute the desired nested result.
The outer query $q_1$ contains one entry for each department, with an
index $\record{a,id}$ in place of each nested collection:
\[\small\bl
r_1 = \resvl\ba[t]{@{}l@{\,}l@{}}
         \recordl\lDepartment = \cStr{Product},  &\lPeople = \record{a,1}\recordr, \\ 
         \recordl\lDepartment = \cStr{Quality},  &\lPeople = \record{a,2}\recordr ,\\ 
         \recordl\lDepartment = \cStr{Research}, &\lPeople = \record{a,3}\recordr ,\\ 
         \recordl\lDepartment = \cStr{Sales},    &\lPeople = \record{a,4}\recordr\resvr \\ 
\ea
\el
\]

The second query $q_2$ generates the data needed for the $\lPeople$
collections:
\[\small\bl
r_2 = 
\resvl
\ba[t]{@{}l@{\,}l@{}}
  \tuplel\tuple{a,1},
          \recordl\lName=\cStr{Bert}, &\lTasks=\tuple{b,1,2}\recordr\tupler, \\
  \tuplel\tuple{a,4},
          \recordl\lName=\cStr{Erik}, &\lTasks=\tuple{b,4,5}\recordr\tupler, \\
  \tuplel\tuple{a,4},
          \recordl\lName=\cStr{Fred}, &\lTasks=\tuple{b,4,6}\recordr\tupler, \\
  \tuplel\tuple{a,1},
          \recordl\lName=\cStr{Pat},  &\lTasks=\tuple{d,1,2}\recordr\tupler, \\
  \tuplel\tuple{a,4},
          \recordl\lName=\cStr{Sue},  &\lTasks=\tuple{d,4,7}\recordr\tupler\resvr \\
\ea\el
\]
The idea is to ensure that we can stitch the results of $q_1$ together
with the results of $q_2$ by joining the \emph{inner indexes} of $q_1$
(bound to the $\lPeople$ field of each result) with the \emph{outer
  indexes} of $q_2$ (bound to the first component of each result). In
both cases the static components of these indexes are the same tag
$a$.  Joining the $\lPeople$ field of $q_1$ to the outer index of
$q_2$ correctly associates each person with the appropriate
department.

Finally, let us consider the results of the innermost query $q_3$ for
generating the bag bound to the $\lTasks$ field:
\[\small\bl
r_3 = \resvl\bl
  \tuple{\tuple{b,1,2}, \cStr{build}}, 
  \tuple{\tuple{b,4,5}, \cStr{call}}, 
  \tuple{\tuple{b,4,5}, \cStr{enthuse}}, \\
  \tuple{\tuple{b,4,6}, \cStr{call}}, 
  \tuple{\tuple{d,1,2}, \cStr{buy}}, 
  \tuple{\tuple{d,4,7}, \cStr{buy}}\resvr \\
\el\el
\]
Recall that $q_2$ returns further inner indexes
for the tasks associated with each person. The two halves of the union
have different static indexes for the tasks $b$ and $d$, because they
arise from different comprehensions in the source term. Furthermore,
the dynamic index now consists of two $\lId$ fields ($x.\lId$ and
$y.\lId$) in each half of the union.   Thus,
joining the $\lTasks$ field of $q_2$ to the outer index of $q_3$
correctly associates each task with the appropriate outlier.

Note that each of the queries $q_1,q_2,q_3$ produces records that
contain other records as fields.  This is not strictly allowed by SQL,
but it is straightforward to simulate such nested records by rewriting
to a query with no nested collections in its result type; this is
similar to Van den Bussche's simulation~\cite{vandenbussche01tcs}.
However, this approach incurs extra storage and query-processing cost.
Later in the paper, we explore an alternative approach which collapses
the indexes at each level to a pair $\record{a,i}$ of static index and
a single ``surrogate'' integer, similarly to Ferry's
approach~\cite{grust10pvldb}.  For example, using this approach we
could represent the results of $q_2$ and $q_3$ as follows:
\[\small
\bl
r_2' = \resvl
\ba[t]{@{}l@{\,}l@{}}
  \tuplel\tuple{a,1},
          \recordl\lName=\cStr{Bert}, &\lTasks=\tuple{b,1}\recordr\tupler, \\
  \tuplel\tuple{a,4},
          \recordl\lName=\cStr{Erik}, &\lTasks=\tuple{b,2}\recordr\tupler, \\
  \tuplel\tuple{a,4},
          \recordl\lName=\cStr{Fred}, &\lTasks=\tuple{b,3}\recordr\tupler, \\
  \tuplel\tuple{a,1},
          \recordl\lName=\cStr{Pat},  &\lTasks=\tuple{d,1}\recordr\tupler, \\
  \tuplel\tuple{a,4},
          \recordl\lName=\cStr{Sue},
          &\lTasks=\tuple{d,2}\recordr\tupler\resvr \\
\ea
\smallskip\\
r_3' = \resvl\bl
  \tuple{\tuple{b,1}, \cStr{build}}, 
  \tuple{\tuple{b,2}, \cStr{call}}, 
  \tuple{\tuple{b,2}, \cStr{enthuse}}, \\
  \tuple{\tuple{b,3}, \cStr{call}}, 
  \tuple{\tuple{d,1}, \cStr{buy}}, 
  \tuple{\tuple{d,2}, \cStr{buy}}\resvr \\
\el
\el
\]

The rest of this paper gives the details of the shredding translation,
explains how to stitch the results of shredded queries back together, and shows
how to use $\rownum$ to avoid the space overhead of indexes.  We
will return to the above example throughout the paper, and we will use $\Qorg$,
$\Qoutliers$ and other queries based on this example in the
experimental evaluation.


\section{Shredding translation}
\label{sec:shredding}

As a pre-processing step, we annotate each comprehension body in a
normalised term with a unique name $a$ --- the static component of an
index.  We write the annotations as superscripts, for example:
\[
\For\,(\vec{G}\,\Where\,X)\,\retann{M}{a}
\]
In order to shred nested queries, we introduce an abstract type
$\Index$ of \emph{indexes} for maintaining the correspondence between
outer and inner queries. An index $\Ind{a}{d}$ has a static component
$a$ and a dynamic component $d$. The static component $a$ links the
index to the corresponding $\retann{}{a}$ in the query. The dynamic
component identifies the current bindings of the variables in the
comprehension.

Next, we modify types so that bag types have an explicit index
component and we use indexes to replace nested occurrences of bags
within other bags:
\begin{syntax}
\textrm{Shredded types} & A, B &\cce& \bagt{\tuple{\Index, F}} \\
\textrm{Flat types}     & F    &\cce& O \mid \record{\ora{\cl:F}} \mid \Index \\
\end{syntax}%

We also adapt the syntax of terms to incorporate indexes.  After
shredding, terms will have the following forms:
 \begin{syntax}
\text{Query terms}      & L, M &\cce& \bigunion\, \vec{C} 
\sepskip\\
\text{Comprehensions}   & C    &\cce& \retann{\tuple{I, N}}{a} \\
                        &      &\mid& \For\,(\vec{G}\,\Where\,X)\,C \\
\text{Generators}       & G    &\cce& x \gets t \\
\text{Inner terms}      & N    &\cce& X \mid R \mid I \\
\text{Record terms}     & R    &\cce& \record{\ora{\cl=N}} \\
\text{Base terms}       & X    &\cce& x.\cl \mid c(\vec{X}) \mid \isEmpty{L} 
\sepskip\\
\text{Indexes}          & I, J &\cce& \Ind{a}{d} \\
\text{Dynamic indexes}  & d    &\cce& \iup \mid \idown \\
\end{syntax}%
A comprehension is now constructed from a sequence of generator
clauses of the form $\For\,(\vec{G}\,\Where\,X)$ followed by a body of
the form $\retann{\tuple{I,N}}{a}$. Each level of nesting gives rise
to such a generator clause. The body always returns a pair
$\tuple{I,N}$ of an outer index $I$, denoting where the result values
from the shredded query should be spliced into the final nested
result, and a (flat) inner term $N$.
Records are restricted to contain inner terms, which are either base
types, records, or indexes, which replace nested multisets.
We assume a distinguished top level static index $\top$, which allows
us to treat all levels uniformly. Each shredded term is associated
with an outer index $\iup$ and an inner index $\idown$. (In fact
$\iup$ only appears in the left component of a comprehension body, and
$\idown$ only appears in the right component of a comprehension
body. These properties will become apparent when we specify the
shredding transformation on terms.)

\begin{figure*}[tb]

\[\small
\ba{@{}c@{\qquad}c@{}}
\begin{iequations}
  \shouter{L}_p &=& \bigunion\,(\shlist{L}_{\top,p})
  \sepskip\\
  \shlist{\bigunion\iton\,C_i}_{a,p} &=&
     \Concat([\shlist{C_i}_{a,p}]\iton) \\
  \shlist{\record{\cl_i=M_i}\iton}_{a,{\cl_j}.p} &=& \shlist{M_j}_{a,p} \\
  \shlist{\For\,(\vec{G}\,\Where\,X)\,\retann{M}{b}}_{a,\epsilon} &=& 
     [\For\,(\vec{G}\,\Where\,X)\,\retann{\tuple{\Ind{a}{\iup},\shinner{M}_b}}{b}] \\
  \shlist{\For\,(\vec{G}\,\Where\,X)\,\retann{M}{b}}_{a,\pbag.p} &=& 
    [\For\,(\vec{G}\,\Where\,X)\,C \mid C \gets
      \shlist{M}_{b,p}] \\
\end{iequations}
&
\begin{iequations}
\shinner{x.\cl}_a &=& x.\cl \\
\shinner{c([X_i]\iton)}_a &=& c([\shinner{X_i}_a]\iton) \\
\shinner{\isEmpty{L}}_a &=& \isEmpty{\shouter{L}_\epsilon}
\sepskip\\
\shinner{\record{\cl_i=M_i}\iton}_a &=& \record{\cl_i=\shinner{M_i}_a}\iton \\
\shinner{L}_a &=& \Ind{a}{\idown} \\
\end{iequations} \\
\ea
\]

\caption{Shredding translation on terms\label{fig:shredding-translation}}

\end{figure*}

\subsection{Shredding types and terms}\label{subsec:shredding-types-terms}

We use \emph{paths} to point to parts of types.
\begin{syntax}
\text{Paths}    &p &\cce& \epsilon \mid \pbag.p \mid \cl.p \\
\end{syntax}
The empty path is written $\epsilon$. A path $p$ can be extended by
traversing a bag constructor ($\pbag.p$) or selecting a label
($\cl.p$).
We will sometimes
write $p.\pbag$ for the path $p$ with $\pbag$ appended at the end and
similarly for $p.\cl$; likewise, we will write $p.\vec{\cl}$ for
the path $p$ with all the labels of $\vec{\cl}$ appended. The function
$\paths(A)$ defines the set of paths to bag types in a type $A$:
\begin{equations}
\paths(O)                     &=& \set{} \\
\paths(\record{\cl_i:A_i}\iton) &=&
 \bigcup\iton \set{\cl_i . p \mid p \gets \paths(A_i)} \\
\paths(\bagt{A})               &=&
  \set{\epsilon} \cup \set{\pbag . p \mid p \gets \paths(A)} 
\end{equations}

We now define a shredding translation on types.  This is defined in terms of the \emph{inner shredding}
$\shinner{A}$, a flat type that represents the contents of a bag. 
\begin{equations}
  \shinner{O}                     &=& O \\
  \shinner{\record{\cl_i:A_i}\iton}
  &=& \record{\cl_i:\shinner{A_i}}\iton \\
  \shinner{\bagt{A}} &=& \Index
\end{equations}
Given a path $p \in
\paths(A)$, the type $\shouter{A}_p$ is the \emph{outer shredding} of
$A$ at $p$.  It corresponds to the bag at path $p$ in
$A$.  
\begin{equations}
\shouter{\bagt{A}}_\epsilon                  &=& \bagt{\tuple{\Index, \shinner {A}}} \\
\shouter{\bagt{A}}_{\pbag.p}                &=& \shouter{A}_p \\
\shouter{\record{\ora{\cl:A}}}_{{\cl_i}.p}  &=& \shouter{A_i}_p 
\end{equations}%

For example, consider the result type $\Result$ from
Section~\ref{sec:example}.  Its nesting degree is 3, and its paths are:
\[
\paths(\Result) =
  \set{\epsilon, 
       \pbag.\lPeople.\epsilon,
       \pbag.\lPeople.\pbag.\lTasks.\epsilon}
\]%
We can shred $\Result$ in three ways using these three paths, yielding
three shredded types:
\begin{equations}
A_1 &=& \shouter{\Result}_{\epsilon} \\
A_2 &=& \shouter{\Result}_{\pbag.\lPeople.\epsilon} \\
A_3 &=& \shouter{\Result}_{\pbag.\lPeople.\epsilon.\pbag.\lTasks.\epsilon} \\
\end{equations}%
or equivalently:
\begin{equations}
A_1&=&\bagtl
  \tuple{\Index, \record{\lDepartment:\String, \lPeople:\Index}}
\\
A_2 &=& 
\bagt{\tuple{\Index, \record{\lName:\String, \lTasks:\Index}}}
\\
A_3 &=&
\bagt{\tuple{\Index, \String}}
\end{equations}%


The shredding translation on terms $\shouter{L}_p$ is given in
Figure~\ref{fig:shredding-translation}. This takes a term $L$ and a
path $p$ and gives a query $\shouter{L}_p$ that computes a result of
type $\shouter{A}_p$, where $A$ is the type of $L$. The auxiliary
translation $\shlist{M}_{a,p}$ returns the shredded comprehensions of
$M$ along path $p$ with outer static index $a$. The auxiliary
translation $\shinner{M}_a$ produces a flat representation of $M$ with
inner static index $a$.
Note that the shredding translation is linear in time and space.
Observe that for emptiness tests we need only the top-level query.

\begin{sloppypar}
Continuing the example, we can shred $\Qcomp$ in three ways,
yielding shredded queries:
\[
\bl
q_1 = \shouter{\Qcomp}_{\epsilon} \\
q_2 = \shouter{\Qcomp}_{\pbag.\lPeople.\epsilon} \\
q_3 = \shouter{\Qcomp}_{\pbag.\lPeople.\epsilon.\pbag.\lTasks.\epsilon} \\
\el
\]
or equivalently:
\end{sloppypar}
\[\small
\ba{@{}r@{~}c@{~}l@{}}
  q_1 &=& \bl \For\,(x \gets \tDepartments) \\
  \retann{\tuple{\Unit,
                   \record{\lDepartment = x.\lName,
                           \lPeople     = \Ind{a}{\idown}}}}{a}
   \el
\smallskip\\
q_2 &=&
  (\bl
   \For\,(x \gets \tDepartments) \\
   \For\,(y \gets \tEmployees)\,\Where\,(x.\lName = y.\dpt ~\wedge\quad\\
   \hfill                                (y.\lSalary < 1000 \vee y.\lSalary > 1000000)) \\
   \retlann{b}\recordl
   \Ind{a}{\iup}, \record{\lName    = y.\lName,
                             \lTasks   = \Ind{b}{\idown}}
                               \recordr\retrann{b}) \\
  \el \\
  &\union&
  (\bl
     \For\,(x \gets \tDepartments) \\
     \For\,(y \gets \tContacts)\,
     \Where\,(x.\lName = y.\dpt \wedge y.\lClient) \\
     \retlann{d}\recordl
        \Ind{a}{\iup}, \record{\lName  = y.\lName,
                       \lTasks = \Ind{d}{\idown}}
                         \recordr\retrann{d}) \\
  \el
\smallskip \\
q_3 &=&
  (\bl
     \For\,(x \gets \tDepartments) \\
     \For\,(y \gets \tEmployees)\,\Where\,
                         (x.\lName = y.\dpt ~\wedge\quad \\
     \hfill               (y.\lSalary < 1000 \vee y.\lSalary > 1000000)) \\
     \For\,(z \gets \tTasks)\,\Where\,(z.\emp = y.\emp) \\
     \retann{\tuple{\Ind{b}{\iup}, z.\tsk}}{c}) \\
   \el \\ 
  &\union&
  (\bl
     \For\,(x \gets \tDepartments) \\
     \For\,(y \gets \tContacts)\,
     \Where\,(x.\lName = y.\dpt \wedge y.\lClient) \\
     \retann{\tuple{\Ind{d}{\iup}, \cStr{buy}}}{e}) \\
   \el
\ea
\]

As a sanity check, we show that well-formed normalised terms
shred to well-formed shredded terms of the appropriate shredded types.
We discuss other correctness properties of the shredding translation in
Section~\ref{sec:run-shredded}.
\begin{tronly}
  Typing rules for shredded terms are shown in
  Appendix~\ref{app:typing-rules}.
\end{tronly}
\begin{theorem}
  Suppose $L$ is a normalised flat-nested query with $\vdash L : A$ and
  $p \in \paths(A)$, then $\vdash \shouter{L}_p : \shouter{A}_p$.
\end{theorem}

\subsection{Shredded packages}
\label{subsec:shredded-packages}

\newcommand{\package}{\mathit{package}}
\newcommand{\shred}{\mathit{shred}}
\newcommand{\shredT}{\mathit{Shred}}
\newcommand{\erase}{\mathit{erase}}

\newcommand{\shann}[2]{({#1})^{#2}}

\newcommand{\sht}[1]{\hat{#1}} 

To maintain the relationship between shredded terms and the structure
of the nested result they are meant to construct, we use
\emph{shredded packages}.  A shredded package $\sht{A}$ is a nested
type with annotations, denoted $\shann{-}{\alpha}$, attached to each
bag constructor.
\begin{syntax}
& \sht{A} &\cce& O \mid \record{\ora{\cl:\sht{A}}} \mid \shann{\bagt{\sht{A}}}{\alpha} \\
\end{syntax}%
For a given package, the annotations are drawn from the same set. We
write $\sht{A}(S)$ to denote a shredded package with annotations drawn
from the set $S$. 
We sometimes omit the type
parameter when it is clear from context. \begin{tronly} Typing rules
  for shredded packages are shown in Appendix~\ref{app:typing-rules}.
\end{tronly}

Given a shredded package $\sht{A}$, we can erase its annotations to
obtain its underlying type.
\begin{equations}
\erase(O)                             &=& O \\
\erase(\record{\cl_i:\sht{A_i}}\iton)    &=& \record{\cl_i:\erase(\sht{A_i})}\iton \\
\erase(\shann{\bagt{\sht{A}}}{\alpha}) &=& \bagt{(\erase(\sht{A}))} 
\end{equations}%
%

\begin{tronly}
Given a type $A$ and a shredding function $f : \paths(A) \to S$, we
can construct a shredded package $\sht{A}(S)$.
\begin{equations}
\package_{f}(A)                       &=& \package_{f,\epsilon}(A) 
\sepskip\\
\package_{f,p}(O)                     &=& O \\
\package_{f,p}(\record{\cl_i:A_i}\iton)
  &=& \record{\cl_i:\package_{f,p.\cl_i}(A_i)}\iton \\
\package_{f,p}(\bagt{A})               &=& \shann{\bagt{(\package_{f,p.\pbag}(A))}}{f(p)} 
\end{equations}%


Using $\package$, we lift the type- and term-level shredding
functions $\shouter{-}_{-}$ 
to produce shredded packages, where each annotation contains the
shredded version of the input type or query along the path to the
associated bag constructor.
\begin{equations}
\shred_{B}(A) &=& \package_{(\shouter{B}_{-})}(A) 
\\
\shred_L(A) &=& \package_{(\shouter{L}_{-})}(A) 
\end{equations}%
\end{tronly}
\begin{subonly}
  We lift the type shredding function $\shouter{A}$ to produce
  a shredded package $\shredT(A)$, where each annotation contains the shredded
  version of the input type or query along the path to the associated
  bag constructor.  We define $\shredT(A) = \shredT_\epsilon(A)$ as follows:
  \begin{equations}
\shredT_{p}(O)                     &=& O \\
\shredT_{p}(\record{\cl_i:A_i}\iton)
  &=& \record{\cl_i:\shredT_{p.\cl_i}(A_i)}\iton \\
\shredT_{p}(\bagt{A})               &=& \shann{\bagt{(\shredT_{p.\pbag}(A))}}{\shouter{B}_p} 
\end{equations}%
Similarly, we lift the term shredding function$\shouter{L}$ to produce
a shredded query package $\shred(L) = \shred_\epsilon(L)$, where:
 \begin{equations}
\shred_{p}(O)                     &=& O \\
\shred_{p}(\record{\cl_i:A_i}\iton)
  &=& \record{\cl_i:\shred_{p.\cl_i}(A_i)}\iton \\
\shred_{p}(\bagt{A})               &=& \shann{\bagt{(\shred_{p.\pbag}(A))}}{\shouter{L}_p} 
\end{equations}%
\end{subonly}

  For example, the shredded package for the $\Result$ type from
  Section~\ref{sec:example} is:
  \[\small
  \bl \shredT_\Result(\Result) = \\
  \quad \bagtl\recordl \bl
        \lDepartment:\String, \\
        \lPeople:\bagtl\recordl \bl
        \lName:\String, \\
        \lTasks:(\bagt{\String})^{A_3}\recordr\bagtr^{A_2}\recordr\bagtr^{A_1} \\
        \el \\
  \el
  \el
  \]
  where $A_1,A_2$, and $A_3$ are as shown in
  Section~\ref{subsec:shredding-types-terms}.  Shredding the
  normalised query $Q'$ gives the same package, except the type
  annotations $A_1,A_2,A_3$ become queries $q_1,q_2,q_3$.

Again, as a sanity check we show that erasure is the left inverse of
type shredding and that term-level shredding operations preserve
types. 
\begin{theorem}
  For any type $A$, we have $ \erase(\shred_{A}(A)) = A$.
  Furthermore, if $L$ is a closed, normalised, flat--nested query such
  that $\vdash L : A$ then $\vdash \shred_L(A) : \shred_{A}(A)$.
\end{theorem}


\section{Query evaluation and stitching}
\label{sec:run-shredded}
\renewcommand{\ann}[1]{}

Having shredded a normalised nested query, we can then run all of the
resulting shredded queries separately. If we stitch the shredded
results together to form a nested result, then we obtain the same
nested result as we would obtain by running the nested query
directly. In this section we describe how to run shredded queries and
stitch their results back together to form a nested result.

\begin{figure*}[tb]
\[\small
\ba{@{}c@{\quad}c@{\quad}c@{}} \ba{r@{~}c@{~}l}
  \ssem{L} &=& \ssem{L}_{\varepsilon,1} \\
  \ea
&
  \ba{r@{~}c@{~}l}
  \ssem{\record{\cl=N}\iton}_{\rho,\iota} &=& \record{\cl_i=\ssem{N_i}_{\rho,\iota}}\iton \\
  \ssem{X}_{\rho,\iota} &=& \nsem{X}_\rho \\
  \ea
&
  \ba{r@{~}c@{~}l}
  \ssem{\Ind{a}{\iup}}_{\rho,\iota.i}        &=& \key(\Ind{a}{\iota}) \\
  \ssem{\Ind{a}{\idown}}_{\rho,\iota.i}      &=& \key(\Ind{a}{\iota.i}) \\
  \ea\\
\ea
\]
\[\small
\bl
\ssemcomp{\bigunion\iton C_i}_{\rho,\iota} = \Concat([\ssemcomp{C_i}_{\rho,\iota}]\iton) \hfill
\ssemcomp{\retann{N}{a}}_{\rho,\iota} = [\bsem{N}_{\rho,\iota} \ann{\key(\Ind{a}{\iota})}]
\smallskip\\
\ssemcomp{\For\,([x_i \gets t_i]\iton~\Where\,X)\,C}_{\rho,\iota} = 
\Concat([\ssemcomp{C}_{\rho[x_i \mapsto r_i]\iton,\iota.j} 
   \mid \tuple{j, \vec{r}} \gets 
      \enum([\vec{r} \mid [r_i \gets \tsem{t_i}]\iton, \nsem{X}_{\rho[x_i \mapsto r_i]\iton}])]) \\
\el
\]
\caption{Semantics of shredded queries}
\label{fig:shredded-semantics}
\end{figure*}

\subsection{Evaluating shredded queries}
\label{subsec:shredded-semantics}

The semantics of shredded queries $\ssem{-}$ is given in
Figure~\ref{fig:shredded-semantics}.
Apart from the handling of indexes, it is much like the semantics for
nested queries given in Figure~\ref{fig:bag-semantics}.
To allow different implementations of indexes, we first define a
canonical notion of index, and then parameterise the semantics by the
concrete type of indexes $X$ and a function $\key : \Index \to X$
mapping each canonical index to a distinct concrete index.
A canonical index $\Ind{a}{\iota}$ comprises static index $a$ and
dynamic index $\iota$, where the latter is a list of positive natural
numbers. For now we take concrete indexes simply to be canonical
indexes, and $\key$ to be the identity function. We consider other
definitions of $\key$ in Section~\ref{sec:indexing-schemes}.

The current dynamic index $\iota$ is threaded through the semantics in
tandem with the environment $\rho$. The former encodes the position of
each of the generators in the current comprehension and allows us to
invoke $\key$ to construct a concrete index. The outer index at
dynamic index $\iota.i$ is $\iota$; the inner index is $\iota.i$.
In order for a comprehension to generate dynamic indexes we use the
function $\enum$ (introduced in Section~\ref{sec:background}) that
takes a list of elements and returns the same list with the element
number paired up with each source element.

Running a shredded query yields a list of pairs of indexes and
shredded values.
\begin{syntax}
\text{Results} & s     &\cce& [\tuple{I_1,\iv_1}\ann{J_1}, \dots,
\tuple{I_m,\iv_m}\ann{J_m}] \\
\text{Flat values}      & \iv   &\cce& c \mid \record{\cl_1=\iv_1, \dots,
  \cl_n=\iv_n} \mid I
\end{syntax}%



Given a shredded package $\sht{A}(S)$ and a function $f : S \to T$, we
can map $f$ over the annotations to obtain a new shredded package
$\sht{A'}(T)$ such that $\erase(\sht{A}) = \erase(\sht{A'})$.
\begin{equations}
\pmap_{f}(O)                              &=& O \\
\pmap_{f}(\record{\cl_i:\sht{A_i}}\iton)
  &=& \record{\cl_i:\pmap_{f}(\sht{A_i})}\iton \\
\pmap_{f}(\shann{\bagt{\sht{A}}}{\alpha}) &=& \shann{(\bagt{\pmap_{f}(\sht{A})})}{f(\alpha)} 
\end{equations}%

  \begin{sloppypar}
    The semantics of a shredded query package is a \emph{shredded
      value package} containing indexed results for each shredded
    query.  For each type $A$ we define $\psem{A} = \shred_A(A)$ and
    for each flat--nested, closed $\vdash L : A$ we define $\psem{L}_A
    : \psem{A}$ as $\pmap_{\ssem{-}}~(\shred_L(A))$.  In other words,
    we first construct the shredded query package $\shred_L(A)$, then
    apply the shredded semantics $\ssem{q}$ to each query $q$ in the package.
  \end{sloppypar}

For example, here is the shredded package that we obtain after running
the normalised query $\Qcomp$ from Section~\ref{subsec:normalisation}:
\begin{equations}
\psem{\Qcomp}_A &=& \bagtl\tuplel
   \bl
     \lDepartment:\String, \\
     \lPeople:\bagtl\recordl
       \bl
         \lName:\String, \\
         \lTasks:(\bagt{\String})^{r_3}\recordr\bagtr^{r_2}\tupler\bagtr^{r_1} \\
       \el \\
   \el \\
\end{equations}%
where $r_1$, $r_2$, and $r_3$ are as in Section~\ref{sec:example}
except that indexes are of the form $\Ind{a}{1.2.3}$ instead of
$\tuple{a,1,2,3}$.

\subsection{Stitching shredded query results together}

A shredded value package can be \emph{stitched} back together into a
nested value, preserving annotations, as follows:
\begin{equations}
\buildtop{\sht{A}} &=&
  \build{\Ind{\top}{1}}{\sht{A}} 
\smallskip\\
\build{c}{O}                                              &=& c \\
\build{r}{\record{\cl_i:\sht{A_i}}\iton}  &=&
  \record{\cl_i=\build{r.\cl_i}{\sht{A_i}}}\iton \\
\build{I}{\shann{\bagt{\sht{A}}}{s}} &=& 
    [(\build{\iv}{\sht{A}}) \mid \tuple{I, \iv} \gets s] 
\end{equations}%
The  flat value parameter $\iv$ to the auxiliary function
$\build{\iv}{-}$ specifies which values to stitch along the current
path.

Resuming our running example, once the results $r_1:A_1,r_2:
A_2,r_3:A_3$ have been evaluated on the database, they are
shipped back to the host system where we can run the following code
in-memory to stitch the three tables together into a single
value: the result of the original nested query. The code for this
query $\Qstitch$
follows the same idea as the query $\Qorg$ that constructs the
nested $\tOrganisation$ from $\Sigma$.
\[\small
\bl
  \For\,(x \gets r_1) \\
  \retl\recordl
    \bl
      \lDepartment   = x.\lName, \\
      \lPeople = \bl
           \For\,(\tuple{i, y} \gets r_2)\\
           \Where\,(x.\lPeople = i)) \\
           \retl\recordl
             \bl
               \lName  = y.\lName, \\
               \lTasks = \bl
                   \For\,(\tuple{j, z} \gets r_3)\\
                   \Where\,(y.\lTasks = j) \\
                   \ret{z}\recordr\retr\recordr\retr \\
                 \el \\
             \el \\
        \el \\
    \el \\
\el
\] 


We can now state our key correctness property: evaluating shredded
queries and stitching the results back together yields the same
results as evaluating the original nested query directly.
\begin{theorem}
\label{th:shred-stitch}
If $\vdash L : \bagt{A}$ then: \[\buildtop{\psem{L}_{\bagt{A}}} = \nsem{L}\]
\end{theorem}
\begin{proof}[Sketch]
  We omit the full proof due to space limits; it is available in the
  full version of this paper.  The proof introduces several
  intermediate definitions.  Specifically, we consider an
  \emph{annotated semantics} for nested queries in which each
  collection element is tagged with an index, and we show that this
  semantics is consistent with the ordinary semantics if the
  annotations are erased.  We then prove the correctness of shredding
  and stitching with respect to the annotated semantics, and the
  desired result follows.
\end{proof}

\renewcommand{\ann}[1]{{\color{red}@{#1}}}


\section{Indexing schemes}
\label{sec:indexing-schemes}

\begin{sloppypar}
So far, we have worked with canonical indexes of the form
$\Ind{a}{1.2.3}$.  These could be represented in SQL by using multiple
columns (padding with NULLs if necessary) since for a given query the
length of the dynamic part is bounded by the number of
$\For$-com\-pre\-hensions in the query.  This imposes space and running
time overhead due to constructing and maintaining the indexes.
Instead, in this section we consider alternative, more compact
indexing schemes.
\end{sloppypar}

We can define alternative indexing schemes by varying the $\key$ parameter
of the shredded semantics (see
Section~\ref{subsec:shredded-semantics}). Not all possible
instantiations are valid. To identify those that are, we first define
a function for computing the canonical indexes of a nested query result.
\begin{equations}
\isem{L} &=& \isem{L}_{\varepsilon,\DUnit}
\smallskip\\
\isem{\bigunion\iton\,C_i}_{\rho,\iota}
  &=& \Concat([\isem{C_i}_{\rho,\iota}]\iton) \\
\isem{\record{\cl_i=M_i}\iton}_{\rho,\iota}
  &=& \Concat([\isem{M_i}_{\rho,\iota}]\iton) \\
\isem{X}_{\rho,\iota} &=& \nil
\sepskip\\
\multicolumn{3}{l}
{\bl
 \isem{\For\,([x_i \gets t_i]\iton~\Where\,X)\,\retann{M}{a}}_{\rho,\iota} = \\
 \quad\Concat([\Ind{a}{\iota.j} \cons \isem{M}_{\rho[x_i \mapsto r_i]\iton,\iota.j} \\
 \quad\quad\mid \tuple{j, \vec{r}} \gets 
      \enum([\vec{r} \mid [r_i \gets \tsem{t_i}]\iton, \nsem{X}_{\rho[x_i \mapsto r_i]\iton}])]) \\
 \el}
\end{equations}%
An indexing function $\key : \Index \to X$ is
\emph{valid} with respect to the closed nested query $L$ if it is
injective and defined on every canonical index in $\isem{L}$.
The only requirement on indexes in the proof of
Theorem~\ref{th:shred-stitch} is that $\key$ is valid.
We consider two alternative valid indexing schemes: 
\emph{natural} and \emph{flat
  indexes}.

\newcommand{\flatkey}{\key^\flat}

\subsection{Natural indexes}  
Natural indexes are synthesised from row
data. In order to generate a natural index for a query every table
must have a key, that is, a collection of fields guaranteed to be
unique for every row in the table. For sets, this is always possible
by using all of the field values as a key; this idea is used in Van
den Bussche's simulation for sets~\cite{vandenbussche01tcs}.  However,
for bags this is not always possible, so using natural indexes may
require adding extra key fields.  

\newcommand{\naturalkey}{\key^\natural}
\newcommand{\pkey}{\mathit{key}}
Given a table $t$, let $\pkey_{t}$ be the function that given a row
$r$ of $t$ returns the key fields of $r$.
We now define a function to compute the list of natural indexes for a
query $L$.
\begin{equations}
\inatsem{L} &=& \inatsem{L}_{\varepsilon}
\smallskip\\
\inatsem{\bigunion\iton\,C_i}_{\rho}
  &=& \Concat([\inatsem{C_i}_{\rho}]\iton) \\
\inatsem{\record{\cl_i=M_i}\iton}_{\rho}
  &=& \Concat([\inatsem{M_i}_{\rho}]\iton) \\
\inatsem{X}_{\rho} &=& \nil
\smallskip\\
\multicolumn{3}{l}
{\bl
 \inatsem{\For\,([x_i \gets t_i]\iton~\Where\,X)\,\retann{M}{a}}_{\rho} = \\
 \quad\Concat([\Ind{a}{\tuple{\pkey_{t_i}(r_i)}}\iton \cons \inatsem{M}_{\rho[x_i \mapsto r_i]\iton} \\
 \qquad\qquad\qquad\mid  
     [r_i \gets \tsem{t_i}]\iton, \nsem{X}_{\rho[x_i \mapsto r_i]\iton}]) \\
 \el}
\end{equations}%
If $\Ind{a}{\iota}$ is the $i$-th element of $\isem{L}$, then
$\naturalkey_L(\Ind{a}{\iota})$ is defined as the $i$-th element of
$\inatsem{L}$.
The natural indexing scheme is defined by setting $\key =
\naturalkey_L$.

An advantage of natural indexes is that they can be implemented in
plain SQL, so for a given comprehension all where clauses can be
amalgamated (using the $\wedge$ operator) and no auxiliary subqueries
are needed. The downside is that the type of a dynamic index may still
vary across the component comprehensions of a shredded query,
complicating implementation of the query (due to the need to pad some
subqueries with null columns) and potentially decreasing performance
due to increased data movement.  We now consider an alternative, in
which $\rownum$ is used to generate dynamic indexes.

\subsection{Flat indexes and let-insertion}
\label{subsec:flat-indexes}

The idea of flat indexes is to enumerate all of the canonical dynamic
indexes associated with each static index and use the enumeration as
the dynamic index.
%
%

Let $\iota$ be the $i$-th element of the list $[\iota' \mid
  \Ind{a}{\iota'} \gets \isem{L}]$, then $\flatkey_L(\Ind{a}{\iota}) =
\tuple{a, i}$.
The flat indexing scheme is defined by setting $\key = \flatkey_L$.
Let $\iflatsem{L} = [\flatkey_L(I) \mid I \gets \isem{L}]$ and let
$\sflatsem{L}$ be $\ssem{L}$ where $\key = \flatkey_L$.

In this section, we give a translation called let-insertion that uses
let-binding and an $\cindex$ primitive to manage flat indexes.  In the
next section, we take the final step from this language to SQL.

Our semantics for shredded queries uses canonical indexes. We now
specify a target language providing flat indexes. In order to do so, we introduce
$\Let$-bound sub-queries, and translate each comprehension into the
following form:
\[
  \bl
    \Let~q = \For\,(\ora{G_\iup}\,\Where\,X_\iup)\,\ret{N_\iup}\,\In \\
    \quad \For\,(\ora{G_\idown}\,\Where\,X_\idown)\,\ret{N_\idown}
  \el
\]
The special $\cindex$ expression is available in each loop body, and
is bound to the current index value.

Following let-insertion, the types are as before, except indexes are
represented as pairs of integers $\tuple{a,d}$ where $a$ is the static
component and $d$ is the dynamic component.
\begin{syntax}
\textrm{Types}       & A, B    &\cce& \bagt{\tuple{\tuple{\Int, \Int}, F}} \\
\textrm{Flat types}  & F       &\cce& O \mid \record{\ora{\cl:F}}
                                \mid \tuple{\Int, \Int} \\
\end{syntax}%

The syntax of terms is adapted as follows:
 \begin{syntax}
\text{Query terms}         & L, M &\cce& \bigunion\,\vec{C} 
\sepskip\\
\text{Comprehensions}      & C    &\cce& \Let~q = S\,\In\,S' \\
\text{Subqueries}          & S    &\cce& \For\,(\vec{G}\,\Where\,X)\,\ret{N} \\
\text{Data sources}        & u    &\cce& t \mid q \\
\text{Generators}          & G    &\cce& x \gets u \\
\text{Inner terms}         & N    &\cce& X \mid R \mid \cindex \\
\text{Record terms}        & R    &\cce& \record{\ora{\cl=N}} \\
\text{Base terms}          & X    &\cce& x.\vec{\cl} \mid c(\vec{X}) \mid \isEmpty{L} \\
\end{syntax}

\begin{figure*}[tb!]
\[\small
\ba{@{}cccc@{}}
\begin{iequations}
\lsem{L}                       &=& \lsem{L}_\varepsilon \\
\lsem{\bigunion\jtom\,C_j}_\rho &=& \bagv{\Concat([\lsemcomp{C_j}_\rho]\jtom)} \\
\end{iequations}
&
\begin{iequations}
\lsem{t}_\rho &=& \tsem{t} \\
\lsem{q}_\rho &=& \rho(q) \\
\end{iequations}
&
\begin{iequations}
\lsem{\record{\cl_j=N_j}\jtom}_{\rho,i} &=& \record{\cl_j=\lsem{N_j}_{\rho,i}}\jtom  \\
\lsem{X}_{\rho,i}                       &=& \nsem{X}_\rho\\
\end{iequations}
&
\begin{iequations}
\lsem{\cindex}_{\rho,i} &=& i \\
\end{iequations}
\ea
\]
\small
\begin{equations}
  \lsemcomp{\Let\,q = S_\iup\,\In\,S_\idown}_\rho &=& \lsemcomp{S_\idown}_{\rho[q \mapsto \lsemcomp{S_\iup}_\rho]} \\
  \lsemcomp{\For\,([x_j \gets u_j]\jtom~\Where\,X)\,\ret{N}}_\rho &=& 
  [\lsem{N}_{\rho[x_j \mapsto r_j]\jtom,i} 
   \mid \tuple{i, \vec{r}} \gets 
        \enum([\vec{r} \mid [r_j \gets \lsem{u_j}_\rho]\jtom, \nsem{X}_{\rho[x_j \mapsto r_j]\jtom}])] \\
\end{equations}%
\caption{Semantics of let-inserted shredded queries
         \label{fig:let-inserted-semantics-flat}}
\end{figure*}%
The semantics of let-inserted queries is given in
Figure~\ref{fig:let-inserted-semantics-flat}. Rather than maintaining
a canonical index, it generates a flat index for each subquery.

We first give the  translation on shredded types as follows:
\begin{equations}
\lins(O) &=& O\\
\lins(\record{\ora{\cl:F}}) &=& \record{\ora{\cl:\lins(F)}}\\
\lins(\Index) &=&  \tuple{\Int, \Int}\\
\lins(\bagt{\tuple{\Index, F}}) &=& \bagt{\tuple{\tuple{\Int, \Int}, \lins( F)}} 
\end{equations}
For example:
\[\lins(A_2) =
\bagt{\tuple{\tuple{\Int,\Int},\record{\lName:\String,\lTasks:\tuple{\Int,\Int}}}}\]

\newcommand{\collapse}[1]{\mathopen{\parallel}{#1}\mathclose{\parallel}}
\newcommand{\mkfor}{\mathit{for}}

\newcommand{\len}{\mathit{length}}
\newcommand{\expand}{\mathit{expand}}
\newcommand{\gens}{\mathit{gens}}
\newcommand{\conds}{\mathit{conds}}
\newcommand{\body}{\mathit{body}}
\newcommand{\gs}{\mathit{gs}}

\begin{figure*}[tb]
\[\small\bl
\begin{iequations}
\lins(\bigunion\iton\,C_i) &=& \bigunion\iton\,\lins(C_i) \\
   \lins(C) &=&
  \ba[t]{@{}l@{}}
  \Let\,q = (\For\,(\ora{G_\iup}\,\Where\,X_\iup)\,
             \ret{\tuple{R_\iup, \cindex}})\,
  \In\, 
    \For\,(z \gets q, \ora{G_\idown}\,\Where\,\lins_{\vec{y}}(X_\idown))\,
           \ret{\lins_{\vec{y}}(N)} \\
  \ea \\
  \multicolumn{3}{c}
  {\ba{cccc}
     \begin{iequations}
     \text{where}\quad\ora{G_\iup}         &=& \Concat~(\init~(\gens~C)) \\
     X_\iup               &=& \bigwedge\init~(\conds~C) \\
     \end{iequations}
     &
     \begin{iequations}
     \ora{y = t} &=& \ora{G_\iup} \\
     R_\iup               &=& \tuple{\expand(y_i,t_i)}\iton \\
     \end{iequations}
     &
     \begin{iequations}
     \ora{G_\idown}       &=& \last~(\gens~C) \\
     X_\idown             &=& \last~(\conds~C) \\
     \end{iequations}
     &
     \begin{iequations}
     N                   &=& \body~C \\
     n                   &=& \len~\ora{G_\iup} \\
     \end{iequations} \\
   \ea} \\
\end{iequations} \sepskip\\
\ba{@{}cc@{}}
\begin{iequations}
\lins_{\vec{y}}(x.\cl)       &=& 
  \left\{
    \ba{ll}
      x.\cl,      &\text{if }x \notin \set{y_1, \dots, y_n} \\
      z.1.i.\cl,  &\text{if }x = y_i \\
    \ea
  \right. \\
\lins_{\vec{y}}(c(X_1,\dots,X_m))  &=& c(\lins_{\vec{y}}(X_1),\dots,\lins_{\vec{y}}(X_m)) \\
\end{iequations}
&
\begin{iequations}
\lins_{\vec{y}}(\isEmpty{L}) &=& \isEmpty(\lins_{\vec{y}}(L)) \\
\lins_{\vec{y}}(\bigunion\iton\,C_i) &=& \bigunion\iton\,\lins_{\vec{y}}(C_i) \\
\lins_{\vec{y}}(\For\,(\vec{G}\,\Where\,X)\,
               \retann{\tuple{a,N}}{a})
               &=&
  \bl
    \For\,(\vec{G}\,\Where\,\lins_{\vec{y}}(X)) \\
    \ret{\tuple{a,\lins_{\vec{y}}(N)}}
  \el \\
\end{iequations}
\sepskip\\
\begin{iequations}
\lins_{\vec{y}}(\record{\cl_j=X_j}\jtom) &=& \record{\cl_j=\lins_{\vec{y}}(X_j)}\jtom  \\
\lins_{\vec{y}}(\Ind{a}{d} ) &=& \tuple{a,\lins(d)} \\
\end{iequations}
&
\begin{iequations}
\lins   (\iup)          &=& z.2 \\
\lins (\idown)          &=& \cindex \\
\end{iequations} \\
\ea \el
\]

\caption{The let-insertion translation\label{fig:let-insertion-flat}}
\end{figure*}

Without loss of generality we rename all the bound variables in our
source query to ensure that all bound variables have distinct names,
and that none coincides with the distinguished name $z$ used for
let-bindings. The let-insertion translation $\lins$ is defined in
Figure~\ref{fig:let-insertion-flat}, where we use the following auxiliary functions:
\begin{equations}
\expand(x, t) &=& \record{\cl_i=x.\cl_i}\iton\\
&& \text{ where }
\Sigma(t) = \bagt{\record{\ora{\cl:A}}}
\smallskip\\
\gens~(\For\,(\vec{G}\,\Where\,X)\,C) &=& \vec{G} \cons \gens~C \\
\gens~(\retann{N}{a})                 &=& \nil \smallskip\\
\conds~(\For\,(\vec{G}\,\Where\,X)\,C) &=& X \cons \conds~C \\
\conds~(\retann{N}{a})                 &=& \nil \smallskip\\
\body~(\For\,(\vec{G}\,\Where\,X)\,C) &=& \body~C \\
\body~(\retann{N}{a})                &=& N 
\end{equations}%
Each comprehension is rearranged
into two sub-queries. The first generates the outer indexes. The second
computes the results.  The translation sometimes produces $n$-ary projections in order to
refer to values bound by the first subquery inside the second.

For example, applying $\lins$ to $q_1$ from
Section~\ref{subsec:shredded-packages} yields:
\[\small
\bl
    \For\,(x \gets \tDepartments) \\
    \ret{\tuple{\tuple{\top,1},\record{\dpt=x.\lName, \lPeople=\cindex}}}
\el\]
and $q_2$ becomes:
\[\small
\bl
(\bl
     \Let\,q =
       \For\,(x \gets \tDepartments)\,\ret{\tuple{\record{\dpt=x.\lName}, \cindex}}\,
     \In \\
       ~~\For\,(z \gets q, y \gets \tEmployees)\,\Where\,(z.\lab{1}.\lab{1}.\lName = y.\dpt ~\wedge\quad\\
       \hfill                                (y.\lSalary < 1000 \vee y.\lSalary > 1000000)) \\
       ~~ \retl\recordl
          \tuple{a, z.\lab{2}}, \record{\lName  = y.\lName,
                                  \lTasks = \tuple{b, \cindex}}
                                    \recordr\retr) \\
  \el \\
  \union \\
  (\bl
     \Let\,q =
       \For\,(x \gets \tDepartments)\,\ret{\tuple{\record{\dpt=x.\lName}, \cindex}}\,
     \In \\
       ~~\For\,(z \gets q, y \gets \tContacts)\,
         \Where\,(z.\lab{1}.\lab{1}.\lName = y.\dpt \wedge y.\lClient) \\
       ~~ \retl\recordl
            \tuple{a, z.\lab{2}}, 
       \record{\lName  = y.\lName,
               \lTasks = \tuple{d, \cindex}}
                 \recordr\retr) \\
   \el \\
\el
\]
%
As a sanity check, we show that the translation is type-preserving:
\begin{theorem}
  Given shredded query $\vdash M :
  \bagt{\tuple{\Index, F}}$, then $ \vdash \lins(M) :
  \lins(\bagt{\tuple{\Index, F}}) $.
\end{theorem}

To prove the correctness of let-insertion, we need to show that the
shredded semantics and let-inserted semantics agree.  In the statement
of the correctness theorem, recall that $\sflatsem{L}$ refers to the
version of $\ssem{L}$ using $\key =\flatkey_L$.
\begin{theorem}
\begin{sloppypar}
Suppose $\vdash L : A$ and $\shouter{L}_p = M$. Then $\sflatsem{M} = \lsem{\lins(M)}$.
\end{sloppypar}
\end{theorem}
\begin{proof}[\proofheader sketch]
The high-level idea is to separate results into data and indexes and
compare each separately.
It is straightforward, albeit tedious, to show that the different definitions
are equal if we replace all dynamic indexes by a
dummy value.
It then remains to show that the dynamic indexes agree. The pertinent
case is the translation of a comprehension:
\[
[\For\,(\vec{G_i} \gets X_i)]\iton\,
\For\,(\vec{G_\idown} \gets X_\idown)\,\retann{\tuple{\Ind{a}{\iup}, N}}{b}
\]
which becomes $\Let\,q = S_\iup\,\In\,S_\idown$ 
for suitable $S_\iup$ and $S_\idown$. The dynamic indexes computed by
$S_\iup$ coincide exactly with those of $\iflatsem{L}$ at static index
$a$, and the dynamic indexes computed by $S_\idown$, if there are any,
coincide exactly with those of $\iflatsem{L}$.
\end{proof}


\section{Conversion to SQL}
\label{sec:to-sql}

Earlier translation stages have employed nested records for
convenience, but SQL does not support nested records. At this stage,
we eliminate nested records from queries. For example, we can represent a nested record 
$\record{\lab{a}=\record{\lab{b}=1,\lab{c}=2},\lab{d}=3}$
 as a flat record $\record{\lab{a{\delim}b}=1,\lab{a{\delim}c}=2,\lab{d}=3}$.
\begin{tronly}
  The (standard) details are presented in
  Appendix~\ref{app:record-flattening}.
\end{tronly}

In order to interpret shredded, flattened, let-inserted terms as SQL,
we interpret index generators using SQL's OLAP facilities.
\begin{syntax}
\text{Query terms}         & L    &\cce& (\Unionall)\,\vec{C} 
\sepskip\\
\text{Comprehensions}      & C    &\cce& \With\,q\,\As\,(S)\,C \mid S' \\
\text{Subqueries}          & S    &\cce& \Select\,R~\From\,\vec{G}\,\Where\,X \\
\text{Data sources}        & u    &\cce& t \mid q \\
\text{Generators}          & G    &\cce& u\,\As\,x \\
\text{Inner terms}         & N    &\cce& X \mid \rownumindex{\vec{X}} \\
\text{Record terms}        & R    &\cce& \ora{N\,\As\,\cl} \\
\text{Base terms}          & X    &\cce& x.\cl \mid c(\vec{X}) \mid \isEmpty{L} \\
\end{syntax}
Continuing our example, $\lins(q_1)$ and $\lins(q_2)$ translate to $q_1'$
and $q_2'$ where:
\[\small\bl
q_1' =\bl
             \Select\,
                         \bl
                           x.\lName\,
                           \As\,\lab{i\consl{1}\lName}, \\
                           \rownumindex{x.\lName}\,\As\,\lab{i\consl{1}\lPeople} \\
                         \el \\
                       \From\,\tDepartments\,\As\,x \\ 
\el \\
q_2' = \bl
  (\bl
     \With\,q\,\As\,(\bl 
                       \Select\,
                         \bl
                           x.\lName\,              \As\,\lab{i\consl{1}\lName}, \\
                           \rownumindex{x.\lName}\,\As\,\lab{i2} \\
                         \el \\
                       \From\,\tDepartments\,\As\,x) \\ 
                     \el \\
     \Select\,
       \bl
         a\,        \As\,\lab{i\consl{1}1},
         z.\lab{i2}\,\As\,\lab{i\consl{1}2},
         y.\lName\, \As\,\lab{i\consl{2}\lName},
         b\,        \As\,\lab{i\consl{2}\consl{\lTasks}1}, \\
         \rownum()\,\Over ~(\orderby\,z.\lab{i\consl{1}\lName},
         z.\lab{i2}, \\
\hfill y.\dpt, y.\emp, y.\lSalary)~ \As\,\lab{i\consl{2}\consl{\lTasks}2} \\
       \el \\
     \From\,\tEmployees\,\As\,y,\, q\,\As\,z \\
     \Where\,(z.\lab{i\consl{1}\lName} = y.\dpt ~\wedge \\
       \hfill     (y.\lSalary < 1000 \vee y.\lSalary > 1000000))) \\
  \el \\
\Unionall \\
  (\bl
     \With\,q\,\As\,(\bl 
                       \Select\,
                         \bl
                           x.\lName\,              \As\,\lab{i\consl{1}\lName}, \\
                           \rownumindex{x.\lName}\,\As\,\lab{i2} \\
                         \el \\
                       \From\,\tDepartments\,\As\,x) \\ 
                     \el \\
     \Select\,
       \bl
         a\,        \As\,\lab{i\consl{1}1},
         z.\lab{i2}\,\As\,\lab{i\consl{1}2},
         y.\lName\, \As\,\lab{i\consl{2}\lName},
         d\,        \As\,\lab{i\consl{2}\consl{\lTasks}1}, \\
         \rownum()\,\Over ~(\orderby\,z.\lab{i\consl{1}\lName},
         z.\lab{i2}, \\
\hfill y.\dpt, y.\lName, y.\lClient)~ \As\,\lab{i\consl{2}\consl{\lTasks}2} \\
       \el \\
     \From\,\tContacts\,\As\,y,\, q\,\As\,z \\
     \Where\,(z.\lab{i\consl{1}\lName} = y.\dpt \wedge y.\lClient)) \\
   \el \\
\el
\el
\]
%
Modulo record flattening, the above fragment of SQL is almost
isomorphic to the image of the let-insertion translation. The only
significant difference is the use of $\rownum$ in place of
$\cindex$. Each instance of $\cindex$ in the body $R$ of a subquery of
the form $ \For\,(\ora{x \gets t}\;\Where\,X)\,\ret{R} $ is simulated
by a term of the form $\rownumindex{\ora{x.\cl}}$, where:
\[\small
\bl x_i :
\record{\cl_{i,1}:A_{i,1}, \dots, \cl_{i,m_i}:A_{i,m_i}}\\
\ora{x.\cl} = 
\bl
 x_1.\cl_{1,1}, \dots, x_1.\cl_{1,m_1}, 
    ~~\dots,~~ 
   x_n.\cl_{n,1}, \dots, x_n.\cl_{n,m_n} \\
 \el
\el
\]

A possible concern is that $\rownum$ is non-deterministic. It computes
row numbers ordered by the supplied columns, but if there is a tie,
then it is free to order the equivalent rows in any order. However,
we always order by \emph{all} columns
of all tables referenced from the current subquery, so our use of
$\rownum$ is always deterministic.
(An alternative could be to use
nonstandard features such as PostgreSQL's OID or MySQL's ROWNUM, but
sorting would still be necessary to ensure consistency across inner
and outer queries.)


\section{Experimental evaluation}
\label{sec:implementation}

Ferry's loop-lifting translation has been implemented in Links
previously by Ulrich, a member of the Ferry
team~\cite{ulrich11ferry-links}, following the approach described by
Grust et al.~\cite{grust10pvldb} to generate SQL:1999 algebra plans,
then calling Pathfinder~\cite{pathfinder} to optimise and evaluating
the resulting SQL on a PostgreSQL database.  We have also implemented
query shredding in Links, running against PostgreSQL; our
implementation\footnote{\small\url{http://github.com/slindley/links/tree/shredding}}
does not use Pathfinder.  We performed initial experiments with a
larger ad hoc query benchmark, and developed some optimisations,
including inlining certain \verb|WITH| clauses to unblock rewrites,
using keys for row numbering, and implementing stitching in one pass
to avoid construction of intermediate in-memory data structures that
are only used once and then discarded.  We report on shredding with
all of these optimisations enabled.

 \begin{figure}[tb]
   \centering
\small
\begin{verbatim}
QF1: SELECT e.emp FROM employees e 
     WHERE e.salary > 10000
QF2: SELECT e.emp, t.tsk
     FROM employees e, tasks t 
     WHERE e.emp = t.emp
QF3: SELECT e1.emp, e2.emp 
     FROM employees e1, employees e2 
     WHERE e1.dpt = e2.dpt
       AND e1.salary = e2.salary 
       AND e1.emp <> e2.emp
QF4: (SELECT t.emp FROM tasks t 
      WHERE t.tsk = 'abstract')
     UNION ALL (SELECT e.emp FROM employees 
                WHERE e.salary > 50000)
QF5: (SELECT t.emp FROM tasks t 
      WHERE t.tsk = 'abstract')
     MINUS
     (SELECT e.emp FROM employees e 
      WHERE e.salary > 50000)
QF6: ((SELECT t.emp FROM tasks t 
       WHERE t.tsk = 'abstract')
      UNION ALL (SELECT e.emp FROM employees e 
                 WHERE e.salary > 50000))
     MINUS 
     ((SELECT t.emp FROM tasks t 
       WHERE t.tsk = 'enthuse')
      UNION ALL (SELECT e.emp FROM employees e 
                 WHERE e.salary > 10000))
\end{verbatim}
\caption{SQL versions of flat queries used in
  experiments}\label{fig:flat-benchmark}
\end{figure}
\begin{figure}[tb]
\[\small\ba{@{}l@{~}l@{}}
\mathtt{Q1}:& \bl
  \For\,(d \gets \tDepartments) \\
  \retl\recordl
    \bl
      \lName      = d.\lName, \\
      \lEmployees = \employeesOfDept~d,\\
      \lContacts = \contactsOfDept~d\recordr
      \el
\el
\smallskip\\
\mathtt{Q2:} &
\bl
\For\, (d \gets \mathtt{Q1})\\
    \Where\, (all~d.\lEmployees~(\lambda x. contains~x.\lTasks~\cStr{abstract}))\\
    \ret{\record{\dpt=d.\dpt}}
\el
\smallskip\\
\mathtt{Q3:}&\bl
\For\,(e \gets \tEmployees) \\
    \ret{\record{\lName=e.\lName,\tsk=\tasksOfEmp(e)}}
\el
\smallskip\\
\mathtt{Q4:}&\bl
\For\,(d \gets \tDepartments) \\
\ret{\recordl \dpt=d.\dpt, \lEmployees= \bl 
  \For\, (e \gets \tEmployees) \\
  \Where\, (d.\dpt = e.\dpt) \\
                \ret{e.\emp} \recordr 
\el}
\el
\smallskip\\
\mathtt{Q5:}&\bl
\For\,(t \gets \tTasks) \, \ret{\record{a=t.\tsk, b=\employeesByTask~t}}
\el
\smallskip\\
\mathtt{Q6:}&\bl
\For\, (d \gets \mathtt{Q1}) \\
  \retl \recordl
     \bl
       \lDepartment = d.\lName, \\
       \lPeople = \\
       \bl
       \quad \getTasks(\outliers(d.\lEmployees))~(\lambda y.\,y.\lTasks)   \\
       \union~
           \getTasks(\clients(d.\lContacts))~(\lambda y.\,\ret{\cStr{buy}})\recordr \retr\\
       \el \\
     \el
\el
\ea
\]%
   \caption{Nested queries used in experiments}
   \label{fig:nested-benchmark}
 \end{figure}

\paragraph*{Benchmark queries} 
\begin{sloppypar}
There is no standard benchmark for queries returning nested results.
In particular, the popular TPC-H benchmark is not suitable for
comparing shredding and loop-lifting: the TPC-H queries do not return
nested results, and can be run directly on any SQL-compliant RDBMS, so
neither approach needs to do any work to produce an SQL query.
\end{sloppypar}

We selected twelve queries over the organisation schema described in
Section~\ref{sec:example} to use as a benchmark.  The first six queries,
named \verb|QF1|--\verb|QF6|, return flat results and can be
translated to SQL using existing techniques, without requiring either
shredding or loop-lifting.  We considered these queries as a sanity
check and in order to measure the overhead introduced by loop-lifting
and shredding. Figure~\ref{fig:flat-benchmark} shows the SQL versions
of these queries.

We also selected six queries \verb|Q1|--\verb|Q6| that do involve
nesting, either within query results or in intermediate stages.  They
are shown in Figure~\ref{fig:nested-benchmark}; they use the auxiliary
functions defined in Section~\ref{sec:example}.
\verb$Q1$ is the query $\Qorg$ that builds the nested
organisation view from Section~\ref{sec:example}. \verb$Q2$ is a flat
query that computes a flat result from \verb$Q1$ consisting of all
departments in which all employees can do the ``abstract'' task; it is
a typical example of a query that employs higher-order
functions. \verb$Q3$ returns records containing each employee and the
list of tasks they can perform. \verb$Q4$ returns records containing
departments and the set of employees in each department. \verb$Q5$
returns a record of tasks paired up with sets of employees and their
departments. \verb$Q6$ is the outliers query $\Qoutliers$ introduced in
Section~\ref{sec:example}.

\begin{figure}[tb]
\centering
  \includegraphics[scale=0.5]{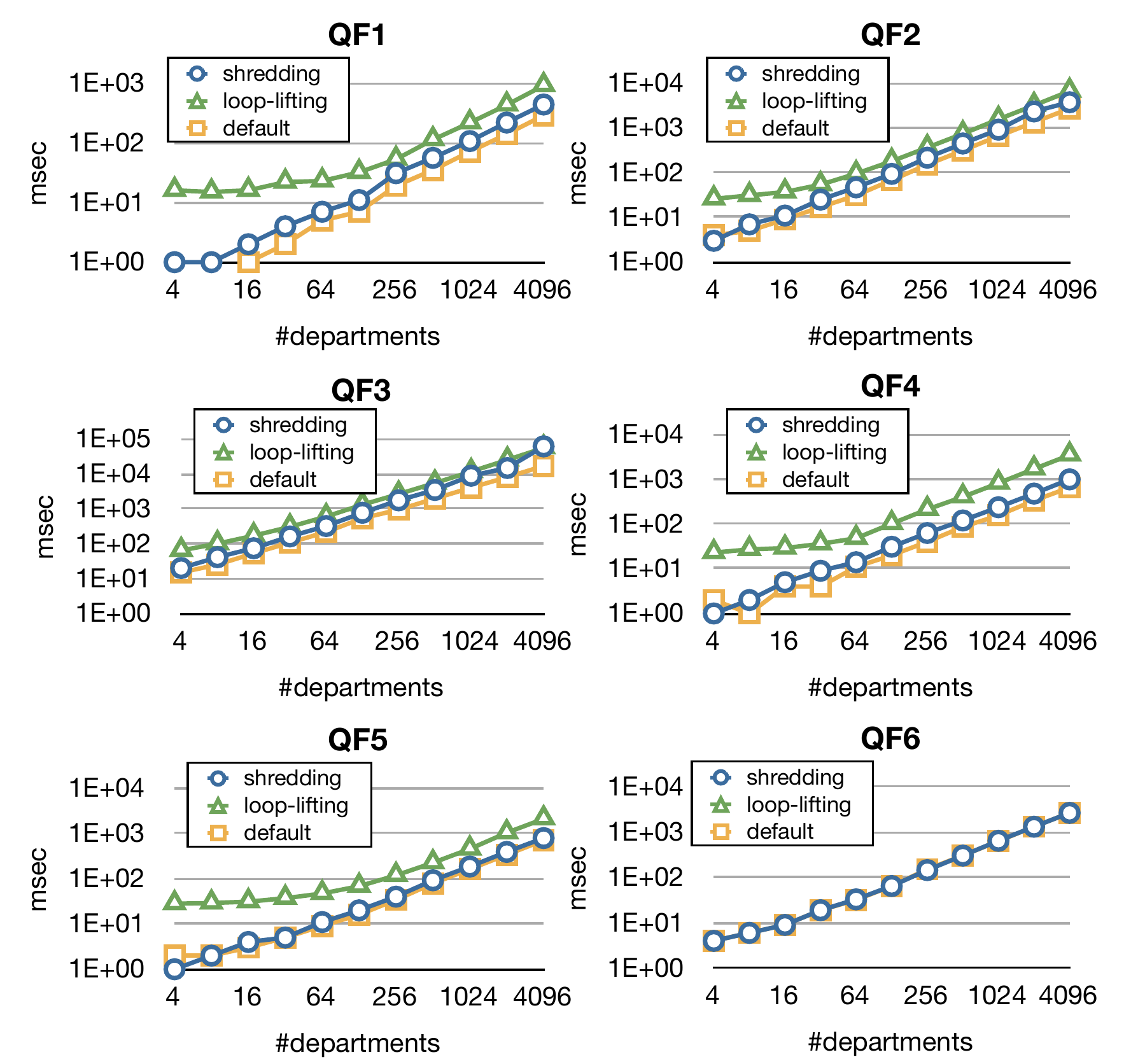}
  \caption{Experimental results (flat queries)}
  \label{fig:experiments-flat}
\end{figure}

\begin{figure}[tb]
  \centering
  \includegraphics[scale=0.48]{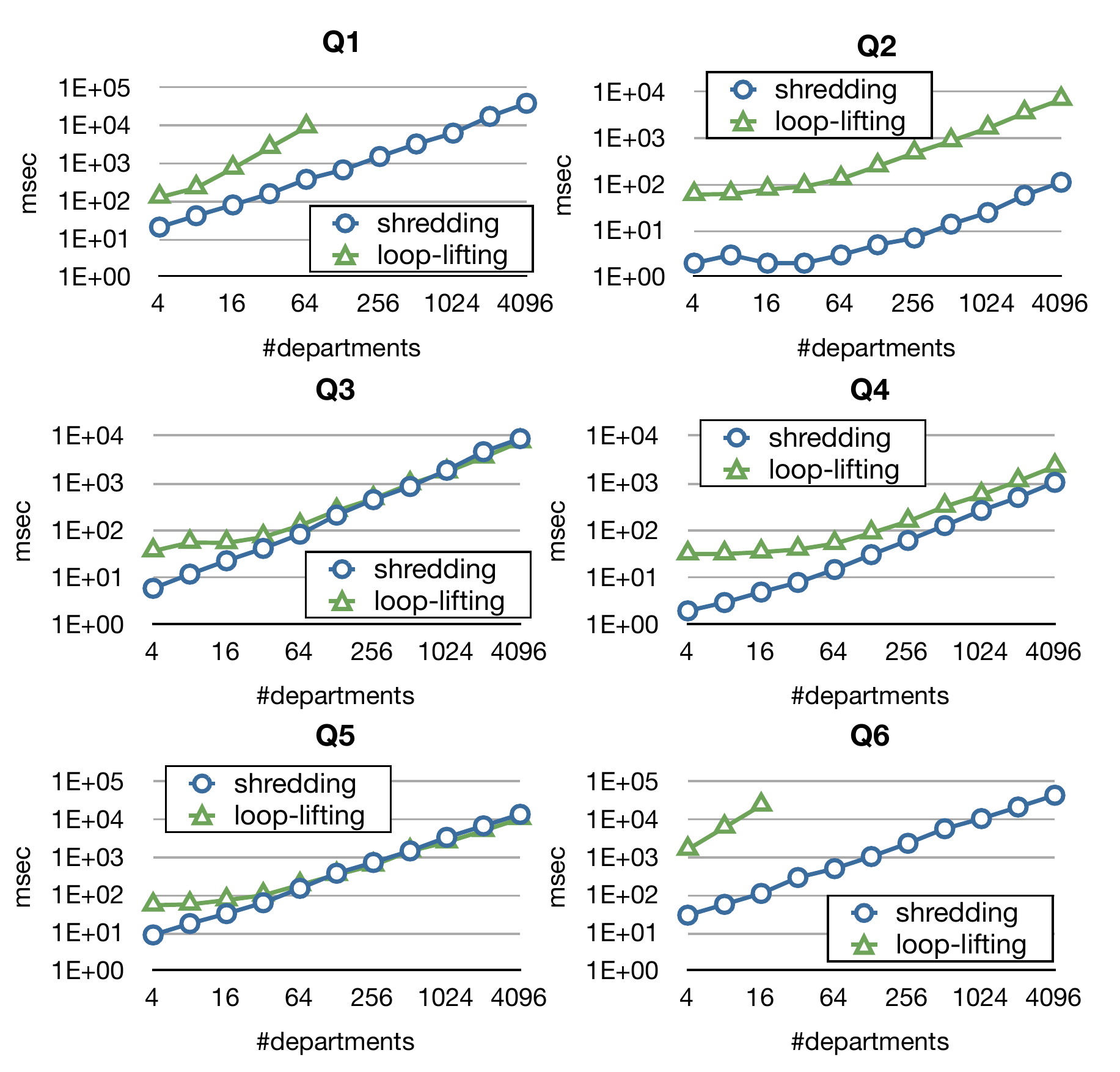}
  \caption{Experimental results (nested queries)}
  \label{fig:experiments-nested}
\end{figure}
\paragraph*{Experimental results}

We measured the query execution time for all queries on randomly
generated data, where we vary the number of departments in the
organisation from 4 to 4096 (by powers of 2). Each department has on
average 100 employees and each employee has 0--2 tasks, and the
largest (4096 department) database was approximately 500MB. Although the test
datasets are moderate in size, they suffice to illustrate the asymptotic
trends in the comparative performance of the different techniques.
All tests were
performed using PostgreSQL 9.2 running on a MacBook Pro with 4-core
2.6GHz CPU, 8GB RAM and 500GB SSD storage, with the database server
running on the same machine (hence, negligible network latency).  We
measure \emph{total} time to translate a nested query to SQL, evaluate
the resulting SQL queries, and stitch the results together to form a
nested value, measured from Links.

We evaluated each query using query shredding and loop-lifting, and
for the flat queries we also measured the time for Links' default
(flat) query
evaluation. The experimental results for the flat queries are shown in
Figure~\ref{fig:experiments-flat} and for the nested queries in
Figure~\ref{fig:experiments-nested}.  Note that both axes are
logarithmic.  All times are in milliseconds; the times are medians of 5 runs.  
The times for small data sizes provide a
comparison of the overhead associated with shredding, loop-lifting or
Links' default query normalisation algorithm.

\paragraph*{Discussion} 
 We should re-emphasise that Ferry (and
Ulrich's loop-lifting implementation for Links) supports grouping and
aggregation features that are not handled by our translation.  We
focused on queries that both systems can handle, but Ferry has a clear
advantage for grouping and aggregation queries or when ordering is
important (e.g.~sorting or top-$k$ queries).  Ferry is based on list
semantics, while our approach handles multiset semantics.  So, some
performance differences may be due to our exploitation of
multiset-based optimisations that Ferry (by design) does not exploit.

The results for flat queries show that shredding has low per-query
overhead in most cases compared to Links' default flat query
evaluation, but the queries it generates are slightly
slower. Loop-lifting, on the other hand, has a noticeable per-query
overhead, likely due to its invocation of Pathfinder and associated
serialisation costs.  In some cases, such as \verb|QF4| and
\verb|QF5|, loop-lifting is noticeably slower asymptotically; this
appears to be due to additional sorting needed to maintain list
semantics.  We encountered a bug that prevented loop-lifting from
running on query \verb|QF6|; however, shredding had negligible
overhead for this query.  In any case, the overhead of either
shredding or loop-lifting for flat queries is irrelevant: we can
simply evaluate such queries using Links' default flat query
evaluator.  Nevertheless, these results show that the overhead of
shredding for such queries is not large, suggesting that the queries
it generates are similar to those currently generated by Links.
(Manual inspection of the generated queries confirms this.)

The results for nested queries are more mixed.  In most cases, the
overhead of loop-lifting is dominant on small data sizes, which again
suggests that shredding may be preferable for OLTP or Web workloads
involving rapid, small queries.  Loop-lifting scales poorly on two
queries (\verb|Q1| and \verb|Q6|), and did not finish within 1 minute
even for small data sizes.  Both \verb|Q1| and \verb|Q6| involve 3 levels of nesting
and in the innermost query, loop-lifting generates queries with
Cartesian products inside OLAP operators such as \verb|DENSE_RANK| or
\verb|ROW_NUMBER| that Pathfinder was not able to remove.  The queries
generated by shredding in these cases avoid this pathological
behaviour.  For other queries, such as \verb|Q2| and \verb|Q4|,
loop-lifting performs better but is still slower than shredding as
data size increases. Comparison of the queries generated by
loop-lifting and shredding reveals that loop-lifting encountered similar problems with
hard-to-optimise OLAP operations.  Finally, for \verb|Q3| and
\verb|Q5|, shredding is initially faster (due to the overhead of
loop-lifting and calling Pathfinder) but as data size increases,
loop-lifting wins out.  Inspection of these generated queries reveals
that the queries themselves are similar, but the shredded queries
involve more data movement.  Also, loop-lifting returns sorted
results, so it avoids in-memory hashing or sorting while constructing
the nested result.  It should be possible to incorporate similar
optimisations into shredding to obtain comparable performance.

Our experiments show that shredding performs similarly or
better than loop-lifting on our (synthetic) benchmark queries on moderate (up to
500MB) databases.  Further work may need to be done to investigate
scalability to larger databases or consider more realistic query benchmarks.



\section{Related and future work}
\label{sec:related-work}

We have discussed related work on nested queries, Links, Ferry and
LINQ in the introduction.  Besides Cooper~\cite{cooper09dbpl}, several
authors have recently considered higher-order query
languages. Benedikt et al.~\cite{benedikt10pods,vu11icdt} study higher-order queries over flat
relations. The Links approach was adapted to LINQ in F\# by Cheney et
al.~\cite{cheney13icfp}.  Higher-order features are also being added
to XQuery 3.0~\cite{xquery3}.

Research on shredding XML data into relations and evaluating XML
queries over such representations~\cite{xmltosql} is superficially
similar to our work in using various indexing or numbering schemes to
handle nested data. Grust et al.'s work on translating XQuery to SQL
via Pathfinder~\cite{pathfinder} is a mature solution to this problem,
and Grust et al.~\cite{grust07icde} discuss optimisations in the
presence of unordered data processing in XQuery. However, XQuery's ordered
data model and treatment of node identity would block
transformations in our algorithm that assume unordered, pure operations.

We can now give a more detailed comparison of our approach with the
indexing strategies in Van den Bussche's work and in Ferry. Van den
Bussche's approach uses natural indexes (that is, $n$-tuples of ids),
but does not preserve multiset semantics. Our approach preserves
multiplicity and can use natural indexes, we also propose a flat indexing scheme based on
$\rownum$. In Ferry's indexing scheme, the surrogate indexes only link
adjacent nesting levels, whereas our indexes take information at all
higher levels into account. Our flat indexing scheme relies on this
property, and Ferry's approach does not seem to be an instance of ours
(or vice versa). Ferry can generate multiple SQL:1999 operations and
Pathfinder tries to merge them but cannot always do so. Our approach
generates $\rownum$ operations only at the end, and does not rely on
Pathfinder.  Finally, our approach uses normalisation and tags parts
of the query to disambiguate branches of
unions.

Loop-lifting has been implemented in Links by
Ulrich~\cite{ulrich11ferry-links}, and Grust and
Ulrich~\cite{grust13dbpl} recently presented techniques for supporting
higher-order functions as query results. By using Ferry's
loop-lifting translation and Pathfinder, Ulrich's system also supports
list semantics and aggregation and grouping operations; to our
knowledge, it is an open problem to either prove their correctness or
adapt these techniques to fit our approach. Ferry's approach supports
a list-based semantics for queries, while we assume a bag-based
semantics (matching SQL's default behaviour). Either approach can
accommodate set-based semantics simply by eliminating duplicates in
the final result. In fact, however, we believe the core query
shredding translation
(Sections~\ref{sec:shredding}--\ref{sec:indexing-schemes}) works just
as well for a list semantics. The only parts that rely on unordered
semantics are normalisation (Section~\ref{subsec:normalisation}) and
conversion to SQL (Section~\ref{sec:to-sql}). We leave these
extensions to future work.

Our work is also partly inspired by work on unnesting for nested data
parallelism.  Blelloch and Sabot~\cite{blelloch90jpdc} give a
compilation scheme for NESL, a data-parallel language with nested
lists; Suciu and Tannen~\cite{suciu94spaa} give an alternative scheme
for a nested list calculus.  This work may provide an alternative (and
parallelisable) implementation strategy for Ferry's list-based
semantics~\cite{grust10pvldb}.



\section{Conclusion}
\label{sec:conclusion}

Combining efficient database access with high-level programming
abstractions is challenging in part because of the limitations of flat
database queries.  Currently, programmers must write flat queries and
manually convert the results to a nested form.  This damages
declarativity and maintainability.  Query shredding can help to bridge
this gap.  Although it is known from prior work that query shredding
is possible in principle, and some implementations (such as Ferry)
support this, getting the details right is tricky, and can lead to
queries that are not optimised well by the relational engine.  Our
contribution is an alternative shredding translation that handles
queries over bags and delays use of OLAP operations until the final
stage.  Our translation compares favourably to loop-lifting in
performance, and should be straightforward to extend and to
incorporate into other language-integrated query systems.

\paragraph*{Acknowledgements}
We thank Ezra Cooper, Torsten Grust, and Jan Van den Bussche for
comments.  We are very grateful to Alex Ulrich for sharing his
implementation of Ferry in Links and for extensive discussions.
This work was supported by EPSRC
grants EP/J014591/1 and EP/K034413/1 and a Google Research Award.


%
{\small
\bibliographystyle{abbrv} 
\bibliography{shredding}
}
\newpage
\appendix

\normalsize
\section{Behaviour of Van den Bussche's simulation on multisets}
\label{app:vdb-incorrect}
\newcommand{\adom}{\textbf{adom}}

\newcommand{\Wrong}{\mathit{Wrong}}

As noted in the Introduction, Van den Bussche's simulation of nested
queries via flat queries does not work properly over multisets.  To
illustrate the problem, consider a simple query $R \cup S$, where $R$
and $S$ have the same schema $\bagtl \record{A:\Int,B:\bagtl{\Int}}$.
Suppose $R$ and $S$ have the following values:
\[
R = \begin{array}{|cc|}
\hline
A & B\\\hline
       1 &\{1\}\\
	2 &\{2\}\\
\hline
\end{array}
\quad 
S = \begin{array}{|cc|}
\hline
A & B\\\hline
       1 &\{3,4\}\\
	2 &\{2\}\\
\hline
\end{array}
\]
then their multiset union is:
\[       
R \cup S = \begin{array}{|cc|}
\hline
A & B\\\hline
       1 &\{1\}\\
	1 &\{3,4\}\\
	2 &\{2\}\\
	2 &\{2\}\\
\hline
\end{array}
\]
Van den Bussche's simulation (like ours) represents these nested
values by queries over flat tables, such as the following:
\[
R_1 = \begin{array}{|cc|}
\hline
     A& id \\
\hline
     1 & a       \\
     2 & b       \\
\hline
\end{array}
\quad
R_2 = 
\begin{array}{|cc|}
\hline
     id& B  \\
\hline
     a & 1       \\
     b &2       \\
\hline
\end{array}
\]
\[
S_1 = \begin{array}{|cc|}
\hline
     A&  id \\
\hline
     1 & a       \\
     2 & b       \\
\hline
\end{array}
\quad
S_2 = 
\begin{array}{|cc|}
\hline
     id& B  \\
\hline
     a & 3       \\
     a & 4       \\
    b &2       \\
\hline
\end{array}
\]
where $a,b$ are arbitrary distinct ids. Note however that $R$ and $S$
have overlapping ids, so if we simply take the union of $R_1$ and
$S_1$, and of $R_2$ and $S_2$ respectively, we will get:
\[\Wrong_1 = \begin{array}{|cc|}
\hline
     A& id \\
\hline
     1 & a       \\
     2 & b       \\
     1 & a       \\
     2 & b       \\
\hline
\end{array}
\quad
\Wrong_2 = 
\begin{array}{|cc|}
\hline
     id& B  \\
\hline
    a & 1       \\
    b &2       \\
     a & 3       \\
     a & 4       \\
    b &2       \\
\hline
\end{array}
\]
corresponding to nested value:
\[\Wrong = 
\begin{array}{|cc|}
\hline
     A& B  \\
\hline
     1 &\{1,3,4\}       \\
     1 &\{1,3,4\}       \\
     2 &\{2,2\}\\
     2 &\{2,2\}\\
\hline
\end{array}
\]
Instead, both Van den Bussche's simulation and our approach avoid clashes among ids
when taking unions.  Van den Bussche's simulation does this by adding
two new id fields to represent the result of a union, say $id_1$ and $id_2$.
 The tuples originating from $R$ will have \emph{equal} $id_1$ 
and $id_2$ values, while those originating from $S$ will have \emph{different} $id_1$ and 
$id_2$ values.  In order to do this, the simulation defines two queries, one for each 
result table.  The first query is of the form:
\begin{eqnarray*}
T_1 &=& (R_1 \times {(id_1:x,id_2:x) \mid x \in \adom} )\\
   &\cup& (S_1 \times {(id_1:x,id_2:x') \mid x \neq x' \in \adom})
\end{eqnarray*}
and similarly 
\begin{eqnarray*}
T_2 &=& (R_2 \times {(id_1:x,id_2:x) \mid x \in \adom} )\\
   &\cup& (S_2 \times {(id_1:x,id_2:x') \mid x \neq x' \in \adom})
\end{eqnarray*}
where $\adom$ is the active domain of the database (in this case,
$\adom = \{1,2,3,4,a,b\}$) --- this is of course also definable as a query.
Thus, in this example, the result is of the form
\[
T_1 = 
\begin{array}{|cccc|}
\hline
     A& id  & id_1 & id_2\\
\hline
     1 & a & x & x       \\
     2 & b & y & y \\
     1 & a & z  & z'       \\
    2 & b  & v & v'\\
\hline
\end{array}
\qquad
T_1 = 
\begin{array}{|cccc|}
\hline
     id  & id_1 & id_2 & B\\
\hline
      a & x & x    & 1   \\
      b & y & y  & 2\\
      a & z  & z'   & 3    \\
      a & w  & w'    & 4   \\
    b  & v & v' & 2\\
\hline
\end{array}
\]
where $x,y$ are any elements of $\adom$ (rows mentioning $x$ and $y$ stand for 6 
instances) and $z\neq z'$, $w \neq w'$ and $v \neq v'$ are any pairs of distinct 
elements of $\adom$ (rows mentioning these variables stand for 30 instances of 
distinct pairs from $\adom$).  This leads to an $O(|\adom| * |R| + 
|\adom|^2 * |S|)$ blowup in the number of tuples.  Specifically, for
our example, $|T_1| = 72$, whereas the actual number
of tuples in a natural representation of $R \cup S$ is only 9.
In a set semantics,
the set value simulated by these tables is correct even with all of
the extra tuples; however, for a
multiset semantics, this quadratic blowup is not correct ---
even for our example, with $|\adom| = 6$, the result of evaluating $R
\cup S$ yields a different number of tuples from the result of
evaluating $S \cup R$, and neither represents the correct multiset in
an obvious way.  It may be possible (given knowledge of the query and
active domain, but not the source database) to ``decode'' the flat query
results and obtain the correct nested result, but doing so appears no
easier than developing an alternative, direct algorithm.

\section{Typing rules}
\label{app:typing-rules}

The (standard) typing rules for \langname queries are shown in Figure~\ref{fig:typing-rules}.
The typing rules for shredded terms, packages, and values are shown in
Figures~\ref{fig:shredded-typing-rules},
\ref{fig:shredded-package-typing-rules}, \ref{fig:indexed-nested-value-typing-rules} and \ref{fig:shredded-value-typing-rules}.

\begin{figure*}[p!]

\begin{mathpar}
\inferrule[Var]
{ }
{\Gamma, x:A \vdash x : A}

\inferrule[Constant]
{\Sigma(c) = \tuple{O_1,\dots,O_n} \to O' \\
 [\Gamma \vdash M_i : O_i]\iton}
{\Gamma \vdash c(M_1,\dots,M_n) : O'}

\inferrule[Lam]
{\Gamma, x:A \vdash M : B}
{\Gamma \vdash \lam{x} M : A \to B}

\inferrule[App]
{\Gamma \vdash M : A \to B \\ \Gamma \vdash N : A}
{\Gamma \vdash M~ N : B}

\inferrule[Record]
{[\Gamma \vdash M_i : A_i]\iton}
{\Gamma \vdash \record{\cl_i=M_i}\iton : \record{\cl_i=A_i}\iton}

\inferrule[Project]
{\Gamma \vdash M : \record{\cl_i:A_i}\iton}
{\Gamma \vdash M.\cl_j : A_j}

\inferrule[If]
{\Gamma \vdash M : \Bool \\
 \Gamma \vdash N : A \\
 \Gamma \vdash N' : A}
{\Gamma \vdash \If\, M~ \Then~ N~ \Else~ N' : A}

\inferrule[Empty]
{ }
{\Gamma \vdash \emptybag : \bagt{A}}

\inferrule[Singleton]
{\Gamma \vdash M : A}
{\Gamma \vdash \ret{M} : \bagt{A}}

\inferrule[Union]
{\Gamma \vdash M : \bagt{A} \\
 \Gamma \vdash N : \bagt{A}}
{\Gamma \vdash M \union N : \bagt{A}}

\inferrule[For]
{\Gamma \vdash M : \bagt{A} \\ \Gamma, x:A \vdash N : \bagt{B}}
{\Gamma \vdash \For\,(x \gets M)\,N : \bagt{B}}

\inferrule[IsEmpty]
{\Gamma \vdash M : \bagt{A}}
{\Gamma \vdash \isEmpty{M} : \Bool}

\inferrule[Table]
{\Sigma(t)  =\bagt{\record{\ora{\cl:O}}}}
{\Gamma \vdash \Table~t : \bagt{\record{\ora{\cl:O}}}}
\end{mathpar}

\caption{Typing rules for higher-order nested queries\label{fig:typing-rules}}

\begin{mathpar}
\inferrule[Var]
{ }
{\Gamma, x:\record{\ora{\cl:F}} \vdash x : \record{\ora{\cl:F}}}

\inferrule[Constant]
{\Sigma(c) = \tuple{O_1,\dots,O_n} \to O' \\
 [\Gamma \vdash X_i : O_i]\iton}
{\Gamma \vdash c(X_1,\dots,X_n) : O'}

\inferrule[Record]
{[\Gamma \vdash N_i : A_i]\iton}
{\Gamma \vdash \record{\cl_i=N_i}\iton : \record{\cl_i=A_i}\iton}

\inferrule[Project]
{\Gamma \vdash x : \record{\cl_i:A_i}\iton}
{\Gamma \vdash x.\cl_j : A_j}


\inferrule[Index]
{ }
{\Gamma \vdash \Ind{a}{d} : \Index}

\inferrule[Singleton]
{\Gamma \vdash I : \Index \\ \Gamma \vdash N : F}
{\Gamma \vdash \retann{\tuple{I,N}}{a} : \bagt{\tuple{\Index, F}}}

\inferrule[Union]
{[\Gamma \vdash C_i : \bagt{\tuple{\Index, F}}]\iton}
{\Gamma \vdash \bigunion\,\vec{C} : \bagt{\tuple{\Index, F}}}

\inferrule[For]
{[\Sigma(t_i) = \bagt{A_i}]\iton \\
 \Gamma, [x_i:A_i]\iton \vdash X : \Bool \\
 \Gamma, [x_i:A_i]\iton \vdash C : \bagt{\tuple{\Index, F}}}
{\Gamma \vdash \For\,([x_i \gets t_i]\iton\,\Where\,X)\,
                                       C : \bagt{\tuple{\Index, F}}}

\inferrule[IsEmpty]
{\Gamma \vdash L : \bagt{\tuple{\Index, F}}}
{\Gamma \vdash \isEmpty{L} : \Bool}

\end{mathpar}

\caption{Typing rules for shredded terms
\label{fig:shredded-typing-rules}}

\begin{mathpar}
\inferrule[Base]
{ }
{\vdash O(\sterm) : O(\stype)}

\inferrule[Record]
{[ \vdash \sht{A_i}(\sterm) : \sht{A'_i}(\stype)]\iton}
{\vdash \record{\cl_i:\sht{A_i}(\sterm)}\iton : \record{\cl_i:\sht{A'_i}(\stype)}\iton}

\inferrule[Bag]
{\vdash \sht{A}(\sterm) : \sht{A'}(\stype) \\ \vdash L : A \\ \erase(\sht{A}(\sterm)) = A}
{\vdash \shann{\bagt{\sht{A}(\sterm)}}{L} : \shann{\bagt{\sht{A'}(\stype)}}{A}}
\end{mathpar}
 We write $\sterm$ for the set of shredded terms, and
$\stype$ for the set of shredded types.
\caption{Typing rules for shredded packages\label{fig:shredded-package-typing-rules}}

\end{figure*}


\section{Query normalisation}
\label{app:sn}

\newcommand{\Split}[2]{\llparenthesis{#1}\rrparenthesis_{#2}}
\newcommand{\splitbag}[2]{\mathcal{B}\llparenthesis{#1}\rrparenthesis_{#2}^\star}
\newcommand{\splitr}[2]{\mathcal{R}\llparenthesis{#1}\rrparenthesis_{#2}}
\newcommand{\splitfield}[2]{\mathcal{F}\llparenthesis{#1}\rrparenthesis_{#2}}

Following previous work~\cite{lindley12tldi}, we separate query
normalisation into a rewriting phase and a type-directed structurally
recursive function. Where we diverge from this previous work is that we
extend the rewrite relation to hoist all conditionals up to the
nearest comprehension (in order to simplify the rest of the
development), and the structurally recursive function is generalised
to handle nested data. Thus normalisation can be divided into three
stages.

\begin{compactitem}
\item The first stage performs symbolic evaluation, that is,
  $\beta$-reduction and commuting conversions, flattening unnecessary
  nesting and eliminating higher-order functions.

\item The second stage hoists all conditionals up to the nearest
  enclosing comprehension in order to allow them to be converted to where
  clauses.

\item The third and final stage hoists all unions up to the top-level,
  $\eta$-expands tables and variables, and turns all conditionals into
  where clauses.
\end{compactitem}

\paragraph*{Weak normalisation and strong normalisation}

Given a term $M$ and a rewrite relation $\rewriteto$, we write $M
\nrewriteto$ if $M$ is irreducible, that is no rewrite rules in
$\rewriteto$ apply to $M$. The term $M$ is said to be in \emph{normal
  form}.

A term $M$ is \emph{weakly normalising} with respect to a rewrite
relation $\rewriteto_r$, or \emph{$r$-WN}, if there exists a finite
reduction sequence
\[
M \rewriteto M_1 \rewriteto \dots \rewriteto M_n \nrewriteto
\]
A term $M$ is \emph{strongly normalising} with respect to a rewrite
relation $\rewriteto_r$, or \emph{$r$-SN}, if every reduction sequence
starting from $M$ is finite. If $M$ is $r$-SN, then we write
$\maxr_r(M)$ for the maximum length of a reduction sequence starting
from $M$.

A rewrite relation $\rewriteto_r$ is \emph{weakly-normalising}, or
\emph{WN}, if all terms $M$ are weakly-normalising with respect to
$\rewriteto_r$.  Similarly, a rewrite relation $\rewriteto_r$ is
\emph{strongly-normalising}, or \emph{SN}, if all terms $M$ are
strongly-normalising with respect to $\rewriteto_r$.

For any $\rewriteto_r$ we can define a partial function $\nf_r$ such
that $\nf_r(M)$ is a normal form of $M$, if one exists, otherwise
$\nf_r(M)$ is undefined.  Note that there may be more than one normal
form for a given $M$; $\nf_r$ chooses a particular normal form for
each $M$ that has one.  If $\leadsto_r$ is weakly normalising, then
$\nf_r$ is a total function.

We now describe each normalisation stage in turn.

\subsection{Symbolic evaluation}

The $\beta$-rules perform symbolic evaluation, including substituting
argument values for function parameters, record field projection,
conditionals where the test is known to be true or false, or iteration
over singleton bags.
\begin{equations}
(\lam{x}N)~ M &\rewriteto_c&  N[x := M] \\
\record{\ora{\cl=M}}.\cl_i &\rewriteto_c& M_i \\
\If\, \True\, \Then\, M\, \Else\,N &\rewriteto_c& M \\
\If\, \False\, \Then\, M\, \Else\,N &\rewriteto_c& N \\
\For\,(x \gets \ret{M})\, N &\rewriteto_c& N[x := M] \\
\end{equations}
Fundamentally, $\beta$-rules always follow the same pattern. Each is
associated with a particular type constructor $T$, and the left-hand
side always consists of an \emph{introduction} form for $T$ inside an
\emph{elimination} form for $T$. For instance, in the case of
functions (the first $\beta$-rule above), the introduction form is a
lambda and the elimination form is an application. Applying a
$\beta$-rule \emph{eliminates} $T$ in the sense that the introduction
form from the left-hand side either no longer appears or has been
replaced by a term of a simpler type on the right-hand side.


Each instance of a $\beta$-rule is associated with an
\emph{elimination frame}. An \emph{elimination frame} is simply the
elimination form with a designated \emph{hole} $[~]$ that can be
\emph{plugged} with another expression. For instance, elimination
frames for the function $\beta$-rule are of the form $E[~] =
[~]~M$. If we plug an introduction form $\lambda x . N$ into
$E[~]$, written $E[\lambda x . N]$, then we obtain the left-hand side
of the associated $\beta$-rule $(\lambda x.N)~M$.

\begin{tronly}
  (This notion of an expression with a hole that can be filled by
  another expression is commonly used in rewriting and operational
  semantics for programming languages~\cite[ch. 19]{pierce02}; here,
  it is not essential but helps cut down the number of explicit rules,
  and helps highlight commonality between rules.)
\end{tronly}
The elimination frames of $\langname$ are as follows.
\begin{syntax}
&  E[~]  &  \cce & [~]~ M
            \mid [~].\cl
            \mid  \If\,[~]\, \Then\,M\, \Else\, N
            \mid \For\,(x \gets [~])\,N 
\end{syntax}%
The following rules express that comprehensions, conditionals, empty
bag constructors, and unions can always be \emph{hoisted} out of the
above elimination frames. In the literature such rules are often
called \emph{commuting conversions}. They are necessary in order to
expose all possible $\beta$-reductions. For instance,
\[(\If\,M\,\Then\,\record{\cl=N}\,\Else\,N').\cl\]
 cannot
$\beta$-reduce, but $\If\,M\,\Then\,\record{\cl=N}.\cl\,\Else\,N'.\cl$
can.

%
%
\begin{equations}
E[\For\,(x \gets M)\, N] &\rewriteto_c& \For\,(x \gets M)\,
E[N] \\ E[\If\,L\,\Then\,M\,\Else\,N] &\rewriteto_c&
\If\,L\,\Then\,E[M]\,\Else\,E[N] \\ E[\emptybag] &\rewriteto_c& \emptybag
\\ E[M_1 \union M_2] &\rewriteto_c& E[M_1] \union E[M_2] \\
\end{equations}%
For example:
\[(\If\,L\,\Then\,M_1\,\Else\,M_2 )~M \rewriteto_c
\If\,L\,\Then\,M_1~M\,\Else\,M_2 ~M
\]
\begin{sloppypar}
Note that some combinations of elimination frames and rewrite rule are
impossible in a well-typed term, such as $ \If\,\emptybag\, \Then\,M\,
\Else\, N$.
For the purposes of reduction we treat $\Empty$ like an uninterpreted
constant, that is, we do reduce inside emptiness tests, but they do
not in any other way interact with the reduction rules.
\end{sloppypar}

Next we prove that $\rewriteto_c$ is strongly normalising. 
\begin{tronly}
  The proof is based on a previous proof of strong normalisation for
  simply-typed $\lambda$-calculus with sums~\cite{lindley07tlca},
  which generalises the $\top\top$-lifting
  approach~\cite{lindley05tlca}, which in turn extends Tait's proof of
  strong normalisation for simply-typed $\lambda$-
  calculus~\cite{tait67intensional}.
\end{tronly}

\newcommand{\fsId}{\mathit{Id}}
\newcommand{\fsArrow}{\multimap}
\newcommand{\fsApp}[2]{{#1}[{#2}]}
\newcommand{\fsLength}[1]{|{#1}|}

\newcommand{\IFF}{\quad\Longleftrightarrow\quad}
\newcommand{\DEFF}{\quad\stackrel{{\mathit{def}}}{\Longleftrightarrow}\quad}

\newcommand{\hastype}{:}

\newcommand{\reducestoc}{\rewriteto_c}

\newcommand{\subst}[3]{{#1}[{#2}\mathbin{:{=}}{#3}]}
\newcommand{\substrange}[6]{{#1}[{#2}\mathbin{:{=}}#3#4#5\mathbin{:{=}}{#6}]}

\paragraph*{Frame stacks}

\begin{syntax}
(\text{frame stacks}) & S &\cce& \fsId \mid S \circ E \\[1ex]
(\text{stack length}) & \fsLength{\fsId}     &=& 0 \\
                      & \fsLength{S \circ E} &=& \fsLength{S} + 1 
\sepskip\\
(\text{plugging})     & \fsApp{\fsId}{M}       &=& M \\
                      & \fsApp{(S \circ E)}{M} &=& \fsApp{S}{(\fsApp{E}{M})} \\
\end{syntax}  

\begin{tronly}
  Following previous work~\cite{lindley07tlca} we assume variables
  are annotated with types.
\end{tronly}
We assume variables are annotated with types.
We write $A \fsArrow B$ for the type of frame
stack $S$, if $\fsApp{S}{M} \hastype B$ for all terms $M \hastype A$.

\paragraph*{Frame stack reduction}
\[
\begin{array}{@{}rcl@{}}
S \reducestoc S'
   &\DEFF& \forall M . \fsApp{S}{M} \reducestoc \fsApp{S'}{M} \\
   &\IFF& \fsApp{S}{x} \reducestoc \fsApp{S'}{x}
\end{array}
\]
\noindent
Frame stacks are closed under reduction. A frame stack $S$ is
\emph{$c$-strongly normalising}, or \emph{$c$-SN}, if all reduction
sequences starting from $S$ are finite.

\begin{lemma}
\label{lem:frame-reduction}
If $S \reducestoc S'$, for frame stacks $S, S'$, then $\fsLength{S'}
\leq \fsLength{S}$.
\end{lemma}

\begin{proof}
Induction on the structure of $S$.
\end{proof}

\paragraph*{Reducibility}  We define \emph{reducibility} as follows:
~
\begin{compactitem}
\item $\fsId$ is reducible.
\item $S \circ( [~]~N) \hastype (A \to B) \fsArrow C$ is reducible
  if $S$ and $N$ are reducible.
\item $S \circ ([~].\cl) \hastype (\ora{\cl:A}) \fsArrow C$ is
  reducible if $S$ is reducible.
\item $S \hastype \bagt{A} \fsArrow C$ is reducible if
  $\fsApp{S}{\ret{M}}$ is $c$-SN for all reducible $M \hastype A$.
\item $S \hastype \Bool \fsArrow C$ is reducible if $\fsApp{S}{\True}$
  is $c$-SN and $\fsApp{S}{\False}$ is $c$-SN.
\item $M:A$ is reducible if $\fsApp{S}{M}$ is $c$-SN for all reducible $S
\hastype A \fsArrow C$.
\end{compactitem}

\begin{lemma}
\label{lem:redterm-sn}
If $M \hastype A$ is reducible then $M$ is $c$-SN.
\end{lemma}
\begin{proof}
Follows immediately from reducibility of $\fsId$ and the definition of
reducibility on terms.
\end{proof}

\begin{lemma}
\label{lem:xred}
$x \hastype A$ is reducible.
\end{lemma}

\begin{proof}
By induction on $A$ using Lemma~\ref{lem:frame-reduction} and
Lemma~\ref{lem:redterm-sn}.
\end{proof}

\begin{corollary}
\label{lem:redfs-sn}
If $S \hastype A \fsArrow C$ is reducible then $S$ is $c$-SN.
\end{corollary}

Each type constructor has an associated $\beta$-rule.  Each
$\beta$-rule gives rise to an SN-closure property.

\begin{lemma}[SN-closure]
\label{lem:sn-closure}
~
\begin{sloppypar}
\begin{description}
\item[($\to$)] If $\fsApp{S}{\subst{M}{x}{N}}$ and $N$ are $c$-SN then
  $\fsApp{S}{(\lambda x.M)~N}$ is $c$-SN.
\item[($\record{}$)] If $\vec{M}$, are $c$-SN then
  $\fsApp{S}{\record{\ora{\cl=M}}.\cl_i}$ is $c$-SN.
\item[($\bagt{-}$)] If $\fsApp{S}{\subst{N}{x}{M}}$ and $M$ are $c$-SN
  then \\ $\fsApp{S}{\For\,(x \gets \ret{M})\,N}$ is $c$-SN.
\item[($\Bool$)] If $\fsApp{S}{N}$ and $\fsApp{S}{N'}$ are $c$-SN then
  \\
  $\fsApp{S}{\If\,\True\,\Then\,N\,\Else\,N'}$ is $c$-SN and \\
  $\fsApp{S}{\If\,\False\,\Then\,N\,\Else\,N'}$ is $c$-SN.
\end{description}
\end{sloppypar}
\end{lemma}

\begin{proof}
~
\begin{description}
\item[($\to$):] By induction on $\max_c(S)+\max_c(M)+\max_c(N)$.
\item[($\record{}$):] By induction on
  \[
  \max_c{S}+(\sum\iton \max_c(M_i))+\max_c(N)+(\sum\iton \max_c(M'_i)).
  \]
\item[($\bagt{}$):] By induction on
  \[\fsLength{S}+\max_c(\fsApp{S}{\subst{N}{x}{M}})+\max_c(M).\]
\item[($\Bool$):] By induction on
  $\fsLength{S}+\max_c(\fsApp{S}{N})+\max_c(\fsApp{S}{N'})$.
\end{description}
This completes the proof.
\end{proof}

Now we obtain reducibility-closure properties for each type
constructor.

\begin{lemma}[reducibility-closure]
\label{lem:red-closure}
~
\begin{description}
\item[($\to$)] If $\subst{M}{x}{N}$ is reducible for all reducible $N$,
  then $\lambda x.M$ is reducible.
\item[($\record{}$)] If $\vec{M}$ are reducible, then
  $\record{\ora{\cl=M}}$ is reducible.
\item[($\bagt{}$)] If $M$ is reducible, $\subst{N}{x}{M'}$ is reducible
  for all reducible $M'$, then $\For\,(x \gets M)\,N$ is
  reducible.
\item[($\Bool$)] If $M, N, N'$ are reducible then
  $\If\,M\,\Then\,N\,\Else,N'$ is reducible.
\end{description}
\end{lemma}
\begin{proof}
Each property follows from the corresponding part of
Lemma~\ref{lem:sn-closure} using
Lemma~\ref{lem:redterm-sn} and Corollary~\ref{lem:redfs-sn}.
\end{proof}
We also require additional closure properties for the empty bag and
union constructs.
\begin{lemma}[reducibility-closure II]
\label{lem:red-closure-monadplus}
~
\begin{description}
\item[($\emptybag{}$)] The empty bag $\emptybag{}$ is reducible.
\item[($\union$)] If $M, N$ are reducible, then $M \union N$ is reducible.
\end{description}
\end{lemma}
\begin{proof}
~
\begin{description}
\item[($\emptybag{}$):] Suppose $S : \bagt{A} \fsArrow C$ is reducible. We
  need to prove that $\fsApp{S}{\emptybag{}}$ is $c$-SN. The proof is by
  induction on $\fsLength{S} + \max_c(S)$. The only interesting case
  is hoisting the empty bag out of a bag elimination frame, which
  simply decreases the size of the frame stack by 1.
\item[($\union$):] Suppose $M,N : \bagt{A}$, and $S : \bagt{A} \fsArrow C$
  are reducible. We need to show that $\fsApp{S}{M \union N}$ is
  $c$-SN. The proof is by induction on $\fsLength{S} +
  \max_c(\fsApp{S}{M}) + \max_c(\fsApp{S}{N})$. The only interesting
  case is hoisting the union out of a bag elimination frame, which
  again decreases the size of the frame stack by 1, whilst leaving the
  other components of the induction measure unchanged.
\end{description}
This completes the proof.
\end{proof}

\begin{theorem}
\label{th:reducibility-fs}
Let $M$ be any term. Suppose $x_1 \hastype A_1,\dots,x_n \hastype A_n$
includes all the free variables of $M$.  If $N_1 \hastype
A_1,\dots,N_n \hastype A_n$ are reducible then
$\subst{M}{\vec{x}}{\vec{N}}$ is reducible.
\end{theorem}
\begin{proof}
By induction on the structure of terms using
Lemma~\ref{lem:red-closure} and Lemma~\ref{lem:red-closure-monadplus}.
\end{proof}

\begin{theorem}[strong normalisation]
\label{th:beta-sn}
The relation $\reducestoc$ is strongly normalising.
\end{theorem}
\begin{proof}
Let $M$ be a term with free variables $\vec{x}$. By
Lemma~\ref{lem:xred}, $\vec{x}$ are reducible. Hence, by
Theorem~\ref{th:reducibility-fs}, $M$ is $c$-SN.
\end{proof}

\begin{tronly}
  It is well known that $\beta$-reduction in simply-typed
  $\lambda$-calculus has non-elementary complexity in the worst
  case~\cite{beckmann01bounds}. The relation $\rewriteto_c$ includes
  $\beta$-reduction, so it must be at least as bad (we conjecture that
  it has the same asymptotic complexity, as $\rewriteto_c$ can be
  reduced to $\beta$-reduction on simply-typed $\lambda$-calculus via
  a CPS translation following de Groote~\cite{deGroote02}).
  However, we believe that the asymptotic complexity is unlikely to
  pose a problem in practice, as the kind of higher-order code that
  exhibits worst-case behaviour is rare. It has been our experience
  with Links that query normalisation time is almost always dominated
  by SQL execution time.
\end{tronly}

\subsection{If hoisting}

To hoist conditionals ($\If$-expressions) out of constant
applications, records, unions, and singleton bag constructors, we
define $\If$-hoisting frames as follows:
\begin{syntax}
  &  F[~]  &\cce& c(\vec{M},[~],\vec{N})
             \mid \record{\ora{\cl'=M},\cl=[~],\ora{\cl''=N}} \sepskip\\
  &        &\mid& [~] \union N \mid  M \union [~] \mid \ret{[~]} \\
\end{syntax}%
The $\If$-hoisting rule says that if an expression contains an
$\If$-hoisting frame around a conditional, then we can lift the
conditional up and push the frame into both branches:
\begin{equations}
F[\If\,L\,\Then\,M\,\Else\,N] &\rewriteto_h& \If\,L\,\Then\,F[M]\,\Else\,F[N] \\
\end{equations}%

We write $\size(M)$ for the size of $M$, as in the total number of
syntax constructors in $M$.

\begin{lemma}
\label{lem:snh}
~
\begin{compactenum}
\item\label{enum:snh-if} If $M, N, N'$ are $h$-SN then
  $\If\,M\,\Then\,N\,\Else\,N'$ is $h$-SN.
\item\label{enum:snh-union} If $M, N$ are $h$-SN then $M \union N$ is
  $h$-SN.
\item\label{enum:snh-const} If $\vec{M}$ are $h$-SN then $c(\vec{M})$
  is $h$-SN.
\item\label{enum:snh-record} If $\vec{M}$ are $h$-SN then
  $\record{\ora{\cl=M}}$ is $h$-SN.
\end{compactenum}
\end{lemma}

\begin{proof}
~
\begin{sloppypar}
\begin{Cases}
\item[\ref{enum:snh-if}:]
By induction on $\recordl\maxr_h(M)$, $\size(M)$, $\maxr_h(N) +
\maxr_h(N')$, $size(N) + size(N')\recordr$.
\item[\ref{enum:snh-union}:] By induction on $\recordl \maxr_h(M) +
  \maxr_h(N), \size_h(M) + \size_h(N) \recordr$ using
  (\ref{enum:snh-if}).
\item[\ref{enum:snh-const} and \ref{enum:snh-record}:] By induction on
  \[\recordl \sum\iton \maxr_h(M_i), \sum\iton \size(M_i) \recordr\]
 using
  (\ref{enum:snh-if}).
\end{Cases}
\end{sloppypar}
This concludes the proof.
\end{proof}

\begin{proposition}
\label{prop:if-hoist-sn}
  The relation $\rewriteto_h$ is strongly normalising.
\end{proposition}

\begin{proof}
  By induction on the structure of terms using Lemma~\ref{lem:snh}.
\end{proof}

\subsection{The query normalisation function}

The following definition of the function $\norm$ generalises the
normalisation algorithm from previous work~\cite{lindley12tldi}.

\begin{equations}
\norm_A(M) &=& \Split{\nf_h(\nf_c(M))}{A}
\\[2ex]
\multicolumn{3}{l}{\text{where:}}\\[2ex]
\Split{c([X_i:O_i]\iton)}{O} &=& c([\Split{X_i}{O_i}]\iton) \\
\Split{x.\cl}{O} &=& x.\cl \\
\Split{\isEmpty{(M : \bagt{A})}}{\Bool} &=& \isEmpty{(\Split{M}{\bagt{A}})} \\
\Split{M}{\record{\cl_i:A_i}\iton}
  &=& \record{\cl_i=\splitfield{M}{A_i,{\cl_i}}}\iton \\
\Split{M}{\bagt{A}}
  &=& \bigunion\,(\splitbag{M}{A,\nil,\True})
\\[2ex]
\splitbag{\ret{M}}{A,\vec{G},L}
  &=& [ \For\,(\vec{G}\,\Where\,L)\,\ret{\Split{M}{A}} ] \\
\splitbag{\For\,(x \gets t)\,M}{A,\vec{G},L}
  &=& \splitbag{M}{A,\vec{G} \append [x \gets t],L} \\
\splitbag{\Table\,t}{A,\vec{G},L}
  &=& \splitbag{\ret{x}}{A,\vec{G}\append [x \gets t],L} \\
&&\quad\text{($x$ fresh)}\\
\splitbag{\emptybag}{A,\vec{G},L} &=& \nil \\
\splitbag{M \union N}{A,\vec{G},L}
  &=& \splitbag{M}{A,\vec{G},L} \append \splitbag{N}{A,\vec{G},L} \\
\splitbag{\If\,L'\,M\,\,N}{A,\vec{G},L}
  &=& \splitbag{M}{A,\vec{G},L \wedge L'}  \\
   && \quad    \append~ \splitbag{N}{A,\vec{G},L \wedge \neg L'}
\\[2ex]
\splitfield{x}{A,\cl_i}
  &=& \Split{x.\cl_i}{A} \\
\splitfield{\record{\cl_i=M_i}\iton}{A,\cl_i}
  &=& \Split{M_i}{A} \\

\end{equations}%
Strictly speaking, in order for the above definition of $\norm_A$ to
make sense we need the two rewrite relations to be confluent. It is
easily verified that the relation $\rewriteto_c$ is locally confluent,
and hence by strong normalisation and Newman's Lemma it is
confluent. The relation $\rewriteto_h$ is not confluent as the
ordering of hoisting determines the final order in which booleans are
eliminated. However, it is easily seen to be confluent modulo
reordering of conditionals, which in turn means that $\norm_A$ is
well-defined if we identify terms modulo commutativity of conjunction,
which is perfectly reasonable given that conjunction is indeed
commutative.

\begin{theorem}
The function $\norm_A$ terminates.
\end{theorem}

\begin{proof}
The result follows immediately from strong normalisation of
$\rewriteto_c$ and $\rewriteto_h$, and the fact that the functions
$\Split{-}{}$, $\splitbag{-}{}$, and $\splitfield{-}{}$ are
structurally recursive (modulo expanding out the right-hand-side of
the definition of 
$\splitbag{\Table\,t}{A,\vec{G},L}$.
\end{proof}

\section{Proof of correctness of shredding}
\label{app:main-proof}
 The main correctness result is that if we shred an annotated
normalised nested query, run all the resulting shredded queries, and
stitch the shredded results back together, then that is the same as
running the nested query directly.


Intuitively, the idea is as follows: 

\begin{compactenum}
\item define shredding on nested values; 
\item show that shredding commutes
  with query execution; and 
\item show that stitching is a left-inverse to
  shredding of values.
\end{compactenum}
Unfortunately,
this strategy is a little too naive to work, because nested values do not contain
enough information about the structure of the source query in order to
generate the same indexes that are generated by shredded queries.

\begin{figure}[tb]
\[
\begin{tikzpicture}[font=\scriptsize]
\matrix(m)[matrix of math nodes, row sep=6em, column sep=2em, text
  height=1.5ex, text depth=0.25ex] {
         & & \sem{\Sigma}\\
\nsem{A} & &
\asem{A} & & & \pasem{A} \\
};
\path[->] (m-1-3) edge node[left,xshift=-2ex] {$\nsem{\erase(L)}$} (m-2-1);
\path[->] (m-1-3) edge node[left] {$\asem{L}$} (m-2-3);
\path[->] (m-1-3) edge node[right,xshift=1ex] {$\pasem{L}_A$} (m-2-6);
\path[->] (m-2-3) edge node[auto] {$\erase$} (m-2-1);
\path[->] (m-2-3) edge[bend left=15] node[above,xshift=-2.5ex] {$\shred_{(-)}(A)$} (m-2-6);
\path[->] (m-2-6) edge[bend left=15] node[auto] {$\buildname$} (m-2-3);
\end{tikzpicture}
\]

\caption{Correctness of shredding and stitching\label{fig:shred-stitch-pic}}

\end{figure}
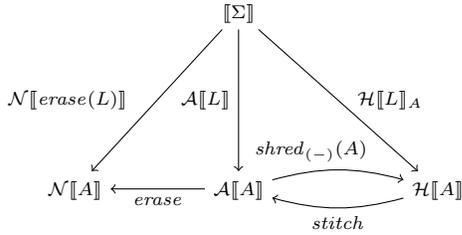

To fix the problem, we will give an annotated semantics $\asem{-}$
that is defined only on queries that have first been converted to
normal form and annotated with static indexes as described in
Section~\ref{sec:shredding}. We will ensure that $\nsem{\erase(L)} =
\erase(\asem{L})$, where we overload the $\erase$ function to erase
annotations on terms and values. The structure of the resulting proof
is depicted in the diagram in Figure~\ref{fig:shred-stitch-pic}, where
we view a query as a function from the interpretation of the input
schema $\Sigma$ to the interpretation of its result type. To prove
correctness is to prove that this diagram commutes.

\begin{figure*}[tb]
\begin{mathpar}
\inferrule[Result]
{[\vdash v_i : A_i]\iton \\ [\vdash J : \Index]\iton}
{\vdash [v_i\ann{J_i}]\iton : \bagt{A}}

\inferrule[Record]
{[\vdash v_i : A_i]\iton}
{\vdash \record{\cl_i=v_i}\iton : \record{\cl_i=A_i}\iton}

\inferrule[Constant]
{\Sigma(c) = O}
{\vdash c : O}
\end{mathpar}

\caption{Typing rules for nested annotated values
         \label{fig:indexed-nested-value-typing-rules}}

\[
\ba{@{}c@{\quad}c@{\quad}c@{}} \ba{r@{~}c@{~}l}
  \ssem{L} &=& \ssem{L}_{\varepsilon,1} \\
  \ea
&
  \ba{r@{~}c@{~}l}
  \ssem{\record{\cl=N}\iton}_{\rho,\iota} &=& \record{\cl_i=\ssem{N_i}_{\rho,\iota}}\iton \\
  \ssem{X}_{\rho,\iota} &=& \nsem{X}_\rho \\
  \ea
&
  \ba{r@{~}c@{~}l}
  \ssem{\Ind{a}{\iup}}_{\rho,\iota.i}        &=& \key(\Ind{a}{\iota}) \\
  \ssem{\Ind{a}{\idown}}_{\rho,\iota.i}      &=& \key(\Ind{a}{\iota.i}) \\
  \ea\\
\ea
\]
\[
\bl
\ssemcomp{\bigunion\iton C_i}_{\rho,\iota} = \Concat([\ssemcomp{C_i}_{\rho,\iota}]\iton) \hfill
\ssemcomp{\retann{N}{a}}_{\rho,\iota} = [\bsem{N}_{\rho,\iota} \ann{\key(\Ind{a}{\iota})}]
\smallskip\\
\ssemcomp{\For\,([x_i \gets t_i]\iton~\Where\,X)\,C}_{\rho,\iota} = 
\Concat([\ssemcomp{C}_{\rho[x_i \mapsto r_i]\iton,\iota.j} 
   \mid \tuple{j, \vec{r}} \gets 
      \enum([\vec{r} \mid [r_i \gets \tsem{t_i}]\iton, \nsem{X}_{\rho[x_i \mapsto r_i]\iton}])]) \\
\el
\]
\caption{Semantics of shredded queries, with annotations\label{fig:shredded-semantics-app}}
\end{figure*}

\subsection{Annotated semantics of nested queries}
\label{sec:annotated-semantics}
%

We annotate bag elements with distinct indexes as follows:
\begin{syntax}
\text{Results}           & s     &\cce& [\iv_1\ann{I_1}, \dots, \iv_m\ann{I_m}] \\
\text{Inner values}      & \iv   &\cce& c \mid r \mid s \\
\text{Rows}              & r     &\cce& \record{\cl_1=\iv_1, \dots, \cl_n=\iv_n} \\
\text{Indexes}           & I, J \\
\end{syntax}%
Typing rules for these values are in Figure~\ref{fig:indexed-nested-value-typing-rules}.
The index annotations $\ann{I}$ on collection elements are needed
solely for our correctness proof, and do not need to be present at run
time.

The \emph{canonical} choice for representing indexes is to take $I =
\Ind{a}{\iota}$. We also allow for alternative indexing schemes in
by parameterising the semantics
over an indexing function $\key$ mapping each canonical index
$\Ind{a}{\iota.j}$ to a concrete representation. In the simplest case,
we take $\key$ to be the identity function (where static indexes are
viewed as distinct integers). We will later consider other definitions
of $\key$ that depend on the particular query being shredded. Informally, the only constraints on
$\key$ are that it is defined on every canonical index
$\Ind{a}{\iota}$ that we might need in order to shred a query, and it
is injective on all canonical indexes. We will formalise these
constraints later, but for now,
$\key$ can be assumed to be the identity function.

Given an indexing function $\key$, the semantics of nested queries on
annotated values is defined as follows:
\begin{equations} 
\asem{L} &=& \asem{L}_{\varepsilon,\DUnit}
\sepskip\\
\asem{\bigunion\iton\,C_i}_{\rho,\iota}
  &=& \Concat([\asemcomp{C_i}_{\rho,\iota}]\iton) \\
\asem{\record{\cl_i=M_i}\iton}_{\rho,\iota}
  &=& \record{\cl_i=\asem{M_i}_{\rho,\iota}}\iton \\
\asem{X}_{\rho,\iota} &=& \nsem{X}_\rho
\sepskip\\
\multicolumn{3}{l}
{\bl
 \asemcomp{\For\,([x_i \gets t_i]\iton~\Where\,X)\,\retann{M}{a}}_{\rho,\iota} = \\
 \quad[\asem{M}_{\rho[x_i \mapsto r_i]\iton,\iota.j} \ann{\key(\Ind{a}{\iota.j})} \\
 \quad\quad\mid \tuple{j, \vec{r}} \gets 
      \enum([\vec{r} \mid [r_i \gets \tsem{t_i}]\iton, \nsem{X}_{\rho[x_i \mapsto r_i]\iton}])] \\
 \el}
\end{equations}%
As well as an environment, the current dynamic index is threaded
through the semantics.  Note that we use the ordinary semantics to
evaluate values that cannot have any annotations (such as boolean
tests in the last case).

It is straightforward to show by induction that the annotated
semantics erases to the standard semantics:
\begin{theorem}\label{th:ann-semantics-erases}
  For any $L$ we have $\erase(\asem{L}) = \nsem{\erase(L)}$.
\end{theorem}

\begin{figure*}[tb]
\begin{mathpar}
\inferrule[Result]
{[\vdash I_i : \Index]\iton \\ [\vdash \iv_i : F]\iton \\ [\vdash J_i : \Index]\iton}
{\vdash \bagv{\record{I_1,\iv_1\ann{J_1}}, \dots, \record{I_n,\iv_n\ann{J_n}}} : \bagt{\record{\Index,F}}}

\inferrule[Constant]
{\Sigma(c) = O}
{\vdash c : O}

\inferrule[Record]
{[\vdash n_i : F_i]\iton}
{\vdash \record{\cl_i=\iv_i}\iton : \record{\cl_i=F_i}\iton}
\end{mathpar}
\caption{Typing rules for shredded values\label{fig:shredded-value-typing-rules}}
\end{figure*}

\subsection{Shredding and stitching annotated values}

To define the semantics of shredded queries and packages, we use
annotated values in which collections are annotated pairs of indexes
and annotated values. Again, we allow annotations $\ann{J}$ on
elements of collections to facilitate the proof of correctness.
Typing rules for these values are shown in
Figure~\ref{fig:shredded-value-typing-rules}.
\begin{syntax}
\text{Results}         & s     &\cce& 
  [\tuple{I_1,\iv_1}\ann{J_1}, \dots, \tuple{I_m,\iv_m}\ann{J_m}]  \\
\text{Shredded values} & \iv   &\cce& c \mid r \mid I \\
\text{Rows}            & r     &\cce& \record{\cl_1=\iv_1, \dots, \cl_n=\iv_n} \\
\text{Indexes}         & I, J \\
\end{syntax}%
%

Having defined suitably annotated versions of nested and shredded
values, we now extend the shredding function to operate on nested
values.
\begin{equations}
\shouter{s}_p                              &=& \bagv{\shlist{s}_{\Unit,p}} 
\sepskip\\
\shlist{\bagv{[v_i \ann{J_i}]\iton}}_{I,\epsilon}
  &=& [\tuple{I,\shinner{v_i}_{J_i}} \ann{J_i}]\iton \\
\shlist{\bagv{[v_i \ann{J_i}]\iton}}_{I,\pbag.p} &=& \Concat([\shlist{v_i}_{J_i,p}]\iton) \\
\shlist{\record{\cl_i=v_i}\iton}_{I,\cl_i.p}   &=& \shlist{v_i}_{I,p} 
\sepskip\\
\shinner{c}_I                              &=& c \\
\shinner{\record{\cl_i=v_i}\iton}_I      &=& \record{\cl_i=\shinner{v_i}_I}\iton \\
\shinner{s}_I                              &=& I 
\end{equations}%
Note that in the cases for bags, the index annotation $J$ is passed as
an argument to the recursive call: this is where we need the ghost
indexes, to relate the semantics of nested and shredded queries.


We lift the nested result shredding function to build an (annotated)
shredded value package in the same way that we did for nested types
and nested queries.
\begin{equations}
\shred_s(\bagt{A}) = \package_{(\shouter{s}_{-})}(\bagt{A}) \\
\shred_{s,p}(\bagt{A}) = \package_{(\shouter{s}_{-}), p}(\bagt{A}) 
\end{equations}%




 We adjust the stitching function slightly to
maintain annotations; the adjustment is consistent with its behaviour
on unannotated values.
\begin{equations}
\buildtop{\sht{A}} &=&
  \build{\Ind{\top}{1}}{\sht{A}} 
\sepskip\\
\build{c}{O}                                              &=& c \\
\build{r}{\record{\cl_i:\sht{A_i}}\iton}  &=&
  \record{\cl_i=\build{r.cl_i}{\sht{A_i}}}\iton \\
\build{I}{\shann{\bagt{\sht{A}}}{s}} &=& 
    [(\build{\iv}{\sht{A}}) \ann{J}
             \mid \tuple{I, \iv} \ann{J} \gets s]
\end{equations}%
The inner value parameter $\iv$ to the auxiliary function
$\build{\iv}{-}$ specifies which values to stitch along the current
path.


\newcommand{\noann}[1]{{}}

We now show how to interpret shredded queries and query packages over
shredded values.  The semantics of shredded queries (including
annotations) is given in Figure~\ref{fig:shredded-semantics-app}.  The
semantics of a shredded query package is a \emph{shredded value
  package} containing indexed results for each shredded query.  For
each type $A$ we define $\psem{A} = \shred_A(A)$ and for each
flat--nested, closed $\vdash L : A$ we define $\psem{L}_A : \psem{A}$
as $\pmap_{\ssem{-}}~(\shred_L(A))$.

\subsection{Main results}

There are three key theorems. The first states that shredding commutes
with the annotated semantics.
\begin{theorem}
\label{th:sem-shred-commute}
If $\vdash L : \bagt{A}$ then:
\[
\psem{L}_{\bagt{A}} = \shred_{\asem{L}}(\bagt{A})
\]
\end{theorem}%


In order to allow shredded results to be correctly stitched together,
we need the indexes at the end of each path to a bag through a nested
value to be unique. We define $\indexes{p}{v}$, the indexes of nested
value $v$ along path $p$ as follows.
\begin{equations}
\indexes{\epsilon}{\bagv{[v_i\ann{J_i}]\iton}}     &=& [J_i]\iton \\
\indexes{\pbag.p}{\bagv{[v_i\ann{J_i}]\iton}} &=&
  \Concat([\indexes{p}{v_i}]\iton) \\
\indexes{\cl_i.p}{\record{\cl_i=v_i}\iton}       &=& \indexes{p}{v_i} \\
\end{equations}%

We say that a nested value $v$ is \emph{well-indexed (at type $A$)}
provided $\vdash v : A$, and for every path $p$ in $\paths(A)$, the
elements of $\indexes{p}{v}$ are distinct.  

\begin{lemma}
\label{lem:uniqueness}
If $\vdash L : A$, then $\asem{L}$ is well-indexed at $A$.
%
\end{lemma}

Our next theorem states that for well-indexed values, stitching is a
left-inverse of shredding.
\begin{theorem}
\label{th:shred-build-is-id}
\begin{sloppypar}
If $\vdash s : \bagt{A}$ and $s$ is well-indexed at type $\bagt{A}$ then
$\buildtop{\shred_s(\bagt{A})} = s$.
\end{sloppypar}
\end{theorem}%
Combining Theorem~\ref{th:sem-shred-commute},
Lemma~\ref{lem:uniqueness} and Theorem~\ref{th:shred-build-is-id} we
obtain the main correctness result (see
Figure~\ref{fig:shred-stitch-pic}).
\begin{theorem}
\label{th:shred-stitch-app}
If $\vdash L : \bagt{A}$ then:
$
 \buildtop{\psem{L}_{\bagt{A}}} = \asem{L}
$, 
and in particular:
$
 \erase(\buildtop{\psem{L}_{\bagt{A}}}) = \nsem{L}
$.
\end{theorem}

Not all possible indexing schemes satisfy Lemma~\ref{lem:uniqueness}.
To identify those that do, we recall the function for computing
the canonical indexes of a nested query.
\[\small
\begin{array}{rcl}
\isem{L} &=& \isem{L}_{\varepsilon,\DUnit}
\sepskip\\
\isem{\bigunion\iton\,C_i}_{\rho,\iota}
  &=& \Concat([\isem{C_i}_{\rho,\iota}]\iton) \\
\isem{\record{\cl_i=M_i}\iton}_{\rho,\iota}
  &=& \Concat([\isem{M_i}_{\rho,\iota}]\iton) \\
\isem{X}_{\rho,\iota} &=& \nil
\sepskip\\
\multicolumn{3}{l}
{\bl
 \isem{\For\,([x_i \gets t_i]\iton~\Where\,X)\,\retann{M}{a}}_{\rho,\iota} = \\
 \quad\Concat([\Ind{a}{\iota.j} \cons \isem{M}_{\rho[x_i \mapsto r_i]\iton,\iota.j} \\
 \quad\quad\mid \tuple{j, \vec{r}} \gets 
      \enum([\vec{r} \mid [r_i \gets \tsem{t_i}]\iton, \nsem{X}_{\rho[x_i \mapsto r_i]\iton}])]) \\
 \el}
\end{array}
\]
Note that $\isem{-}$ resembles $\asem{-}$,
but instead of the nested value $v$ it computes all indexes of $v$.
An indexing function $\key : \Index \to A$ is
\emph{valid} with respect to the closed nested query $L$ if it is
injective and defined on every canonical index in $\isem{L}$.
%
\begin{lemma}
If $\key$ is valid for $L$ then $\asem{L}$ is well-indexed.
\end{lemma}
The only requirement on indexes in the proof of
Theorem~\ref{th:shred-stitch} is that nested values be well-indexed,
hence the proof extends to any valid indexing scheme.  The concrete,
natural, and flat indexing schemes are all valid.

\subsection{Detailed proofs}

We first show that the inner shredding
function $\shinner{-}$ commutes with the annotated semantics.
\begin{lemma}
\label{lem:sem-shinner-commute}
$
\ssem{\shinner{M}_a}_{\rho,\iota} =
  \shinner{\asem{M}_{\rho,\iota}}_{\Ind{a}{\iota}}
$
\end{lemma}

\begin{proof}
By induction on the structure of $M$.

\begin{Cases}


\item[Case $x.\cl$:]
\begin{derivation}
 & \ssem{\shinner{x.\cl}_a}_{\rho,\iota} \\
=& \\
 & \ssem{x.\cl}_{\rho,\iota} \\
=& \\
 & \asem{x.\cl}_{\rho,\iota} \\
=& \\
 & \shinner{\asem{x.\cl}_{\rho,\iota}}_{\Ind{a}{\iota}} \\
\end{derivation}

\newcommand{\sqbrack}[1]{[{#1}]}

\item[Case $c(\sqbrack{X_i}\iton)$:]
\begin{derivation}
 & \ssem{\shinner{c([X_i]\iton)}_a}_{\rho,\iota} \\
=& \\
 & \ssem{c([\shinner{X_i}_a]\iton)}_{\rho,\iota} \\
=& \\
 & \asem{c([\shinner{X_i}_a]\iton)}_{\rho,\iota} \\
=& \\
 & \sem{c}([\asem{\shinner{X_i}_a}_{\rho,\iota}]\iton) \\
=& \\
 & \shinner{\asem{c([X_i]\iton)}_{\rho,\iota}}_{\Ind{a}{\iota}} \\
\end{derivation}

\item[Case $\isEmpty{M}$:]
\begin{derivation}
 & \ssem{\shinner{\isEmpty{M}}_a}_{\rho,\iota} \\
=& \\
 & \ssem{\isEmpty{\shouter{M}_\epsilon}}_{\rho,\iota} \\
=& \\
 & \asem{\isEmpty{\shouter{M}_\epsilon}}_{\rho,\iota} \\
=& \\
 & \shinner{\asem{\isEmpty{M}}_{\rho,\iota}}_{\Ind{a}{\iota}} \\
\end{derivation}

\item[Case $\record{\ora{\cl=M}}$:]

\begin{derivation}
 & \ssem{\shinner{\record{\cl_i=M_i}\iton}_a}_{\rho,\iota} \\
=& \\
 & \ssem{\record{\cl_i=\shinner{M_i}_a}\iton}_{\rho,\iota} \\
=& \\
 & \record{\cl_i=\ssem{\shinner{M_i}_a}_{\rho,\iota}}\iton \\
=& \byy{\text{Induction hypothesis}} \\
 & \record{\cl_i=\shinner{\asem{M_i}_{\rho,\iota}}_{\Ind{a}{\iota}}}\iton \\
=& \\
 & \shinner{\record{\cl_i=\asem{M_i}_{\rho,\iota}}\iton}_{\Ind{a}{\iota}} \\
=& \\
 & \shinner{\asem{\record{\cl_i=M_i}\iton}_{\rho,\iota}}_{\Ind{a}{\iota}} \\
\end{derivation}

\item[Case $L$:]

\begin{derivation}
 & \ssem{\shinner{L}_a}_{\rho,\iota} \\
=& \\
 & \ssem{\Ind{a}{\idown}}_{\rho,\iota} \\
=& \\
 & \Ind{a}{\iota} \\
=& \\
 &\shinner{\asem{L}_{\rho,\iota}}_{\Ind{a}{\iota}} 
\end{derivation}
\end{Cases}
This completes the proof.
\end{proof}
The first part of the following lemma allows us to run a shredded
query by concatenating the results of running the shreddings of its
component comprehensions.
Similarly, the second part allows us to shred the results of running a
nested query by concatenating the shreddings of the results of running
its component comprehensions.
\begin{lemma}

\label{lem:concat-lift}
~
\begin{compactenum}
\item\label{enum:concat-lift-sem-shouter}
   $\begin{array}[t]{l}
     \ssemcomp{\bigunion\,\shlist{\bigunion\iton\,C_i}_{a,p}}_{\rho,\iota} =
       \Concat([\ssem{\shlist{C_i}_{a,p}}_{\rho,\iota}]\iton)
     \end{array}$
\item\label{enum:concat-lift-shouter-sem}
   $\begin{array}[t]{l}
    \shlist{\asem{\bigunion\iton\,C_i}_{\rho,\iota}}_{\Ind{a}{\iota},p} =
      \Concat([\shlist{\bagv{\asemcomp{C_i}_{\rho,\iota}}}_{\Ind{a}{\iota},p}]\iton) \\
    \end{array}$
\end{compactenum}

\end{lemma}
\begin{proof}
~
\begin{compactenum}
\item
\begin{derivation}
 & \ssemcomp{\bigunion\,\shlist{\bigunion\iton\,C_i}_{a,p}}_{\rho,\iota} \\
=&  \\
 & \ssemcomp{\bigunion\,\Concat([\shlist{C_i}_{a,p}]\iton)}_{\rho,\iota} \\
=& \\
 & \ssemcomp{\bigunion\,\Concat([\shlist{C}_{a,p} \mid C \gets \vec{C}])}_{\rho,\iota} \\
=& \\
 & \Concat(\Concat([\ssemcomp{C'}_{\rho,\iota} \mid C \gets \vec{C}, C' \gets \shlist{C}_{a,p}])) \\
=& \\
 & \Concat([\ssem{\bigunion\,\shlist{C}_{a,p}}_{\rho,\iota} \mid C \gets \vec{C}]) \\
=& \\
 & \Concat([\ssem{\bigunion\,\shlist{C_i}_{a,p}}_{\rho,\iota}]\iton)
\end{derivation}

\item
\begin{derivation}
 & \shlist{\asem{\bigunion\iton\,C_i}_{\rho,\iota}}_{\Ind{a}{\iota},p} \\
=& \\
 & \shlist{\bagv{\Concat([\asemcomp{C_i}_{\rho,\iota}]\iton)}}_{\Ind{a}{\iota},p} \\
=& \\
 & \shlist{\bagv{\Concat([\asemcomp{C}_{\rho,\iota} \mid C \gets \vec{C}])}}_{\Ind{a}{\iota},p} \\
=& \\
 & \Concat(\Concat([\shlist{\bagv{s}}_{\Ind{a}{\iota},p}
                    \mid C \gets \vec{C}, s \gets \asemcomp{C}_{\rho,\iota}])) \\
=& \\
 & \Concat([\shlist{\bagv{\asemcomp{C}_{\rho,\iota}}}_{\Ind{a}{\iota},p}
            \mid C \gets \vec{C}]) \\
=& \\
 & \Concat([\shlist{\bagv{\asemcomp{C_i}_{\rho,\iota}}}_{\Ind{a}{\iota},p}]\iton) \\
\end{derivation}
\end{compactenum}
This completes the proof.
\end{proof}


We are now in a position to prove that the outer shredding function
$\shouter{-}$ commutes with the semantics.


%
\begin{lemma}
\label{lem:sem-shouter-commute}
$ \ssem{\shouter{L}_p}_{\rho,\DUnit} =
\shouter{\asem{L}_{\rho,\DUnit}}_p $
\end{lemma}%

\begin{proof}
~
We prove the following:
\begin{compactenum}
\item\label{enum:sem-shouter-commute}
  $\ssem{\shouter{L}_p}_{\rho,\DUnit} = \shouter{\asem{L}_{\rho,\DUnit}}_p$
\item\label{enum:sem-shcomp-commute}
  $\ssemcomp{\bigunion\,\shlist{C}_{a,p}}_{\rho,\iota} =
  \shlist{\bagv{\asemcomp{C}_{\rho,\iota}}}_{\Ind{a}{\iota},p}$
\item\label{enum:sem-shlist-commute}
  $\ssemcomp{\shlist{M}_{a,p}}_{\rho,\iota} =
  \shlist{\asem{M}_{\rho,\iota}}_{\Ind{a}{\iota},p}$
\end{compactenum}
The first equation is the result we require. Observe that it follows
from (\ref{enum:sem-shcomp-commute}) and Lemma~\ref{lem:concat-lift}.
We now proceed to prove equations (\ref{enum:sem-shcomp-commute}) and
(\ref{enum:sem-shlist-commute}) by mutual induction on the structure
of $p$.

There are only two cases for (\ref{enum:sem-shcomp-commute}), as the
$\cl_i.p$ case cannot apply.

\begin{Cases}
\item[Case $\epsilon$:]

\begin{derivation}
 & \ssemcomp{\bigunion\,\shlist{\For\,(\ora{x \gets t}\,\Where\,X)\,\retann{M}{b}}_{a,\epsilon}}_{\rho,\iota} \\
=& \byy{\text{Definition of }\ssemcomp{-}} \\
 &\ssemcomp{\For\,(\ora{x \gets t}\,\Where\,X)\,
              \retann{\tuple{\Ind{a}{\iota}, \shinner{M}_b}}{b}}_{\rho,\iota} \\
=& \byy{\text{Definition of }\ssemcomp{-}} \\
 &[\tuple{\Ind{a}{\iota}, \ssem{\shinner{M}_b}_{\rho[\ora{x \mapsto v}],\iota.i}} \ann{\Ind{b}{\iota.i}} \\
 &\quad\mid \tuple{i, \vec{v}} \gets 
       \enum([\vec{v} \mid \ora{v \gets \tsem{t}}, \ssem{X}_{\rho[\ora{x \mapsto v}]}])] \\
=& \byy{\text{Lemma~\ref{lem:sem-shinner-commute}}} \\
 &[\tuple{\Ind{a}{\iota}, \shinner{\asem{M}_{\rho[\ora{x \mapsto v}],\iota.i}}_{\Ind{b}{\iota.i}}} \ann{\Ind{b}{\iota.i}} \\
 &\quad\mid \tuple{i, \vec{v}} \gets 
       \enum([\vec{v} \mid \ora{v \gets \tsem{t}}, \ssem{X}_{\rho[\ora{x \mapsto v}]}])] \\
=& \byy{\text{Definition of }\shlist{-}} \\
 &\shlistl\bagvl[\asem{M}_{\rho[\ora{x \mapsto v}],\iota.i} \ann{\Ind{b}{\iota.i}} \\
 & \quad\mid \tuple{i, \vec{v}} \gets
        \enum([\vec{v} \mid \ora{v \gets \tsem{t}}, \ssem{X}_{\rho[\ora{x \mapsto v}]}])]
    \bagvr\shlistr_{\Ind{a}{\iota},\epsilon} \\
=& \byy{\text{Definition of }\asemcomp{-}} \\
 &\shlist{\bagv{\asemcomp{\For\,(\ora{x \gets t}\,\Where\,X)\,\retann{M}{b}}_{\rho,\iota}}}_{\Ind{a}{\iota},\epsilon} \\
\end{derivation}

\item[Case $\pbag.\epsilon$:]

\begin{derivation}
 & \ssemcomp{\bigunion\,\shlist{\For\,(\ora{x \gets t}\,\Where\,X)\,\retann{M}{b}}_{a,\pbag.p}}_{\rho,\iota} \\
=& \byy{\text{Definition of }\shlist{-}} \\
 & \ssemcomp{\bigunion\,
     [\For\,(\ora{x \gets t}\,\Where\,X)\,C
      \mid C \gets \shlist{M}_{b,p}]}_{\rho,\iota} \\
=& \byy{\text{Definition of }\ssemcomp{-}} \\
 & \Concat(
   [\begin{array}[t]{@{}l}
    \Concat \\
    \quad([\ssem{C}_{\rho[\ora{x \mapsto v}],\iota.i} \\
    \quad ~\mid \tuple{i, \vec{v}} \gets 
                    \enum([\vec{v} \mid \ora{v \gets \tsem{t}}, \ssem{X}_{\rho[\ora{x \mapsto v}]}])]) \\
    \mid C \gets \shlist{M}_{b,p}]) \\
    \end{array} \\
=& \byy{\text{Definition of }\ssemcomp{-}} \\
 & \Concat \\
 & \quad ([\ssemcomp{\shlist{M}_{b,p}}_{\rho[\ora{x \mapsto v}],\iota.i} \\
 & \quad\quad \mid \tuple{i, \vec{v}} \gets 
                    \enum([\vec{v} \mid \ora{v \gets \tsem{t}}, \ssem{X}_{\rho[\ora{x \mapsto v}]}])]) \\
=& \byy{\text{Induction hypothesis (\ref{enum:sem-shlist-commute})}}\\
 & \Concat \\
 & \quad ([\shlist{\asem{M}_{\rho[\ora{x \mapsto v}],\iota.i}}_{b,p} \\
 & \quad\quad \mid \tuple{i, \vec{v}} \gets 
                    \enum([\vec{v} \mid \ora{v \gets \tsem{t}}, \ssem{X}_{\rho[\ora{x \mapsto v}]}])]) \\
=& \byy{\text{Definitions of }\asemcomp{-}\text{ and }\shlist{-}} \\
 & \shlist{\bagv{\asemcomp{\For\,(\ora{x \gets t}\,\Where\,X)\,\retann{M}{b}}_{\rho,\iota}}}_{\Ind{a}{\iota},\pbag.p} \\
\end{derivation}
\end{Cases}

There are three cases for (\ref{enum:sem-shlist-commute}).

\begin{Cases}
\item[Cases $\epsilon$ and $\pbag.\epsilon$:] follow from
  (\ref{enum:sem-shcomp-commute}) by applying the two parts of
  Lemma~\ref{lem:concat-lift} to the left and right-hand side
  respectively.

\item[Case $\cl_i.\epsilon$:]

\begin{derivation}
 & \ssemcomp{\shlist{\record{\ora{\cl=M}}}_{a,\cl_i.p}}_{\rho,\iota} \\
=& \byy{\text{Definition of }\shlist{-}} \\
 & \ssemcomp{\shlist{M_i}_{a,p}}_{\rho,\iota} \\
=& \byy{\text{Induction hypothesis (\ref{enum:sem-shlist-commute})}} \\
 & \shlist{\asem{M_i}_{\rho,\iota}}_{\Ind{a}{\iota},p} \\
=& \byy{\text{Definition of }\asem{-}} \\
 & \shlist{\asem{\record{\ora{\cl=M}}}_{\rho,\iota}}_{\Ind{a}{\iota},\cl_i.p} \\
\end{derivation}
\end{Cases}
This completes the proof.
\end{proof}


We now lift Lemma~\ref{lem:sem-shouter-commute} to shredded packages.


\begin{proof}[\proofheader of Theorem~\ref{th:sem-shred-commute}]
  We need to show that if $\vdash L : \bagt{A}$ then
\[
 \psem{L}_A = \shred_{\asem{L}}(\bagt{A})
 \]
This is straightforward by induction on $A$, using Lemma~\ref{lem:sem-shouter-commute}.
\end{proof}

We have proved that shredding commutes with the semantics.
It remains to show that stitching after shredding is the identity on
index-annotated nested results. We need two auxiliary notions: the
\emph{descendant} of a value at a path, and the \emph{indexes} at the
end of a path (which must be unique in order for stitching to work).

We define $\desc{v}_{p,\iv}$, the \emph{descendant} of a value $v$ at
path $p$ with respect to inner value $\iv$ as follows.
\begin{equations}
\desc{v}_{p,\iv}                                   &=& \desc{v}_{\Unit,p,\iv} 
\sepskip\\
\desc{v}_{J,p,c}                               &=& c \\
\desc{v}_{J,p,\record{\cl_i=\iv_i}\iton}             &=& \record{\cl_i=\desc{v}_{J,p.\cl_i,\iv_i}}\iton \\
\desclist{\bagv{s}}_{J,\epsilon,I} &=& \left\{
\begin{array}{ll}
 s,       &\quad\text{if }J = I \\
 {}\nil,  &\quad\text{if }J \neq I \\
\end{array}
\right.\\
\desclist{\bagv{[v_i \ann{J_i}]\iton}}_{J,\pbag.p,I} &=& \Concat([\desclist{v_i}_{J_i,p,I}]\iton) \\
\desclist{\record{\cl_i=v_i}\iton}_{J,\cl_i.p,I}   &=& \desclist{v_i}_{J,p,I} \\
\end{equations}%
Essentially, this extracts the part of $v$ that corresponds to $w$.
The inner value $\iv$ allows us to specify a particular descendant as
an index, or nested record of indexes; for uniformity it may also
contain constants.

We can now formulate the following crucial technical lemma, which
states that given the descendants of a result $v$ at path $p.\pbag$
and the shredded values of $v$ at path $p$ we can stitch them together
to form the descendants at path $p$.
\begin{lemma}[key lemma]
\label{lem:desc}
If $v$ is well-indexed and $\vdash \desclist{v}_{J,p,I} :
\bagt{A}$, then
\[
  \begin{array}[t]{@{}l}
 \desclist{v}_{J,p,I} =
    [\desc{v}_{J,p.\pbag,\iv}\ann{I_\idown}
       \mid \tuple{I_\iup,\iv}\ann{I_\idown} \gets \shlist{v}_{J,p}, I_\iup = I]
 \end{array}
\]
\end{lemma}

\begin{proof}
First we strengthen the induction hypothesis to account for
records. The generalised induction hypothesis is as follows.

If $v$ is well-indexed and $\vdash \desc{v}_{J',p.\vec{\cl},I} :
\bagt{A}$, then
\[
  \begin{array}[t]{@{}l}
  [\desc{v}_{J',p.\pbag.\vec{\cl},\iv}\ann{I_\idown} 
  \mid \tuple{I_\iup,\iv}\ann{I_\idown} \gets \shlist{v}_{J',p.\vec{\cl}}, I_\iup = I] \qquad \\
       \hfill = [v'.\vec{\cl}\ann{I_\idown} 
                \mid {v'}\ann{I_\idown} \gets \desclist{v}_{J',p,I}] \\
  \end{array}
\]

The proof now proceeds by induction on the structure of $A$ and
side-induction on the structure of $p$.

\begin{Cases}

\item[Case $O$:]
~
\begin{Cases}
\item[Subcase $\epsilon$:]

If $J' \neq I$ then both sides are empty lists. Suppose that
$J' = I$.

\begin{derivation}
 & [(\desc{s}_{J',\pbag.\vec{\cl},c}\ann{I_\idown} \\
 & \qquad \mid \tuple{I_\iup,c}\ann{I_\idown}
               \gets \shlist{s}_{J',\vec{\cl}}, I_\iup = I] \\
=& \byy{\text{Definition of }\desc{-}} \\
 & [c\ann{I_\idown}
    \mid \tuple{I_\iup,c}\ann{I_\idown}
         \gets \shlist{s}_{J',\vec{\cl}}, I_\iup = I] \\
=& \byy{J' = I} \\
 & s \\
=& \byy{\text{Definition of }\desclist{-}}\\
 & [v'.\vec{\cl}\ann{I_\idown}
    \mid {v'}\ann{I_\idown} \gets \desclist{s}_{J',\epsilon,I}] \\
\end{derivation}

\item[Subcase $\cl_i.p$:]

\begin{derivation}
 & [(\desc{\record{\ora{\cl=v}}}_{J',\cl_i.p.\pbag.\vec{\cl},c}\ann{I_\idown} \\
 & \qquad \mid \tuple{I_\iup,c\ann{I_\idown}}
         \gets \shlist{\record{\ora{\cl=v}}}_{J',\cl_i.p.\vec{\cl}}, I_\iup = I] \\
=& \byy{\text{Definition of }\desc{-}} \\
 & [c\ann{I_\idown}
    \mid \tuple{I_\iup,c\ann{I_\idown}}
         \gets \shlist{\record{\ora{\cl=v}}}_{J',\cl_i.p.\vec{\cl}}, I_\iup = I] \\
=& \byy{\text{Definition of }\shlist{-}} \\
 & [c\ann{I_\idown}
    \mid \tuple{I_\iup,c\ann{I_\idown}}
         \gets \shlist{v_i}_{J',p.\vec{\cl}}, I_\iup = I] \\
=& \byy{\text{Definition of }\desc{-}} \\
 & [(\desc{v_i}_{J',p.\pbag.\vec{\cl},c}\ann{I_\idown}
    \mid \tuple{I_\iup,c}\ann{I_\idown}
         \gets \shlist{v_i}_{J',p.\vec{\cl}}, I_\iup = I] \\
=& \byy{\text{Induction hypothesis}} \\
 & [v'.\vec{\cl}\ann{I_\idown}
    \mid {v'}\ann{I_\idown} \gets \desclist{v_i}_{J',p,I}] \\
=& \byy{\text{Definition of }\shlist{-}} \\
 & [v'.\vec{\cl}\ann{I_\idown}
    \mid {v'}\ann{I_\idown} \gets \desclist{\record{\ora{\cl=v}}}_{J',\cl_i.p,I}] \\
\end{derivation}

\item[Subcase $\pbag.p$:]

\begin{derivation}
 & [(\desc{\bagv{[v_i\ann{J_i}]\iton}}_{J',\pbag.p.\pbag.\vec{\cl},c})\ann{I_\idown} \\
 &  \qquad \mid \tuple{I_\iup,c}\ann{I_\idown}
         \gets \shlist{[v_i\ann{J_i}]\iton}_{J',\pbag.p.\vec{\cl}}, I_\iup = I] \\
=& \byy{\text{Definitions of }\desc{-}\text{ and }\shlist{-}}\\
 & \Concat(
    [c\ann{I_\idown}
    \mid \tuple{I_\iup,c\ann{I_\idown}}
         \gets \shlist{v_i}_{J_i,p.\vec{\cl}}, I_\iup = I]\iton) \\
=& \byy{\text{Induction hypothesis}} \\
 & \Concat([v'.\vec{\cl}\ann{I_\idown} \mid {v'}\ann{I_\idown} \gets \desclist{v_i}_{J_i,p}]\iton) \\
=& \byy{\text{Definition of }\desclist{-}} \\
 & [v'.\vec{\cl}\ann{I_\idown} \mid {v'}\ann{I_\idown} \gets \desclist{[v\ann{J}]\iton}_{J',p}] \\
\end{derivation}

\end{Cases}

\item[Case $\unit$:] The proof is the same as for base types with the
  constant $c$ replaced by $\unit$.

\item[Case $\record{\ora{\cl:A}}$ where $|\vec{\cl}| \geq q$:]

We rely on the functions $\zip_{\vec{\cl}}$, for transforming a record
of lists of equal length to a list of records, and
$\unzip_{\vec{\cl}}$, for transforming a list of records to a record
of lists of equal length.
In fact we require special versions of $\zip$ and $\unzip$ that handle
annotations, such that $\zip$ takes a record of lists of equal length
whose annotations must be in sync, and $\unzip$ returns such a record.
\begin{equations}
\zip_{\vec{\cl}}\,\record{\cl_i=\nil}\iton &=& \nil \\
\zip_{\vec{\cl}}\,\record{\cl_i=v_i\ann{J} \cons s_i}\iton &=&
  \record{\cl_i=v_i}\iton\ann{J} \cons \zip_{\vec{\cl}}\,\record{\cl_i=s_i}\iton 
\sepskip\\
\unzip_{\vec{\cl}}(s) &=& \record{\cl_i=[v.\cl_i\ann{J} \mid v\ann{J} \gets s]}\iton \\
\end{equations}
If $\vec{\cl}$ is a non-empty list of column labels then
$\zip_{\vec{\cl}}$ is the inverse of $\unzip_{\vec{\cl}}$.
\[
\zip_{\vec{\cl}}(\unzip_{\vec{\cl}}(s)) = s, \qquad \text{if }|\vec{\cl}| \geq 1
\]

\begin{derivation}
 & [\desc{v}_{J',p.\pbag.\vec{\cl'},\iv}\ann{I_\idown} \\
 & \qquad \mid \tuple{I_\iup,\iv}\ann{I_\idown}
               \gets \shlist{v}_{J',p.\vec{\cl'}}, I_\iup = I] \\
=& \byy{\text{Definition of }\desc{-}} \\
 & [\record{\cl_i=\desc{v}_{J',p.\pbag.\vec{\cl'}.\cl_i,\iv}}\iton\ann{I_\idown}\ \\
 & \qquad \mid \tuple{I_\iup,\iv\ann{I_\idown}}
               \gets \shlist{v}_{J',p.\vec{\cl'}.\cl}, I_\iup = I] \\
=& \byy{\text{Definition of }\zip} \\
 & \zip_{\vec{\cl}}\,
    \recordl\cl_i=
       [v_{\cl_i,\iv}\ann{I_\idown}
        \mid \tuple{I_\iup,\iv\ann{I_\idown}}
                    \gets \shlist{v}_{J',p.\vec{\cl'}.\cl_i},\quad~~ \\
 & \hfill I_\iup = I]\recordr\iton \\
 & \qquad\text{where }v_{\cl,\iv} \text{ stands for } \desc{v}_{J',p.\pbag.\vec{\cl'}.\cl,\iv} \\
=& \byy{\text{Induction hypothesis}}\\
 & \zip_{\vec{\cl}}\,
     \record{\cl_i=[v'.\vec{\cl'}.\cl_i\ann{I_\idown}
                       \mid v'\ann{I_\idown} \gets \desclist{v}_{J',p,I}]}\iton \\
=& \byy{\text{Definition of }\zip} \\
 & [v'.\vec{\cl'}\ann{I_\idown} \mid v'\ann{I_\idown} \gets \desclist{v}_{J',p,I}] \\
\end{derivation}

\item[Case $\bagt{\vec{A}}$:]

\begin{Cases}
\item[Subcase $\epsilon$:]

If $J' \neq I$ then both sides of the equation are equivalent to
the empty bag. Suppose $J' = I$.


\begin{derivation}
 & [\desc{[v_i\ann{J_i}]\iton}_{J',\pbag.\vec{\cl},I_\idown}\ann{I_\idown} \\
 & \qquad \mid \tuple{I_\iup,I_\idown}\ann{I_\idown}
               \gets \shlist{[v_i\ann{J_i}]\iton}_{J',\vec{\cl}}, I_\iup = I] \\
=& \byy{\text{Definition of }\desc{-}} \\
 & [\bagv{\Concat([\desclist{v_i}_{J_i,\vec{\cl},I_\idown}]\iton)}\ann{I_\idown} \\
 & \qquad \mid \tuple{I_\iup,I_\idown}\ann{I_\idown}
               \gets \shlist{[v_i\ann{J_i}]\iton}_{J',\vec{\cl}}, I_\iup = I] \\
=& \byy{\shlist{[v_i\ann{J_i}]\iton}_{J',\vec{\cl}} =
        \tuple{J',J_1}\ann{J_1}, \dots, \tuple{J',J_n}\ann{J_n}} \\
 & [\bagv{\Concat([\desclist{v_i}_{J_i,\vec{\cl},I_\idown}]\iton)}\ann{I_\idown}
         \mid I_\idown \gets \vec{J}, J' = I] \\
=& \byy{J' = I} \\
 & [\bagv{\Concat([\desclist{v_i}_{J_i,\vec{\cl},I_\idown}]\iton)}\ann{I_\idown}
         \mid I_\idown \gets \vec{J}] \\
=& \byy{\text{Definition of }\desclist{-}} \\
 & [\bagvl\Concat([\desclist{v_i}_{J_i,\vec{\cl},I_\idown} \\
 & \hfill  \mid {v_i}\ann{J_i} \gets [v_i\ann{J_i}]\iton])\bagvr\ann{I_\idown}
                                   \mid I_\idown \gets \vec{J}] \\
=& \byy{\text{Definition of }\desclist{-}} \\
 & [\bagv{\Concat([\desclist{v_i.\vec{\cl}}_{J_i,\epsilon,I_\idown}
              \mid {v_i}\ann{J_i} \gets [v_i\ann{J_i}]\iton])}\ann{I_\idown}\quad \\
 & \hfill                  \mid I_\idown \gets \vec{J}] \\
=& \byy{\text{Definition of }\desclist{-}} \\
 & [\bagvl\Concat([
   \bl
     v_i.\vec{\cl} \\
     \mid {v_i}\ann{J_i} \gets [v_i\ann{J_i}]\iton, J_i = I_\idown])\bagvr\ann{I_\idown}
   \mid I_\idown \gets \vec{J}] \\
   \el \\
=& \byy{\text{$v$ is well-indexed}} \\
 & [v_i.\vec{\cl}\ann{I_\idown}
   \mid {v_i}\ann{I_\idown} \gets [v_i\ann{J_i}]\iton] \\
=& \byy{\text{Definition of }\desclist{-}\text{ and }J' = I}\\
 & [v_i.\vec{\cl}\ann{I_\idown}
    \mid {v_i}\ann{I_\idown} \gets \desclist{\bagv{[v_i\ann{J_i}]\iton}}_{J',\epsilon,I}]
\end{derivation}

\item[Subcase $\cl_i.p$:]

\begin{derivation}
 & [\desc{\record{\cl_i=v_i}\iton}_{J',\cl_i.p.\pbag.\vec{\cl'},I_\idown}\ann{I_\idown} \\
 & \qquad \mid \tuple{I_\iup,I_\idown}\ann{I_\idown}
              \gets \shlist{\record{\cl_i=v_i}\iton}_{J',\cl_i.p.\vec{\cl'}}, I_\iup = I] \\
=& \byy{\text{Definitions of }\desc{-}\text{ and }\shlist{-}} \\
 & [\desc{v_i}_{J',p.\pbag.\vec{\cl'},I_\idown}\ann{I_\idown} \\
 & \qquad     \mid \tuple{I_\iup,I_\idown}\ann{I_\idown}
              \gets \shlist{v_i}_{J',p.\vec{\cl'}}, I_\iup = I] \\
=& \byy{\text{Induction hypothesis}} \\
 & [v'.\vec{\cl'}\ann{I_\idown}
    \mid {v'}\ann{I_\idown} \gets \desclist{v_i}_{J',p,I}] \\
=& \byy{\text{Definition of}\desclist{-}} \\
 & [v'.\vec{\cl'}\ann{I_\idown}
    \mid {v'}\ann{I_\idown} \gets \desclist{\record{\cl_i=v_i}\iton}_{J',\cl_i.p,I}] \\
\end{derivation}

\item[Subcase $\pbag.p$:]

\begin{derivation}
 & [\desc{[v_i\ann{J_i}]\iton}_{J',\pbag.p.\pbag.\vec{\cl},I_\idown}\ann{I_\idown} \\
 & \qquad \mid \tuple{I_\iup,I_\idown}\ann{I_\idown}
                     \gets \shlist{[v_i\ann{J_i}]\iton}_{I_\iup,\pbag.p.\vec{\cl}}, I_\iup = I] \\
=& \byy{\text{Definitions of }\desc{-}\text{ and }\shlist{-}} \\
 & [(\Concat([\desc{v_i}_{J_i,p.\pbag.\vec{\cl},I_\idown}]))\ann{I_\idown} \\
 & \qquad  \mid \tuple{I_\iup, I_\idown}\ann{I_\idown}
                \gets \Concat([\shlist{v_i}_{J_i,p.\vec{\cl}}]\iton), I_\iup = I] \\
=& \byy{\text{Expanding }\Concat([\shlist{-}]\iton)} \\
 & [(\Concat([\desc{v_i}_{J_i,p.\pbag.\vec{\cl},I_\idown}]))\ann{I_\idown} \\
 & \qquad \mid {v_i}\ann{J_i} \gets [v_i\ann{J_i}]\iton, \\
 & \qquad\quad  \tuple{I_\iup, I_\idown}\ann{I_\idown}
                \gets \shlist{v_i}_{J_i,p.\vec{\cl}}, I_\iup = I] \\
=& \byy{\text{$v$ is well-indexed}} \\
 & [\desc{v_i}_{J_i,p.\pbag.\vec{\cl},I_\idown}\ann{I_\idown} \\
 & \qquad \mid {v_i}\ann{J_i} \gets [v_i\ann{J_i}]\iton, \\
 & \qquad\quad   \tuple{I_\iup, I_\idown}\ann{I_\idown}
                 \gets \shlist{v_i}_{J_i,p.\vec{\cl}}, I_\iup = I] \\

=& \byy{\text{Comprehension $\to$ concatenation}} \\
 & \Concat \\
 & \quad ([\desc{v_i}_{J_i,p.\pbag.\vec{\cl},I_\idown}\ann{I_\idown} \\
 & \quad\qquad     \mid \tuple{I_\iup, I_\idown}\ann{I_\idown}
                     \gets \shlist{v_i}_{J_i,p.\vec{\cl}}, I_\iup = I]\iton) \\

=& \byy{\text{Induction hypothesis}} \\
 & \Concat(
     [v'.\vec{\cl}\ann{I_\idown} \mid {v'}\ann{I_\idown} \gets \desclist{v_i}_{J_i,\pbag.p,I}]\iton) \\
=& \byy{\text{Concatenation $\to$ comprehension}} \\
 & [v'.\vec{\cl}\ann{I_\idown}
    \mid {v'}\ann{I_\idown}
           \gets \desclist{[v_i\ann{J_i}]\iton}_{J',\pbag.p,I}] \\
\end{derivation}
\end{Cases}
\end{Cases}
This completes the proof.
\end{proof}

The proof of this lemma is the only part of the formalisation that
makes use of values being well-indexed.
Stitching shredded results together does not depend on any other
property of indexes, thus any representation of indexes that yields
unique indexes suffices.

\begin{theorem}
\label{th:shred-build-is-desc}
If $s$ is well-indexed and $\vdash \desc{s}_{p,\iv} : A$ then
$\build{\iv}{\shred_{s,p}(A)} = \desc{s}_{p,\iv} $.
\end{theorem}


\begin{proof}
By induction on the structure of $A$.

\begin{Cases}
\item[Case $O$:]

\begin{derivation}
 & \build{c}{\shred_{s,p}(O)} \\
=& \\
 & c \\
=& \\
 & \desc{s}_{p,c} \\
\end{derivation}

\item[Case $\record{\ora{\cl=A}}$:]

\begin{derivation}
 & \build{\record{\cl_i=\iv_i}\iton}{\shred_{s,p}(\record{\cl_i:A}\iton)} \\
=& \\
 & \record{\cl_i=\build{\iv_i}{\shred_{s,p.l}(A)}}\iton \\
=& \byy{\text{Induction hypothesis}} \\
 & \record{\cl_i=\desc{s}_{p.l,\iv_i}}\iton \\
=& \\
 & \desc{s}_{p,\record{\cl_i=\iv_i}\iton} \\
\end{derivation}

\item[Case $\bagt{A}$:]

\begin{derivation}
 & \build{I}{\shred_{\bagv{s},p}(\bagt{A})} \\
=& \\
 & \build{I}{(\shann{\bagv{\shred_{\bagv{s},p.\pbag}(A)}}{\shouter{\bagv{s}}_p})} \\
=& \\
 & \bagvl\bl
           [(\build{n}{\shred_{\bagv{s},p.\pbag}(A))\ann{I_\idown}} \\
           \qquad\mid \tuple{I_\iup,\iv}\ann{I_\idown} \gets \shlist{s}_p, I_\iup = I]\bagvr \\
         \el \\
=& \byy{\text{Induction hypothesis}} \\
 & \bagv{[\desc{\bagv{s}}_{p.\pbag,\iv}\ann{I_\idown}
         \mid \tuple{I_\iup,\iv}\ann{I_\idown} \gets \shlist{s}_p, I_\iup = I]} \\
=& \byy{\text{Lemma~\ref{lem:desc}}} \\
 & \desc{\bagv{s}}_{p,I} \\
\end{derivation}
\end{Cases}

This completes the proof.
\end{proof}


\begin{proof}[\proofheader of Theorem~\ref{th:shred-build-is-id}]
By Theorem~\ref{th:shred-build-is-desc}, setting $p = \epsilon$ and $w
= \Unit$.
\end{proof}

We now obtain the main result.


\begin{proof}[\proofheader of Theorem~\ref{th:shred-stitch}]
Recall the statement of  Theorem~\ref{th:shred-stitch}: we need to show that $\vdash L : \bagt{A}$ then:
\[ \buildtop{\psem{L}_{\bagt{A}}} =
\buildtop{\shred_{\asem{L}}(\bagt{A})} = 
 \asem{L}
\]
The first equation follows immediately from
Theorem~\ref{th:sem-shred-commute}, and the second from
Lemma~\ref{lem:uniqueness} and Theorem~\ref{th:shred-build-is-id}.

Furthermore, applying $\erase$ to both sides we have:
\[
 \erase(\buildtop{\psem{L}_{\bagt{A}}}) = \erase(\asem{L}) = \nsem{L}
\]
where the second step follows by Theorem~\ref{th:ann-semantics-erases}.
\end{proof}


\newpage

\section{Record flattening}
\label{app:record-flattening}

\paragraph*{Flat types}

For simplicity we extend base types to include the unit type
$\record{}$. This allows us to define an entirely syntax-directed
\emph{unflattening} translation from flat record values to nested
record values.
\begin{syntax}
\textrm{Types}       & A, B &\cce& \bagt{\record{\ora{\cl:O}}} \\
\textrm{Base types}  & O    &\cce& \Int \mid \Bool \mid \String \mid \record{} \\
\end{syntax}%
\begin{tronly}
  The Links implementation diverges slightly from the presentation
  here. Rather than treating the unit type as a base type, it relies
  on type information to construct units when unflattening values.
\end{tronly}
\paragraph*{Flat terms}

\begin{syntax}
\text{Query terms}         & L, M &\cce& \bigunion\, \vec{C} 
\sepskip\\
\text{Comprehensions}      & C    &\cce& \Let~q = S\,\In\,S' \\
\text{Subqueries}          & S    &\cce& \For\,(\vec{G}\,\Where\,X)\,\ret{R} \\
\text{Data sources}        & u    &\cce& t \mid q \\
\text{Generators}          & G    &\cce& x \gets u \\
\text{Inner terms}         & N    &\cce& X \mid \cindex \\
\text{Record terms}        & R    &\cce& \record{\ora{\cl=N}} \\
\text{Base terms}          & X    &\cce& x.\cl \mid c(\vec{X}) \mid \isEmpty{L} \\
\end{syntax}

\paragraph*{Flattening types}

\newcommand{\flatten}[1]{{#1}^\succ} 
\newcommand{\unflatten}[1]{{#1}^\prec}

\newcommand{\flattenProj}[1]{{#1}^\dagger}

\newcommand{\const}{\bullet}

The record flattening function $\flatten{(-)}$ flattens record types.
\begin{equations}
\flatten{(\bagt{A})}           &=& \bagt{(\flatten{A})} 
\sepskip\\
\flatten{O}                    &=& \record{ \const:O } \\
\flatten{\record{}}            &=& \record{ \const:\record{} } \\
\flatten{\record{\ora{\cl:F}}} &=&
  \record{[(\consl{\cl_i}\cl'):O
  \mid i \gets \dom(\vec{\cl}), (\cl':O) \gets \flatten{F_i} ]} \\
\end{equations}%
Labels in nested records are concatenated with the labels of their
ancestors. Base and unit types are lifted to 1-ary records (with a
special $\const$ field) for uniformity and to aid with reconstruction
of nested values from flattened values.

\paragraph*{Flattening terms}

The record flattening function $\flatten{(-)}$ is defined on terms as
well as types.
\begin{equations}
\flatten{(\bigunion\iton\,C_i)}    &=& \bigunion\iton\,\flatten{(C_i)}
\sepskip\\
\flatten{(\Let~q = S_\iup~\In\,S_\idown)} &=&
  \Let~q = \flatten{S_\iup}~\In\,\flatten{S_\idown} \\
\flatten{(\For\,(\vec{G}\,\Where\,X)~\ret{N})} &=&
  \For\,(\vec{G}\,\Where\,\flattenProj{X})~\ret{\flatten{N}} \\
\end{equations}%
\begin{equations}
\flatten{X}                           &=& \record{\const = \flattenProj{X}} \\
\flatten{\record{}}                   &=& \record{\const = \record{}} \\
\flatten{(\record{\cl_i=N_i}_{i=1}^m)} &=& 
  \record{\consl{\cl_i}\cl'_j=X_j}_{(i=1,j=1)}^{(m,n_i)}, \\
&&\quad \text{where } \flatten{N_i} = \record{\cl'_j=X_j}_{j=1}^{n_i} \\
\end{equations}%
\begin{equations}
\flattenProj{(x.\cl_1. \cdots .\cl_n)} &=&
  x.\cl_1\delim \dots \delim\cl_n\delim\const \\
\flattenProj{c(X_1, \dots, X_n)} &=&
  c(\flattenProj{(X_1)}, \dots, \flattenProj{(X_n)}) \\
\flattenProj{(\isEmpty{L})} &=&
  \isEmpty{\flatten{L}} \\
\end{equations}%
The auxiliary $\flattenProj{(-)}$ function flattens $n$-ary
projections.

\paragraph*{Type soundness}

\[
\begin{array}{rcl}
{\vdash L : A} &\mathbin{\Rightarrow}& {\vdash \flatten{L} : \flatten{A}} \\ 
\end{array}
\]





\paragraph*{Unflattening record values}

\begin{equations}
\unflatten{[r_1, \dots, r_n]} &=& [\unflatten{(r_1)}, \dots, \unflatten{(r_n)}] \\
\sepskip\\
\unflatten{\record{\const=c}} &=& c \\
\unflatten{\record{\const=\record{}}} &=& \record{} \\
\unflatten{(\record{\consl{\cl_i}\cl'_j=c_j}_{(i=1,j=1)}^{(m,n_i)})} &=&
  \record{\cl_i=\unflatten{(\record{\cl'_j=c_j}_{j=1}^{n_i})}}_{i=1}^m \\
\end{equations}%

\paragraph*{Type soundness}

\[
\begin{array}{rcl}
{\vdash L : A} &\mathbin{\Rightarrow}& {\vdash \unflatten{L} : \unflatten{A}} \\ 
\end{array}
\]

\paragraph*{Correctness}

\begin{proposition}
If $L$ is a let-inserted query and $\vdash L : \bagt{A}$, then
\[
\unflatten{\bigl(\lsem{\flatten{L}}\bigr)} =
\lsem{L}
\]
\end{proposition}

\end{document}